\documentclass[sn-mathphys]{sn-jnl}

\jyear{2021}%

\theoremstyle{thmstyleone}%
\theoremstyle{thmstyletwo}%
\newtheorem{example}{Example}%

\theoremstyle{thmstylethree}%
\newtheorem{definition}{Definition}%

\usepackage{amsmath}
\usepackage{multicol}
\usepackage{lipsum}
\usepackage{mathtools}
\usepackage{mathpartir}
\usepackage[mathscr]{eucal}
\usepackage{stmaryrd}
\usepackage{relsize}

\usepackage{listings}

\usepackage{pdfpages}

\usepackage{cleveref}

\usepackage{multirow}

\usepackage{enumitem}

\usepackage{float}

\definecolor{dkgreen}{rgb}{0,0.6,0}
\definecolor{ltblue}{rgb}{0,0.4,0.4}
\definecolor{dkviolet}{rgb}{0.3,0,0.5}

\lstdefinelanguage{Coq}{ 
    mathescape=true,
    texcl=false, 
    morekeywords=[1]{Section, Module, End, Require, Import, Export,
        Variable, Variables, Parameter, Parameters, Axiom, Hypothesis,
        Hypotheses, Notation, Local, Tactic, Reserved, Scope, Open, Close,
        Bind, Delimit, Definition, Let, Ltac, Fixpoint, CoFixpoint, Add,
        Morphism, Relation, Implicit, Arguments, Unset, Contextual,
        Strict, Prenex, Implicits, Inductive, CoInductive, Record,
        Structure, Canonical, Coercion, Context, Class, Global, Instance,
        Program, Infix, Theorem, Lemma, Corollary, Proposition, Fact,
        Remark, Example, Proof, Goal, Save, Qed, Defined, Hint, Resolve,
        Rewrite, View, Search, Show, Print, Printing, All, Eval, Check,
        Projections, inside, outside, Def},
    morekeywords=[2]{forall, exists, exists2, fun, fix, cofix, struct,
        match, with, end, as, in, return, let, if, is, then, else, for, of,
        nosimpl, when},
    morekeywords=[3]{Type, Prop, Set, true, false, option},
    morekeywords=[4]{pose, set, move, case, elim, apply, clear, hnf,
        intro, intros, generalize, rename, pattern, after, destruct,
        induction, using, refine, inversion, injection, rewrite, congr,
        unlock, compute, ring, field, fourier, replace, fold, unfold,
        change, cutrewrite, simpl, have, suff, wlog, suffices, without,
        loss, nat_norm, assert, cut, trivial, revert, bool_congr, nat_congr,
        symmetry, transitivity, auto, split, left, right, autorewrite},
    morekeywords=[5]{by, done, exact, reflexivity, tauto, romega, omega,
        assumption, solve, contradiction, discriminate},
    morekeywords=[6]{do, last, first, try, idtac, repeat},
    morecomment=[s]{(*}{*)},
    showstringspaces=false,
    morestring=[b]",
    morestring=[d],
    tabsize=3,
    extendedchars=false,
    sensitive=true,
    breaklines=false,
    basicstyle=\small,
    captionpos=b,
    columns=[l]flexible,
    identifierstyle={\ttfamily\color{black}},
    keywordstyle=[1]{\ttfamily\color{dkviolet}},
    keywordstyle=[2]{\ttfamily\color{dkgreen}},
    keywordstyle=[3]{\ttfamily\color{ltblue}},
    keywordstyle=[4]{\ttfamily\color{dkblue}},
    keywordstyle=[5]{\ttfamily\color{dkred}},
    stringstyle=\ttfamily,
    commentstyle={\ttfamily\color{dkgreen}},
    literate=
    {\\forall}{{\color{dkgreen}{$\forall\;$}}}1
    {\\exists}{{$\exists\;$}}1
    {<-}{{$\leftarrow\;$}}1
    {=>}{{$\Rightarrow\;$}}1
    {==}{{\code{==}\;}}1
    {==>}{{\code{==>}\;}}1
    {->}{{$\rightarrow\;$}}1
    {<->}{{$\leftrightarrow\;$}}1
    {<==}{{$\leq\;$}}1
    {\#}{{$^\star$}}1 
    {\\o}{{$\circ\;$}}1 
    {\@}{{$\cdot$}}1 
    {\/\\}{{$\wedge\;$}}1
    {\\\/}{{$\vee\;$}}1
    {++}{{\code{++}}}1
    {~}{{\ }}1
    {\@\@}{{$@$}}1
    {\\mapsto}{{$\mapsto\;$}}1
    {\\hline}{{\rule{\linewidth}{0.5pt}}}1
}[keywords,comments,strings]

\makeatletter
\newcommand{\oset}[3][0ex]{%
  \mathrel{\mathop{#3}\limits^{
    \vbox to#1{\kern-0\ex@
    \hbox{$\scriptstyle#2$}\vss}}}}
\makeatother
\newcommand{\doublebrackets}[1]{[\![ #1 ]\!]}
\newcommand{\progspec}[1]{\{ #1 \}}
\newcommand{\tuple}[1]{\langle #1 \rangle}
\newcommand{\lsep}{ ~\vert~ }
\newcommand{\triple}[3]{\progspec{ #1 } #2 \progspec{ #3 }}

\DeclareRobustCommand{\VDash}{\mathrel{\|}\joinrel\Relbar}
\DeclareRobustCommand{\fpupdate}{\oset{\cdot}{\mathrel{\vert}\joinrel\mathrel{\Rrightarrow}}}

\newcommand{\circnum}[1]{{\bigcirc\mkern-12.5mu{\footnotesize\text{#1}}\mkern6.25mu}}

\newcommand{\btrue}{\texttt{true}}
\newcommand{\bfalse}{\texttt{false}}
\newcommand{\cmdref}{\texttt{ref}}
\newcommand{\cmdfork}[1]{\texttt{fork}\{#1\}}
\newcommand{\cmdloop}[1]{\texttt{loop}_{#1}\,}
\newcommand{\cmdbreak}{\texttt{break}\,}
\newcommand{\cmdcontinue}{\texttt{continue}}
\newcommand{\cmdif}{\texttt{if}\,}
\newcommand{\cmdthen}{\,\texttt{then}\,}
\newcommand{\cmdelse}{\,\texttt{else}\,}
\newcommand{\cmdseq}{\,;\mkern-8mu;\mkern-4mu}
\newcommand{\cmdreturn}{\texttt{return}\,}
\newcommand{\cmdcall}{\texttt{call}\,}

\newcommand{\hred}{\rightarrow_h}
\newcommand{\tred}{\rightarrow_t}

\newcommand{\cred}{\text{red}}

\newcommand{\emp}{\texttt{emp}}

\newcommand{\WPRE}{\textsf{WP}}
\newcommand{\wand}{-\mkern-8mu*\,}
\newcommand{\upd}{\dot{\vert\mkern-8mu\Rrightarrow}\,}
\newcommand{\later}{\triangleright}
\newcommand{\wpre}[5]{\textsf{WP}\, #1\, \progspec{#2, [#3, #4, #5]}}

\newcommand{\cseq}{\,;\mkern-8mu;\mkern-4mu}
\newcommand{\cif}[3]{\ensuremath{\mathtt{if}~#1~\mathtt{then}~#2~\mathtt{else}~#3}}
\newcommand{\cfor}[2]{\ensuremath{\mathtt{for}(;\mkern-8mu;\mkern-4mu#2)~#1}}
\newcommand{\cwhile}[2]{\ensuremath{\mathtt{while}(#1)~#2}}
\newcommand{\cbreak}{\ensuremath{\mathtt{break}}}
\newcommand{\ccontinue}{\ensuremath{\mathtt{continue}}}

\newcommand{\cskip}{\ensuremath{\mathtt{skip}}}

\newcommand{\bigvalid}{\vDash_\text{b}}
\newcommand{\exitkind}{\text{exit\_kind}}
\newcommand{\ek}{\text{ek}}
\newcommand{\ekb}{\text{brk}}
\newcommand{\ekc}{\text{con}}
\newcommand{\ekr}{\text{ret}}

\newcommand{\guard}[3]{\progspec{#1} (#2, #3)}

\newcommand{\kloops}[3]{\text{KLoop}_{#3}(#1, #2)}
\newcommand{\kseq}[1]{\text{KSeq}(#1)}

\newtheorem{thm}{Theorem}
 
\newtheorem{lem}[thm]{Lemma}
\newtheorem{hypo}[thm]{Hypothesis}

\begin{document}

\title[Deep Embedding v.s Shallow Embedding]{Verifying Programs with Logic and Extended Proof Rules: Deep Embedding v.s. Shallow Embedding}


\author*[1]{Zhongye Wang}\email{wangzhongye1110@sjtu.edu.cn}
\author[1]{Qinxiang Cao}\email{caoqinxiang@gmail.com}
\author[1]{Yichen Tao}\email{taoyc0904@sjtu.edu.cn}


\affil[1]{\orgname{Shanghai Jiao Tong University}, \orgaddress{Shanghai, China}}


\abstract{
  Many foundational program verification tools have been developed to build machine-checked program correctness proofs, a majority of which are based on Hoare logic.
  Their program logics, their assertion languages, and their underlying programming languages can be formalized by either a shallow embedding or a deep embedding.
  Tools like Iris and early versions of Verified Software Toolchain (VST) choose different shallow embeddings to formalize their program logics.
  But the pros and cons of these different embeddings were not yet well studied.
  Therefore, we want to study the impact of the program logic's embedding on logic's proof rules in this paper.

  This paper considers a set of useful extended proof rules, and four different logic embeddings: one deep embedding and three common shallow embeddings.
  We prove the validity of these extended rules under these embeddings and discuss their main challenges.
  Furthermore, we propose a method to lift existing shallowly embedded logics to deeply embedded ones to greatly simplify proofs of extended rules in specific proof systems.
  We evaluate our results on two existing verification tools.
  We lift the originally shallowly embedded VST to our deeply embedded VST to support extended rules, and we implement Iris-CF and deeply embedded Iris-Imp based on the Iris framework to evaluate our theory in real verification projects.
}




\maketitle

\section*{Statements and Declarations}
\subsection*{Funding}
This research is sponsored by National Natural Science Foundation of China (NSFC) Grant No. 61902240.
\subsection*{Competing interests}
\begin{itemize}
  \item The authors have no relevant financial or non-financial interests other than those declared in the Funding section.
  \item The authors have no competing interests to declare that are relevant to the content of this article.
\end{itemize}
\newpage

\section{Introduction}
\label{sec:intro}


Computer scientists have gained great success in developing foundationally sound program verification systems in the past decade.
The word \emph{foundationally} means: not only are programs' correctness properties verified, but the correctness proofs and the depended proof rules are also verified and machine-checked in a proof assistant like Coq \cite{COQ} or Isabelle \cite{nipkow2002isabelle, paulson1994isabelle}.
Foundational verification tools like Verified Software Toolchain (VST) \cite{VST, VST-Floyd},
Iris \cite{jung2018iris}, CSimpl \cite{sanan2017csimpl},
CakeML \cite{CF-CakeML} etc. enable their users to verify programs written in real programming languages like C, Rust and OCaml using Hoare-style logic.

\paragraph{Extended Proof Rules.}
Theoretically, a Hoare logic with compositional rules, the consequence rule, and proof rules for singleton commands (like assignments) is powerful enough to prove all valid judgement \cite{cook1978soundness}.
These rules are often referred to as primary rules.
Figure~\ref{fig:example}.A shows a sound and complete logic for a toy language with only skip command, assignment, sequential composition, and if-statement.

\begin{figure}[h]
\textbf{A. Primary Rules.}\\
\vspace{-1em}
\begin{mathpar}
    \inferrule[hoare-skip]{}{
        \vdash \triple{P}{\cskip}{P}
    }
    \and
    \inferrule[hoare-assign]{}{
        \vdash \triple{P[x\mapsto e]}{x=e}{P}
    }
    \and
    \inferrule[hoare-seq]{
        \vdash \triple{P}{c_1}{R}
        \and
        \vdash \triple{R}{c_2}{Q}
    }{
        \vdash \triple{P}{c_1 \cseq c_2}{Q}
    }
    \and
    \inferrule[hoare-if]{
        \vdash \triple{P \land \doublebrackets{e}=\text{true}}{c_1}{Q} \\\\
        \vdash \triple{P \land \doublebrackets{e}=\text{false}}{c_2}{Q}
    }{
        \vdash \triple{P}{\cif{e}{c_1}{c_2}}{Q}
    }
    \and
    \inferrule[hoare-conseq]{
        P\vdash P' \\     
        \vdash \triple{P'}{c}{Q'} \\
        Q' \vdash Q
    }{
        \vdash \triple{P}{c}{Q}
    }
\end{mathpar}

\vspace{1em}
\textbf{B. Extended Rules.}\\
\begin{minipage}{0.33\linewidth}
    \begin{mathpar}
        \inferrule[\underline{seq-assoc}]{
            \vdash \triple{P}{(c_1 \cseq c_2) \cseq c_3}{Q}
        }{
            \vdash \triple{P}{c_1 \cseq (c_2 \cseq c_3)}{Q}
        }
    \end{mathpar}
\end{minipage}
\begin{minipage}{0.66\linewidth}
    \begin{mathpar}
        \inferrule[\underline{if-seq}]{
                \vdash \triple{P}{\cif{e}{c_1\cseq c_3}{c_2\cseq c_3}}{Q}
        }{
            \vdash \triple{P}{(\cif{e}{c_1}{c_2}) \cseq c_3}{Q}
        }
    \end{mathpar}
\end{minipage}

\begin{minipage}{0.33\linewidth}
    \textbf{C. Derived Rules.}
\end{minipage}
\begin{minipage}{0.6\linewidth}
\ \\
\begin{mathpar}
        \inferrule*[left=hoare-conseq-pre]{
            P \vdash P' \\
            \vdash\triple{P'}{c}{Q}
        }{
            \vdash\triple{P}{c}{Q}
        }
    \end{mathpar}
\end{minipage}
\vspace{1em}
\label{fig:example}
\caption{A example Hoare logic for a toy language: primary rules and extended rules}
\end{figure}

But in practice \cite{VST-Floyd}, a verification tool can be much easier to use if more rules are added to aid our proofs.
Figure~\ref{fig:example}.B shows two extended proof rules for the toy language.
\textsc{seq-assoc} changes the associativity of sequential composition and allows provers to verify a program both from the beginning to the end in the forward style and from the end to the beginning in the backward style;
\textsc{if-seq} distributes the command $c_3$ after an if-block to two branches inside it, and in this way, instead of finding an intermediate assertion that merges post-conditions of two branches $c_1$ and $c_2$, provers can directly verify $c_1\cseq c_3$ and $c_2\cseq c_3$, with no obligation to find a common post-condition for two branches.
In this paper, we also incorporate control flow reasoning in program logics, which allows realistic program verification and this induces more extended rules with increased complexity, which are also covered in this paper.

We make a distinction between extended rules and derived rules.
Derived rules are directly derived from primary rules with mostly trivial proofs, and these proofs are the same across different embedding of a logic with same primary rules.
Figure~\ref{fig:example}.C shows an example of derived proof rules.
\textsc{hoare-conseq-pre} is easily derivable from \textsc{hoare-conseq} and it allows weakening the precondition but keep the postcondition untouched.
This paper focuses on extended rules, whose proofs are usually non-trivial and differs across different embeddings.
Both derived rules and extended rules are useful in practice and we admit both as necessary components of a verification tool.


\paragraph{Extended Rules \& Hoare Logic Embeddings.}
To utilize these extended rules in program verifications, developers should first prove them in their logics.
However, these extended rules may not be easily provable in every foundational verification tool, e.g., shallowly embedded VST does not provide \textsc{seq-inv} (in figure~\ref{fig:extended}) due to its tremendous proof burden while we can develop a simple formal proof of it in our deeply embedded VST.
This is because different verification tools may choose various ways, i.e., various embeddings, to formalize their program languages and program logics to obtain unique features.
For example, a deeply embedded program logic defines the logic by inductive proof trees, while a shallowly embedded one directly defines Hoare triple's validity using a program's semantics.

Due to differences in embedding methods, not all these extended rules are valid under each embedding, and challenges we might encounter in their proofs also vary from embedding to embedding.
Therefore, in this paper, we (1) \textbf{present four mainstream Hoare logic embeddings}\footnote{
    One deep embedding, and three shallow embedding respectively based on big-step semantics, weakest precondition, and continuation.
} based on our survey of Hoare-logic-based program verification projects including VST, Iris, FCSL, etc., (2) \textbf{identify a set of extended proof rules} that can benefit verification automation; and (3) \textbf{formally verify these extended proof rules under these different embeddings}.

\paragraph{Another Contribution: A Lifting Approach.}

We summarize main challenges in proofs of these extended rules and find that most proofs of extended proof rules in shallowly embedded program logics is more complicated than those in the deeply embedded one.
It would take tremendous proof effort to equip verification tools using shallow embeddings with these extended proof rules.
\textbf{Another contribution of our paper is to present a much easier way to extended shallowly embedded program logic with extended proof rules:}
we can first lift the shallowly embedded logic into a deeply embedded one with acceptable proof effort and then use simpler proofs of extended rules under the deep embedding.

We may not be the first one to discover and utilize such approach, but to our knowledge, it has never been formally presented in literatures and we believe it can help optimize existing foundational verification tools.
To evaluate the lifting approach, we apply it to two existing verification projects.
\begin{itemize}
    \item[$\triangleright$] We lift the shallowly embedded VST \cite{VST}, a verification framework for C programs, to the deeply embedded VST, which supports most extended rules and is the basis of the VST currently in the industrial use \cite{VST}.
    \item[$\triangleright$] We also instantiate it on the Iris \cite{jung2018iris} framework.
    We first extend the shallowly embedded Iris framework to support control flow reasoning in our Iris-CF.
    Then we encode a deeply embedded logic, Iris-Imp, above it using the lifting approach.
\end{itemize}

Our survey and formalization techniques can benefit future verification tools in following ways.
Firstly, our surveys and proofs indicate the level of support for extended proof rules under differently embedded program logics.
Knowing which extended proof rules a verification tool supports, verification tool users may decide in advance which tool best matches their proof goals.
And by understanding advantages and limitations of different embedding methods, verification tool developers can design proof systems satisfying the growing industrial demands.
Secondly, our lifting approach can help existing verification tools to improve their framework by implementing extended rules with alleviated proof burden, as we demonstrate in our VST and Iris-Imp projects.

\paragraph{Paper Structure.}

Firstly, in section~\ref{sec:rules}, we fix primitives of the programming language and primary proof rules of the program logic that this paper focuses on.
Based on this logic, we present three categories of extended proof rules and explain how they benefit verification tools.

We then clarify the concept of deep/shallow embeddings for programming languages, assertion languages, and program logics, along with the settings of this paper in section~\ref{sec:nomen}, since they will affect the correctness of extended proof rules.
In section~\ref{sec:embed}, we formally define a deep embedding and three shallow embeddings of the Hoare logic supporting control flow reasoning.
We also briefly review existing Hoare-logic based verification frameworks using these embeddings.

Then, in section~\ref{sec:proof}, we present formal proofs of these extended rules under each embedding method and discuss their main challenges.
In section~\ref{sec:d2s}, we present our approach to lift a shallowly embedded program logic into a deeply embedded one to avoid otherwise challenging proofs of extended rules.
This approach relies on careful choices of primary rules, which is explained in section~\ref{sec:choice}.

After that, we discuss various extensions to the Hoare logic and its embedding in section~\ref{sec:hypo}.
They include separation logic, procedure calls, total correctness, non-determinism, and impredicative assertions.

In section~\ref{sec:project}, we apply our lifting method to shallowly embedded VST to obtain our deeply embedded VST.
We then present Iris-CF, a shallowly embedded program logic which extends Iris to support control flow reasoning, and we lift Iris-CF into a deeply embedded logic, Iris-Imp, and equip it with extended rules for the demonstration of our lifting method.

Lastly, we discuss related works of program logic embeddings in section~\ref{sec:related-embed} and conclude the paper in section~\ref{sec:conclusion}.

We formalize all results of this paper in Coq in a repository\cite{exrule-repo}.
It contains formalizations and proofs of extended rules under three different shallow embeddings, the deeply embedded VST, and our Iris-CF and Iris-Imp.
\section{Program Logic \& Extended Rules}
\label{sec:rules}

A verification tool always comes with some very basic primary rules, like compositional rules, the consequence rule, and singleton command rules.
However, we could further enrich its capability by adding extended proof rules, some of which will be of great assistance to users for proof simplification and automation.

In this section, we first present in section~\ref{sec:setting} a toy language and a set of primary rules for the program logic that we will use for demonstration throughout this paper.
Then based on this program logic, we introduce in the remainder of this section three categories of useful extended proof rules: transformation rules (section~\ref{sec:rules-trans}), structural rules (section~\ref{sec:rules-struct}), and inversion rules (section~\ref{sec:rules-inv}).
We demonstrate some representative rules in each category (figure~\ref{fig:extended}) and show their potential usages.



\subsection{The Toy Language and the Program Logic}
\label{sec:setting}

\begin{figure}[h]    
\small\paragraph{While-CF Language} \

\begin{minipage}{\linewidth}
\begin{center}
    $x, y \in \text{program-variable}$ \ \
    $a, b \in \text{logic-variable}$ \\
    $v \in \text{value}$ \ \
    $e \in \text{expression}$ \ \
    $\sigma \in \text{state}$
\end{center}
$$
\begin{array}{rcl}
c \in \text{command} &:=& \cskip \lsep x = e \lsep c_1 \cseq c_2
\lsep \cif{e}{c_1}{c_2} \\
&\lsep& \cfor{c_1}{c_2} \lsep \cbreak \lsep \ccontinue
\end{array}
$$
\end{minipage}

\small\paragraph{Assertion Language} \
\begin{minipage}{\linewidth}
$$
\begin{array}{rcl}
P,Q,R \in \text{assertion} &:=&
\top \lsep \bot \lsep \lnot P \lsep P \land Q \lsep P \lor Q
\lsep P \rightarrow Q \\
&\lsep& \forall a. P \lsep \exists b. P \lsep \doublebrackets{e} = v \lsep P[x\mapsto e] \lsep \cdots
\end{array}
$$
\end{minipage}

\small\paragraph{Primary Proof Rules} \
\vspace{-2ex}
\begin{mathpar}
    \inferrule[hoare-skip]{}{
        \vdash \triple{P}{\cskip}{P, [\bot, \bot]}
    }
    \and
    \inferrule[hoare-break]{}{
        \vdash \triple{P}{\cmdbreak}{\bot, [P, \bot]}
    }
    \and
    \inferrule[hoare-continue]{}{
        \vdash \triple{P}{\cmdcontinue}{\bot, [\bot, P]}
    }
    \and
    \inferrule[hoare-seq]{
        \vdash \triple{P}{c_1}{Q', [\vec{R}]} \\\\
        \vdash \triple{Q'}{c_2}{Q, [\vec{R}]}    
    }{
        \vdash \triple{P}{c_1 \cseq c_2}{Q, [\vec{R}]}        
    }
    \and
    \inferrule[hoare-loop]{
        \vdash \triple{P}{c_1}{I, [Q, I]} \\\\
        \vdash \triple{I}{c_2}{P, [Q, \bot]}
    }{
        \vdash \triple{P}{\cfor{c_1}{c_2}}{Q, [\bot, \bot]}
    }
    \and
    \inferrule[hoare-if]{
        \vdash \triple{P \land \doublebrackets{e} = \btrue}{c_1}{Q, [\vec{R}]} \\\\
        \vdash \triple{P \land \doublebrackets{e} = \bfalse}{c_2}{Q, [\vec{R}]}
    }{
        \vdash \triple{P}{\cmdif e \cmdthen c_1 \cmdelse c_2}{Q, [\vec{R}]}
    }
    \and
    \inferrule[hoare-assign]{}{
        \vdash \triple{P[x \mapsto e]}{x = e}{P, [\bot, \bot]}
    }
    \and
    \inferrule[hoare-consequence]{
        \vdash \triple{P'}{c}{Q', [R_\ekb', R_\ekc']} \\\\
        P \vdash P' \and Q' \vdash Q \and
        R_\ekb' \vdash R_\ekb \and R_\ekc' \vdash R_\ekc
    }{
        \vdash \triple{P}{c}{Q, [R_\ekb, R_\ekc]}
    }
\end{mathpar}

\caption{The While-CF Programming Language and Primary Proof Rules}
\label{fig:setting}
\end{figure}


The While-CF language (\textbf{While}-language with \textbf{c}ontrol \textbf{f}low commands) in figure~\ref{fig:setting} includes the assignment statement, $x = e$, which assigns the value of $e$ to variable $x$, and empty command $\cskip$ which does nothing.
It includes three basic structural commands: $c_1 \cmdseq c_2$ executes $c_1$ and $c_2$ in sequence; $\cmdif e \cmdthen c_1 \cmdelse c_2$ is the regular if-statement;
$\cfor{c_1}{c_2}$ is the for-loop in C language style, where $c_1$ is the loop body and $c_2$ is the increment step after the execution of each loop iteration.\footnote{
    Although the language is named While-CF, we use the for-loop instead of the while-loop because the for-loop will yield more interesting extended proof rules, and it is also the choice of the VST.
    It is also non-trivial because a continue command in the for-loop body $c_1$ does not skip the increment step $c_2$, which cannot be encoded by a while-loop.
}
It has $\cmdbreak$ command and $\cmdcontinue$ command to manipulate control flows in a loop.

Figure~\ref{fig:setting} shows some terms of the assertion language.
It includes basic first-order logic terms and terms to support express evaluation and variable assignment.
The term $\doublebrackets{e}$ denotes the value of $e$ evaluated on given state $\sigma$, i.e., $\text{eval}(e, \sigma)$, and $\doublebrackets{e} = v$ means expression $e$ evaluates to value $v$.
$P[x \mapsto e]$ denotes the assertion after substitution of all occurrences of $x$ by the value of $e$ evaluated on the current state.

Primary proof rules are listed in figure~\ref{fig:setting}.
Intuitively and informally, $\vdash \triple{P}{c}{Q, [R_\ekb, R_\ekc]}$\footnote{We use $[\vec{R}]$ as a syntax sugar for $[R_\ekb, R_\ekc]$.} means that: starting from any program state satisfying pre-condition $P$, 
\begin{description}
    \item[safety:] the execution of the command $c$ never cause an error,
    \item[correctness:] and if the execution of $c$ terminates and its termination is caused by $\cbreak$ or $\ccontinue$, the post-state will satisfy corresponding control flow post-condition $R_\ekb$ and $R_\ekc$.
    Otherwise, if the program terminates naturally, the post-state will satisfy normal post-condition $Q$.
\end{description}

The \textsc{hoare-skip}, \textsc{hoare-break} rule, and \textsc{hoare-continue} are rules for singleton control flow commands that introduces the pre-condition into corresponding control flow post-conditions.
The singleton command rule \textsc{hoare-assign}\footnote{
Here we use the backward version of \textsc{hoare-assign}.
There is another forward version of \textsc{hoare-assign}, $\vdash \triple{P}{x=e}{\exists v.\, P[x \mapsto v] \land \doublebrackets{x} = \doublebrackets{e[x \mapsto v]}}$, which can be derived from the former one and is often used in forward symbolic execution.
We use both versions in this paper without notice.
} asserts that the value of variable $x$ used in the pre-condition will be changed in the post-condition.
The \textsc{hoare-seq} rule splits the proof of two sequentially composed commands.
The \textsc{hoare-loop} rule uses two triples of the loop body and the increment step to build a specification for the for-loop statement.
The \textsc{hoare-if} rule splits the proof of if-statement into two branches to prove them separately.
The \textsc{hoare-consequence} rule allows provers to strengthen the pre-condition and weaken the post-condition.



Although this programming language is a simple toy example, we may extend it with heap manipulation commands and function invocations as the paper will discuss in section~\ref{sec:hypo}.
Other primary rules related to separation logic and function invocation can also be included.
But to demonstrate our results about extended rules in section~\ref{sec:rules} and section~\ref{sec:proof}, the language and proof rules in figure~\ref{fig:setting} are sufficient.
These primary rules are expressive enough to reason about loop control flows and also present in real verification tools like VST\cite{VST}.

In figure~\ref{fig:setting}, the assertion language and the program logic are presented in a deep embedding style using the syntax of assertions and derivation rules, but it is only for the purpose of demonstration.
However, both of them can also be defined using shallow embeddings directly using their interpretation and semantics.
We will soon investigate into different embeddings of them in section~\ref{sec:nomen}.

\begin{figure}[h]
\paragraph{Extended Rules --- Transformation Rules} \
\vspace{-2ex}
\begin{mathpar}
\inferrule*[left=if-seq]{
    \vdash \triple{P}{\cif{e}{c_1 \cseq c_3}{c_2 \cseq c_3}}{Q, [\vec{R}]}
}{
    \vdash \triple{P}{(\cif{e}{c_1}{c_2}) \cseq c_3}{Q, [\vec{R}]}
} \and
\inferrule[loop-nocontinue]{
    c_1, c_2 \text{ contain no } \ccontinue \\
    \vdash \triple{P}{\cfor{c_1 \cseq c_2}{\cskip}}{Q, [\vec{R}]}
}{
    \vdash \triple{P}{\cfor{c_1}{c_2}}{Q, [\vec{R}]}
} \and
\inferrule*[left=loop-unroll1]{
    \vdash \triple{P}{c_1}{P_1, [R_\ekb, P_1]} \\
    \vdash \triple{P_1}{c_2}{P_2, [R_\ekb, P_2]} \\
    \vdash \triple{P_2}{\cfor{c_1}{c_2}}{Q, [\vec{R}]}
}{
    \vdash \triple{P}{\cfor{c_1}{c_2}}{Q, [\vec{R}]}
}
\end{mathpar}

\paragraph{Extended Rules --- Structural Rules} \
\vspace{-2ex}
\begin{mathpar}
    \inferrule*[lab=hoare-ex]{
        \text{forall } x.\, \vdash \triple{P(x)}{c}{Q, [\vec{R}]}
    }{
        \vdash \triple{\exists x.\, P(x)}{c}{Q, [\vec{R}]}
    }
    \and
    \inferrule*[left=nocontinue]{
        c \text{ contains no } \ccontinue \\\\
        \vdash \triple{P}{c}{Q, [R_\ekb, R_\ekc]}
    }{
        \vdash \triple{P}{c}{Q, [R_\ekb, R'_\ekc]}
    }
\end{mathpar}

\paragraph{Extended Rules --- Inversion Rules} \
\vspace{-2ex}
\begin{mathpar}
    \inferrule[seq-inv]{
        \vdash \triple{P}{c_1 \cseq c_2}{Q, [\vec{R}]}
    }{
        \text{exists } Q'.\, \vdash \triple{P}{c_1}{Q', [\vec{R}]} \\\\
        \quad \text{and } \vdash \triple{Q'}{c_2}{Q, [\vec{R}]}
    }
    \and
    \inferrule[loop-inv]{
        \vdash \triple{P}{\cfor{c_1}{c_2}}{Q, [\bot, \bot]}
    }{
        \text{exists } I_1, I_2.\, \vdash \triple{I_1}{c_1}{I_2, [Q, I_2]} \\\\
        \quad \text{and } \vdash \triple{I_2}{c_2}{I_1, [Q, \bot]} \text{ and } P \vdash I_1
    }
    \and
    \inferrule*[left=if-inv]{
        \vdash \triple{P}{\cmdif e \cmdthen c_1 \cmdelse c_2}{Q, [\vec{R}]}
    }{
        \vdash \triple{P \land \doublebrackets{e} = \btrue}{c_1}{Q, [\vec{R}]} \text{ and }
        \vdash \triple{P \land \doublebrackets{e} = \bfalse}{c_2}{Q, [\vec{R}]}
    }
\end{mathpar}

\caption{Three Categories of Extended Rules and Their Representatives}
\label{fig:extended}
\end{figure}

There are lots of useful extended rules that can be implemented above the program logic in figure~\ref{fig:setting}, and we organize them into three categories and present some representatives of them in figure~\ref{fig:extended}.
We explain in detail the usage of these extended rules in the following sections.

\subsection{Transformation Rules}
\label{sec:rules-trans}

It is often the case that we know two different programs are semantically equivalent and we should be able to substitute one of them for the other in a Hoare triple while preserving the triple's validity.
Rules like \textsc{if-seq}, \textsc{loop-nocontinue}, and \textsc{loop-unroll1}\footnote{
As far as we know, Andrew W. Appel and the development team of VST has integrated these transformation rules into VST for improving proof automation since 2017.
} allow such semantic preserving transformations, and provers can transform the program into the one that is easier to verify and sometimes automate its proof.
\textit{We classify these extended rules that transforms the command in a Hoare triple as transformation rules.}

\

\noindent\textit{\textbf{The proof rule}} \textsc{if-seq}.
The program $(\cif{e}{c_1}{c_2}) \cseq c_3$ and $\cif{e}{c_1 \cseq c_3}{c_2 \cseq c_3}$ behaves similarly and we can use \textsc{if-seq} to transform the former into the latter.

\begin{example}
For a concise demonstration, we first start with a toy example (for which exists another work-around instead of using \textsc{if-seq} and we will discuss it soon), where we want to prove the following for some given $m, e, c_2, Q, \vec{R}$,
$$
\vdash \triple{
    \underset{P}{\underbrace{\exists n. \doublebrackets{x} = n \times m \land \doublebrackets{y} = m}}
}{
    (\cif{e}{\cbreak}{\underset{c_1}{\underbrace{z=x/y}}}) \cseq c_2  
}{
    Q, [\vec{R}]
}
$$
Regularly, we need to explicitly provide an intermediate assertion $Q'$ and show\footnote{We abbreviate assertions $\doublebrackets{e} = \text{true}$ as $\doublebrackets{e}$ and $\doublebrackets{e} = \text{false}$ as $\lnot \doublebrackets{e}$ for the sake of space.}
\begin{equation}
\label{eq:ifseq-1}
\inferrule*[leftskip=10pt, rightskip=15pt]{
    \inferrule*[rightskip=15pt]{
        \inferrule*[rightskip=15pt]{
            \inferrule[]{
                \cdots
            }{
                P \wedge \doublebrackets{e} \vdash R_\ekb
            }
        }{
            \vdash\{P \wedge \doublebrackets{e}\}
            \cbreak \{Q', [\vec{R}]\}
        }
        \and
        \inferrule*[]{
            \cdots
        }{
            \vdash \{P \wedge \lnot\doublebrackets{e}\}
            c_1\{Q', [\vec{R}]\}
        }
    }{
        \vdash \triple{P}{(\cif{e}{\cbreak}{c_1})}{Q', [\vec{R}]}
    } \and
    \inferrule*[]{\cdots}{
        \vdash \triple{Q'}{c_2}{Q, [\vec{R}]}
    }
}{
    \vdash \triple{P}{(\cif{e}{\cbreak}{c_1}) \cseq c_2}{Q, [\vec{R}]}
}
\end{equation}
But when we move $c_2$ inside the if-statement by $\textsc{if-seq}$, we only need to prove
\begin{equation}
    \label{eq:ifseq-2}
    \inferrule*[]{
        \inferrule*[]{
            \inferrule*[rightskip=15pt]{
                \inferrule[]{
                    \cdots
                }{
                    P \wedge \doublebrackets{e} \vdash R_\ekb
                }
            }{
                \vdash\{P \wedge \doublebrackets{e}\}
                \cbreak \cseq c_2 \{Q', [\vec{R}]\}
            } \and
            \inferrule*[]{\cdots}{
                \vdash \triple{ P \wedge \lnot\doublebrackets{e} }{ c_1 \cseq c_2 }{Q, [\vec{R}]}
            }
        }{
            \vdash \triple{P}{\cif{e}{\cbreak \cseq c_2}{c_1 \cseq c_2}}{Q, [\vec{R}]}
        }
    }{
        \vdash \triple{P}{(\cif{e}{\cbreak}{c_1}) \cseq c_2}{Q, [\vec{R}]}
    }
\end{equation}
As the reader may wonder, to prove the sequential composition $c_1\cmdseq c_2$ in \eqref{eq:ifseq-2}, we still need to prove the last two goals for $c_1$ and $c_2$ in \eqref{eq:ifseq-1}.
The difference here is that when verifying \eqref{eq:ifseq-2}, we can symbolically execute $c_1$ along with $c_2$ from the beginning
to the end.
But in \eqref{eq:ifseq-1}, we need to symbolically execute two branches to a unifying post-condition $Q'$ before executing $c_2$, while symbolic executions of two branches might produce different free variables in their postconditions and it brings troubles.
We explain this difficulty in detail in following paragraphs.

\textbf{Symbolic execution} in theorem provers is a technique for proof automation. It simulates the execution of the program and apply the execution's effect to the pre-condition according to proof rules.
Both VST and Iris uses such technique to automate their proofs.
For example, symbolic execution can automatically transform the proof goal of the false branch (sequentially composed with a skip command for demonstration) in \eqref{eq:ifseq-1} as \eqref{eq:symb-exe-eg} shows (arrows represent symbolic execution steps by labelled proof rules).
\begin{equation}
    \label{eq:symb-exe-eg}
    \begin{array}{rl}
        &\triple{\exists n.\,\doublebrackets{x} = n \times m \land \doublebrackets{y} = m \land \cdots}{z = x / y \cseq \cskip}{Q', [\vec{R}]} \\
        \xrightarrow[\ ]{\textsc{hoare-ex}} &\triple{\doublebrackets{x} = n \times m \land \doublebrackets{y} = m \land \cdots}{z = x / y \cseq \cskip}{Q', [\vec{R}]} \\
        \xrightarrow[{\vspace{-10pt}\textsc{hoare-seq}}]{\textsc{hoare-assign}} &\triple{\doublebrackets{z} = n \land \doublebrackets{x} = n \times m \land \doublebrackets{y} = m \land \cdots}{\cskip}{Q', [\vec{R}]} \\
        \vspace{-9pt} \\
        \xrightarrow[\textsc{hoare-skip}]{\textsc{hoare-assign}} &(\doublebrackets{z} = n \land \doublebrackets{x} = n \times m \land \doublebrackets{y} = m \land \cdots) \vdash Q'
    \end{array}
\end{equation}
During the symbolic execution, it will apply \textsc{hoare-ex} (demonstrated in section~\ref{sec:rules-struct}) to eliminate quantifiers in the pre-condition and introduce the bounded variable as a free one into the context to reason about it.
On the contrary, symbolic execution cannot automatically determine which free variables should be reverted back into the assertion as bounded ones in which order so that the following execution steps can proceed.

Back to the example, symbolic execution of $c_1$ in \eqref{eq:symb-exe-eg} will introduce $n$ into the context, which is necessary for forwarding $z=x/y$.
To achieve a unifying $Q'$ in \eqref{eq:ifseq-1}, we need to revert $n$ back into $Q'$ in the symbolic execution of $c_1$, since the other branch does not introduce $n$ and it remains bounded.
In other words, the $Q'$ in \eqref{eq:symb-exe-eg} should become $\exists n.\,\doublebrackets{z} = n \land \doublebrackets{x} = n \times m \land \doublebrackets{y} = m \land \cdots$, instead of the final assertion in \eqref{eq:symb-exe-eg}.
However, in \eqref{eq:ifseq-2}, after we introduce $n$ in  the execution of $c_1$, we do not need to revert it back and when $c_2$ access $z$ and $n$ during its execution, we do not need to introduce $n$ again.
\textsc{if-seq} reduces the complexity of symbolic execution in this case.
\end{example}

\paragraph{Remark.}
As we have mentioned, the toy example can be resolved without applying \textsc{if-seq}.
As the proof scheme in \eqref{eq:another} indicates, we can extract $n$ in $P$ before we start the proof of the sequential composition and in this way, there is no need to revert $n$ back again when finding the post-conditions of if-branches.
\begin{equation}
    \label{eq:another}
    \inferrule*[]{
        \inferrule*[]{
            \inferrule*[]{
                \cdots
            }{
                \vdash \triple{\cdots}{\cif{e}{\cbreak}{c_1}}{\cdots}
            }
            \and
            \inferrule*[]{
                \cdots
            }{
                \vdash \triple{\cdots}{c_2}{\cdots}
            }
        }{
            \vdash \triple{\doublebrackets{x} = n \times m \land \doublebrackets{y} = m}{(\cif{e}{\cbreak}{c_1}) \cseq c_2}{Q, [\vec{R}]}
        }
    }{
        \vdash \triple{\exists n. \doublebrackets{x} = n \times m \land \doublebrackets{y} = m}{(\cif{e}{\cbreak}{c_1}) \cseq c_2}{Q, [\vec{R}]}
    }
\end{equation}

However, this does not mean the example is meaningless, because in many cases, the existential quantifier in the precondition cannot and should not be eliminated in advance.
For example, the precondition below is an assertion to state the storage of a linked list in the heap.
To prove the Hoare triple below for some $x, l, e, c_1', c_2, Q, \vec{R}$, we are unable to eliminate all quantifiers in listrep by \textsc{hoare-ex} without a fixed length of list $l$.
In practice, we should be able to prove it without such knowledge and keep the listrep predicate folded until reaching a test-empty operation to the linked list, and then should we unfold the predicate once and extract quantifiers.
\textsc{if-seq} is the only option in these cases.
$$
\vdash \triple{
    \text{listrep}(x, l)
}{
    (\cif{e}{\cbreak}{c_1'}) \cseq c_2  
}{
    Q, [\vec{R}]
}
$$
\vspace{-8pt}
$$
\begin{aligned}
    \text{where } \text{listrep}(x,l) &\triangleq (l = \text{nil} * \emp) \\
    &\lor (\exists y, v, l'.\, l = v :: l' * x\mapsto (v, y) * \text{listrep}(y, l'))
\end{aligned}
$$

\

\

\noindent\textit{\textbf{The proof rule}} \textsc{loop-nocontinue}.
Another exemplary transformation rule in figure~\ref{fig:extended} is \textsc{loop-nocontinue}, which allows merging loop body and incremental step into one command when neither of them contains $\ccontinue$.
In this way, users only need to find one loop invariant for the loop body instead of two (one before the loop body and one before the increment step) when reasoning about a loop with no $\ccontinue$.

\begin{example}
This example combines \textsc{loop-nocontinue} and \textsc{if-seq}.
The program divides $x$ by $y$ twice but separate two divisions in the incremental step and the loop body.
\vspace{-12pt}
$$
\vdash 
    \Bigg\{\underset{P}{\underbrace{\exists n. \begin{array}{c}
        \doublebrackets{x}=m^{2n}\\
        \land \doublebrackets{y} = m
    \end{array}}}\Bigg\} 
    \cfor{
    \texttt{if } \underset{e}{\underbrace{x>1}} \
    \begin{array}{l}
    \\
      \texttt{then } \cbreak \\
      \texttt{else } \underset{c_1}{\underbrace{z=x/y}}
    \end{array}
    }{\underset{c_2}{\underbrace{x=z/y}}}
    \Bigg\{\underset{Q}{\underbrace{\doublebrackets{x} = 1}}, [\bot]\Bigg\}
$$
The precondition $P$ can serve as a loop invariant, but regularly, we still need to find another loop invariant after the loop body but before the incremental step.
However, by the symbolic execution \eqref{eq:loop-nocont-eg}, it first removes the incremental step and with $\cskip$ as the incremental step, it is obvious that another loop invariant is the same as $P$.
Then it can open the loop and prove that the loop body obeys the invariant.
With the help of \textsc{if-seq}, it can automatically reach two easily provable proof goals without manual assistance.
\begin{equation}
\label{eq:loop-nocont-eg}
\begin{aligned}
&\vdash \triple{P}{
    \cfor{\cif{e}{\cbreak}{c_1}}{c_2}
}{Q, [\bot]} \\
\xrightarrow[\textsc{nocontinue}]{\textsc{loop-}} &\vdash \triple{P}{
    \cfor{(\cif{e}{\cbreak}{c_1}) \cseq c_2}{\cskip}
}{Q, [\bot]} \\
\xrightarrow[\textsc{if-seq}]{\textsc{hoare-loop}} &\vdash \triple{P}{
    \cif{e}{\cbreak \cseq c_2}{(c_1 \cseq c_2)}
}{P, [P, Q]} \\
\xrightarrow[\textsc{hoare-break}]{\textsc{hoare-if}} &
    P \land \lnot \doublebrackets{e} \vdash Q \quad \text{and} \quad
    \vdash \triple{P}{c_1 \cseq c_2}{P, [P, Q]}
\end{aligned}
\end{equation}
\end{example}




\noindent\textit{\textbf{The proof rule}} \textsc{loop-unroll1}.
One more transformation rule we want to mention here is the \textsc{loop-unroll1} rule.
It allows peeling the first iteration of a for-loop and prove it as a separate triples.
This is especially useful in two cases: when the first iteration of a loop does something different from the remaining iterations, e.g., initialization of some data, and the loop invariant for the remaining iterations can be greatly simplified without the first iteration; and when the loop only runs constant number of iterations, it is easier to simply unfold it into a sequence of loop bodies and use symbolic execution to prove them automatically.

Although the \textsc{loop-unroll1} rule in figure~\ref{fig:extended} does not perfectly match our criterion for transformation rules, their essence are the same.
The \textsc{loop-unroll1} rule has its counterpart, the following \textsc{while-unroll1} for the while command, in the While language without control flow statements.
It unrolls the first iteration of the while-loop into an if-statement.
\begin{mathpar}
\inferrule*[left=while-unroll1]{
    \vdash \triple{P}{
        \cif{b}{c \cseq \cwhile{b}{c}}{\cskip}
    }{Q}
}{
    \vdash \triple{P}{\cwhile{b}{c}}{Q}
}
\end{mathpar}
If there are control flow commands like $\cbreak$ or $\ccontinue$, we cannot simply transform the for-loop into an if-statement because these commands will jump out of the scope of the original loop and interfere executions outside the loop.
As a result, we need to split the if-statement in \textsc{while-unroll1} into several Hoare triples in \textsc{loop-unroll1} and specify control flow post-conditions for the loop body and the increment command.

\subsection{Structural Rules}
\label{sec:rules-struct}

\textit{We classify extended rules that transforms pre-/post-conditions in proof goals as structural rules.}
These rules allow provers to adjust pre-/post-conditions into forms that permit more organized proofs.

A typical case is the use of the primary rule \textsc{hoare-consequence}\footnote{
    However, we do not formalize \textsc{hoare-consequence} rule as an extended rule but as a primary rule due to historical reasons.
    It is part of the primary rules in Cook's soundness and completeness proof of the Hoare logic \cite{cook1978soundness} and is admitted as a primary rule of any Hoare logic ever since.
}.
For example, in the last step of symbolic execution \eqref{eq:loop-nocont-eg} from the last section, we need to prove the following assumption for the then-branch in order to apply \textsc{hoare-if} to verify the if-statement, where direct application of \textsc{hoare-break} does not work since the pre-condition does not match the break post-condition $P$ and other two post-conditions are not $\bot$.
We will use \textsc{hoare-consequence} to weaken the pre-condition into $Q$ and strengthen the post-condition $P$ into $\bot$ to match the form of \textsc{hoare-break}.
In this way, we can directly apply \textsc{hoare-break} to the prove the new triple.
$$
\inferrule*[]{
    P \land \lnot \doublebrackets{e} \vdash Q
    \and
    \bot \vdash P
    \and
    \vdash \triple{Q}{\cbreak}{\bot,[\bot,Q]}
}{
    \vdash \triple{(\underset{P}{\underbrace{\exists n. \doublebrackets{x}=m^{2n} \land \doublebrackets{y} = m}}) \land \doublebrackets{x} \leq 1}{\cbreak}{P, [P, \underset{Q}{\underbrace{\doublebrackets{x} = 1}}]}
}
$$

In this section, we mainly discuss two structural rules in figure~\ref{fig:extended}, \textsc{hoare-ex} and \textsc{nocontinue}, which we consider as extensions of \textsc{hoare-consequence}.
When separation logic is taken into consideration in section~\ref{sec:frame}, we will encounter two more structural rules, \textsc{frame} and \textsc{hypo-frame}, and we postpone their discussions until then.



\

\noindent\textit{\textbf{The proof rule}} \textsc{hoare-ex}.
As we have already seen in section~\ref{sec:rules-trans}, it is a common situation that we need to prove a Hoare triple whose the precondition is existentially quantified, e.g.,
loop invariants are usually existentially quantified\footnote{
    This is often necessary because logical variables utilized in the proof of the loop body usually are not exposed to the proof outside the loop.
    These variables should not be introduced before the proof of loop body and we need existential quantifiers to introduce this variables in the loop invariant.
}.
If provers can eliminate those existential quantifiers and extract bounded variables and related pure facts in preconditions into the meta-logic context,
they can then apply domain-specific theories to those extracted variables --- they are not buried in assertions any longer.
\textsc{hoare-ex} enable us to perform such extraction.

\begin{example}    
In the symbolic execution example \eqref{eq:symb-exe-eg} in section~\ref{sec:rules-trans} (rephrased below), the first step is applying \textsc{hoare-ex} to remove the existential quantifier.
To verify the assignment $z=x/y$, we know the result of $z$ will be $n$ and it must be a free variable so that we can perform substitution of $z$ by $n$ according to \textsc{hoare-assign}.
$$
\begin{array}{rl}
    &\triple{\exists n.\,\doublebrackets{x} = n \times m \land \doublebrackets{y} = m \land \cdots}{z = x / y \cseq \cskip}{Q', [\vec{R}]} \\
    \xrightarrow[\ ]{\textsc{hoare-ex}} &\triple{\doublebrackets{x} = n \times m \land \doublebrackets{y} = m \land \cdots}{z = x / y \cseq \cskip}{Q', [\vec{R}]} \\
    \cdots
\end{array}
$$
\textsc{hoare-ex} is different from and sometimes cannot be derived directly from \textsc{hoare-consequence}.
\textsc{hoare-consequence} changes pre-/post-conditions by entailments of the assertion logic, while \textsc{hoare-ex} is a direct entailment between two Hoare triples in the meta-logic.
\end{example}

\

\textit{\textbf{The proof rule}} \textsc{nocontinue}.
This rule allows us to modify the continue assertion arbitrarily since the program will never exit by continue.
We can use it as an enhancement to the consequence rule when we need to change the continue assertion.
We can use this proof rule to support derivations of loop-related extended rules.
For example, as we will see later in figure~\ref{fig:loop-nocont-eg}, we can prove \textsc{loop-nocontinue} using the \textsc{nocontinue} and inversion rules, which is the proof taken by VST.

\subsection{Inversion Rules}
\label{sec:rules-inv}

Compositional rules allow us to combine proofs for separate modules into one proof for a larger program specification.
But in some cases, we would like to extract information for these modules from the complete specification for other purposes.
\textit{We classify extended rules that extract premises of a primary rule from its goal as inversion rules.}




\Cref{fig:extended} shows two example inversion rules, the \textsc{seq-inv} and \textsc{loop-inv}.
\textsc{seq-inv} is the reversed sequencing rule, where the triple about $c_1 \cseq c_2$ gives us the triples of $c_1$ and $c_2$.
\textsc{loop-inv} is the reversed loop rule, which gives us the intermediate loop invariant.
They are typical rules that reproduce the premises in compositional rules from the original conclusion.
In the following sections, we discuss only the \textsc{seq-inv} rule because it is the most representative inversion rule and other ones have similar conclusions.

\begin{example}
With inversion rules, we can more easily destruct and reorganize proof trees in deeply embedded logic, and in our deeply embedded VST, they have been used extensively in proving other extended rules like transformation rules and structural rules.
For example, in figure~\ref{fig:loop-nocont-eg}, the transformation rule \textsc{loop-nocontinue} can be proven given \textsc{seq-inv} and \textsc{loop-inv}.
\begin{figure}[h]
    \vspace{-1em}
    \begin{mathpar}
    \inferrule*[Right=conseq, leftskip=5em]{
    \inferrule*[Right=hoare-loop, rightskip=0.5em, leftskip=0.5em]{
    \inferrule*[Right=nocontinue, rightskip=0.5em, leftskip=0.5em]{
    \inferrule*[Right=seq-inv, rightskip=0.5em, leftskip=0.5em]{
    \inferrule*[Right=skip-inv \& conseq, rightskip=0.5em, leftskip=0.5em]{
    \inferrule*[Right=loop-inv, rightskip=0.5em, leftskip=0.5em]{
    c_1, c_2 \text{ has no } \cmdcontinue \\\\
    \vdash \triple{P}{\cfor{c_1 \cseq c_2}{\cskip}}{Q, [\vec{R}]}
    }{
    \vdash \triple{I_1}{\cskip}{P, [Q, \bot]} \\ \vdash \triple{P}{c_1\cseq c_2}{I_1, [Q, P]}
    }
    }{
    \vdash \triple{P}{c_1\cseq c_2}{P, [Q, P]}
    }
    }{
    \vdash \triple{P}{c_1}{I_2, [Q, P]} \\ \vdash \triple{I_2}{c_2}{I_1, [Q, P]}
    }
    }{
    \vdash \triple{P}{c_1}{I_2, [Q, P]} \\ \vdash \triple{I_2}{c_2}{I_1, [Q, \bot]}
    }
    }{
    \vdash \triple{P}{\cfor{c_1}{c_2}}{Q, [\bot, \bot]}
    }
    }{
    \vdash \triple{P}{\cfor{c_1}{c_2}}{Q, [\vec{R}]}
    }
    \end{mathpar}
    \caption{A Derivation of \textsc{loop-nocontinue} based on Inversion Rules and Structural Rules}
    \label{fig:loop-nocont-eg}
    \end{figure}
We first extract information from the original triple by inversion rules, and use \textsc{nocontinue} to adjust the continue post-condition of the increment step $c_2$ to $\bot$ so that we can subsequently apply \textsc{hoare-loop} to obtain the triple of the new loop.
\end{example}

\begin{example}
More crucially, inversion rules enable the destruction of shallowly embedded triples into smaller pieces without unfolding the definition of the triple and users can reorganize existing ``proofs'' to produce new triples.

For example, consider an anotated program $\cfor{c_1\cmdseq \texttt{assert }I \cmdseq c_2}{c_3}$, where $\texttt{assert }I$ asserts that the assertion $I$ holds between executions of $c_1$ and $c_2$ and it can guide prover's verification.
Usually, the main burden of verifying a loop is to find the loop invariant.
But for this program, if we can check $\vdash \triple{P}{c_1}{I, [Q, \bot]}$ and $\vdash \triple{I}{c_2 \cmdseq c_3 \cmdseq c_1}{I, [Q, \bot]}$ to be true, then the following derivation directly proves the triple $\vdash\triple{P}{\cfor{c_1\cmdseq \texttt{assert }I \cmdseq c_2}{c_3}}{Q, [\bot, \bot]}$ for the loop.
\begin{mathpar}
\inferrule*[right=seq-inv]{
    \vdash \triple{P}{c_1}{I, [Q, \bot]} \and
    \vdash \triple{I}{c_2 \cmdseq c_3 \cmdseq c_1}{I, [Q, \bot]}
}{
    \inferrule*[right=seq, rightskip=3em]{
            \inferrule*[]{
                \vdash \triple{P}{c_1}{I, [Q, \bot]} \and
                \vdash \triple{I}{c_2}{I_2, [Q, \bot]} \\\\
                \vdash \triple{I_1}{c_1}{I, [Q, \bot]} \and
                \vdash \triple{I_2}{c_3}{I_1, [Q, \bot]}
            }{
                \vdash \triple{P \lor I_1}{c_1}{I, [Q, \bot]} \and
                \vdash \triple{I}{c_2}{I_2, [Q, \bot]} \and
                \vdash \triple{I_2}{c_3}{I_1, [Q, \bot]}
            }
    }{
        \inferrule*[right=loop, rightskip=0.5em]{
            \vdash \triple{P \lor I_1}{c_1\cmdseq \texttt{assert }I \cmdseq c_2}{I_2, [Q, \bot]} \and
            \vdash \triple{I_2}{c_3}{P \lor I_1, [Q, \bot]}
        }{
            \inferrule*[right=conseq]{
                \vdash\triple{P \lor I_1}{\cfor{c_1\cmdseq \texttt{assert }I \cmdseq c_2}{c_3}}{Q, [\bot, \bot]}
            }{
                \vdash\triple{P}{\cfor{c_1\cmdseq \texttt{assert }I \cmdseq c_2}{c_3}}{Q, [\bot, \bot]}
            }
        }
    }
}
\end{mathpar}
With the help of \textsc{seq-inv}, two intermediate assertions $I_1$ and $I_2$ are extracted for free.
By re-ordering commands into a new loop body $c_2 \cmdseq c_3 \cmdseq c_1$ using $I$ as the invariant, $P \lor I_1$ and $I_2$ become our loop invariants and we establish the final triple without any additional proofs.
Notice that the triples we need to provide as the premise of this derivation can be verified symbolic execution without the need to manually construct intermediate assertions $I_1$ and $I_2$.
This approach of using \textsc{seq-inv} to rearrange program orders helps simplify the proof of this kind of loops by making it easier to apply symbolic executions.



\end{example}


\paragraph{Summary.}
In summary, inversion rules have one-to-one correspondence to primary rules.
For each primary rule, we can derive an inversion rule that produces the premises of it from the conclusion.
People often apply inversion rules to assumptions to extract information from it and use these information to aid their proofs.
On the contrary, transformation rules and structural rules do not have such correspondence and are applied directly to the proof goals (the Hoare triple specification) to generate sub-goals with easier proofs.
Transformation rules change the program in the specification but keep original pre-/post-conditions, while structural rules only have effect on pre-/post-conditions but not the program.

Although the discussion of extended rules in this paper focuses on a programming language with control flow commands, most of these rules are still meaningful in a language that does not feature the control flow commands.
For example, \textsc{if-seq} and \textsc{hoare-ex} does not depend on the control flow post-condition.
And \textsc{loop-unroll1} has a counterpart \textsc{while-unroll} for loops without break and continue.
For any programming language and program logic, it will always have its own set of inversion rules corresponds to its primary rules.
Moreover, we believe adding other features to the language and the logic will yield more extended rules, but they are not well studied and is not the focus of this paper.


\section{Nomenclature}
\label{sec:nomen}



We have presented the program logic and extended rules in section~\ref{sec:rules}, but did not mention how they are formalized at all and we cannot check the correctness of these extended rules until we formally define the program logic.
This section clarifies the notion of the deep embedding and the shallow embedding, which are two different approaches of formalizing languages and logics.
In short, deep embeddings formalize structures while shallow embeddings formalize underlying semantics.

A program-logic-based foundational verification tools contains at least three elements where we could choose different embeddings: \textbf{the programming language}, \textbf{the assertion language} that describe program state properties, and \textbf{the program logic} that reasons about  functional correctness of a program.
Different verification projects choose different combinations of embeddings for these three elements as table~\ref{fig:table-choice} shows.
For shallow embeddings of program logics, we could further divide them into sub-categories: big-step based (BigS.), weakest precondition based (WP), and continuation based (Cont.), which we will discuss in detail in section~\ref{sec:embed}.

\begin{table}[h]
    \begin{center}
\resizebox{\textwidth}{!}{
    \begin{tabular}{ |c|c|c|c| } 
    \hline
        \begin{tabular}{c}
            \textbf{Verification Projects}
        \end{tabular} & \begin{tabular}{c}
            \textbf{Programming} \\\textbf{Language}
        \end{tabular} & \textbf{Program Logic} & \textbf{Assertion} \\
    \hline
    VST (before Sep. 2018) & Deep  & Shallow (Cont.) & Shallow \\ \hline
    VST (after Sep. 2018) & Deep  & Deep & Shallow \\ \hline
    Iris \cite{jung2018iris}, Iris-CF & Deep & Shallow (WP)  &Shallow\\ \hline
    FCSL\cite{FCSL-mechanized} & Shallow & Shallow (WP) &Shallow\\ \hline
    \begin{tabular}{c}
        Software Foundations \cite{SF}\\
        Vol. 2, Vol. 6
    \end{tabular} & Deep & Shallow (BigS.) \& Deep &Shallow\\ \hline
    \begin{tabular}{c}
        Simpl \cite{schirmer2006verification},
        CSimpl \cite{sanan2017csimpl}
    \end{tabular} & Deep & Deep &Shallow\\ \hline
    CAP \cite{yu2003building}, XCAP \cite{ni2006certified} & Deep & Deep &Shallow\\ \hline
    $\mu$C \cite{xu2016practical} & Deep & Deep & Deep\\ \hline
    \end{tabular}
}
    \end{center}
    \caption{Choices between shallow embedding and deep embedding}
    \label{fig:table-choice}
\end{table}

In the rest of this section, we explain what are shallow/deep embeddings for a programming language, an assertion language, and a program logic.


\subsection{Embeddings of Programming Languages}
When formalizing a language (e.g. a programming language or an assertion language), a deep embedding formalizes its syntax tree first and defines its meaning separately.
In contrast, a shallow embedding uses the language's intrinsic semantics as its definition directly.
For example, figure~\ref{fig:sample-lang} shows their differences in defining simple program expressions which only contain integer variables and addition.
The addition operation is defined as a syntax tree constructor in the deeply embedded one and we use another evaluation function to define the semantics of all expressions.
In the shallowly embedded one, the operation is directly defined as a function that computes the summation of its operands' evaluation results or a relation from the expression to the final result.

\begin{figure}[ht]
\centering
\small
$x : \text{variable\_name}$ \ \ \ \ \
$\sigma: \text{prog\_state}$ \ \ \ \ \
$\text{prog\_state} = \text{variable\_name} \to \mathbb{Z}$

\begin{multicols}{2}
Deeply Embedded Program:
$\begin{array}{l}   
e \in \text{expr} := \text{Const}(n) ~\vert~ \text{Var}(v) ~\vert~ \text{Add}(e_1,e_2) \\
\begin{aligned}
\text{eval}(\text{Const}(n), \sigma) &= n \\
\text{eval}(\text{Var}(x), \sigma) &= \sigma(x) \\
\text{eval}(\text{Add}(e_1,e_2), \sigma) &= \text{eval}(e_1, \sigma) + \text{eval}(e_2, \sigma)
\end{aligned}
\end{array}$

Shallowly Embedded Program: 
$\begin{array}{l}
\text{expr} \triangleq \text{prog\_state} \to \mathbb{Z} \\
\begin{aligned}
\text{Const} &= \lambda n. \lambda \sigma. n \\
\text{Var} &= \lambda x. \lambda \sigma. \sigma(x) \\
\text{Add} &= \lambda e_1. \lambda e_2. \lambda \sigma. e_1(\sigma) + e_2(\sigma)
\end{aligned}
\end{array}$
\end{multicols}
\caption{Example: deeply/shallowly embedded expressions}
\label{fig:sample-lang}
\end{figure}

Most foundational verification tools including VST and Iris choose to use deeply embedded programming languages.
The deeply embedded languages separate the syntax and semantics of programs and would allow structural induction over the program's syntax tree in some proofs.
Comparably, typical shallowly embedded languages use functions from initial program states to ending states, or binary relations between initial states and ending states, to represent programs.
In other words, programs are formalized in proof assistants by their denotations instead of syntax trees, which makes it difficult to equip them with structural inductions.
It is also difficult (although not impossible \cite{FCSL-mechanized}) for a shallowly embedded programming language to support concurrency and other extensions.
For example, to implement concurrency, a deeply embedded programming language can syntactically parallel-compose two programs and then choose small-step semantics for describing concurrency since one can interleave parallel steps that programs take.
However, it is less convenient for the denotational shallow embedding to implement such an interleaving.
In the denotational shallow embedding, a program expression itself is a big-step denotation which hides intermediate states of its execution, so we cannot directly define the program with other threads interleaving at this intermediate states.

In some cases, developers would also use mixed embeddings of their programming languages \cite{xu2016practical, FCSL-mechanized}:
these languages use deep embedding to formalize structural compositions and use shallow embedding for singleton commands.
For example, in the $\mu$C framework \cite{xu2016practical} for verifying OS kernels, Xu et. al formalize atomic operations $\gamma$ via a shallow embedding and formalize operation compositions (e.g., sequential composition $s_1 ; s_2$ and non-deterministic choice $s_1 + s_2$) via a deep embedding.
Thus a command $s$ is defined by both embeddings in the formula below.
$$
\begin{array}{rcl}
    s &\in& \text{Command} := \gamma \lsep s_1 ; s_2 \lsep s_1 + s_2 \lsep \textbf{end} \lsep \cdots \\
    \gamma &\in& \text{Abstract State} \rightarrow \text{Abstract State} \rightarrow \text{Prop}
\end{array}
$$
\newpage
\paragraph{Remark.}
In a shallowly embedded language, transformation rules become trivial and meaningless.
For example, in \textsc{if-seq} rule, we have two programs defined below.
\begin{align*}
\cif{e}{c_1 \cseq c_3}{c_2 \cseq c_3} &\triangleq
\lambda \sigma.\, \textsf{match } \text{eval}(e, \sigma) \textsf{ with} ~\vert~ \bfalse \Rightarrow c_3(c_1(\sigma)) \\
&\,\,\qquad\qquad\qquad\qquad\qquad\qquad\vert~ \btrue \Rightarrow c_3(c_2(\sigma)).
\\
\cif{e}{c_1}{c_2} \cseq c_3 &\triangleq
\lambda \sigma.\, c_3(\textsf{match } \text{eval}(e, \sigma) \textsf{ with} ~\vert~ \bfalse \Rightarrow c_1(\sigma) \\
&\,\,\,\quad\qquad\qquad\qquad\qquad\qquad\qquad\vert~ \btrue \Rightarrow c_2(\sigma)).
\end{align*}
They are exactly the same function in the meta-language and will satisfy the same Hoare triple.
Proofs of structural rules and inversion rules for a shallowly embedded language is similar to their proofs in a big-step based shallowly embedded program logic (section~\ref{sec:embed-big}) since both uses big-step semantics to define the semantics.
In general, deeply embedded programming languages are more interesting and we will mainly discuss it in this paper.

\subsection{Embeddings of Assertion Languages}
To reason about a program's effect on program states, a Hoare-style program logic uses an assertion language to describe program state.
A shallowly embedded assertion is a predicate in the meta-logic which directly defines the set of states.
A deeply embedded assertion language specifies syntax trees of assertions first and then defines how to interpret assertions as sets of states.
For example, figure~\ref{fig:setting} defines a deep embedding of an assertion language.
And figure~\ref{fig:shallow-assn} defines a shallow embedding of the assertion language (only a snippet of it), which is also the interpretation for the deeply embedded one.
\begin{figure}[h]
$$
\begin{array}{c}
    \top := \lambda \sigma.\, \text{True} \ \ \ \
    \bot := \lambda \sigma.\, \text{False} \ \ \ \
    P \land Q := (\lambda \sigma.\, \sigma \vDash P \text{ and } \sigma \vDash Q)\\
    \doublebrackets{e} = v := (\lambda \sigma.\, \text{eval}(e, \sigma) = v) \ \ \ \
    P[x\mapsto e] := (\lambda \sigma.\, \sigma \vDash P[x/\text{eval}(e, \sigma)]) \ \ \cdots
\end{array}
$$
\caption{The shallowly embedded assertion language}
\label{fig:shallow-assn}
\end{figure}


Unlike programming languages, we observe that shallowly embedded assertion languages are generally preferred by program verification framework developers, as table~\ref{fig:table-choice} indicates that almost all existing works rely on an shallow assertion language.
The reason is that the shallowly embedded assertion often has more expressiveness.
For one fixed deeply embedded assertion language, we need to define its interpretation in the meta-logic.
This interpretation is a meta-logic function from assertions' syntax trees to sets of program states, i.e., $\doublebrackets{\cdot}\in \text{Assertion}_{\text{deep}} \rightarrow \text{state} \rightarrow \text{Prop}$.
Thus, the expressiveness of this deep embedding is no stronger than the expressiveness of ``sets of states'', as which the assertion language is formalized in a shallow embedding, i.e., $\text{Assertion}_{\text{shallow}} \triangleq \text{state} \rightarrow \text{Prop}$.

The cause of different embedding choices for the programming language and the assertion language is the difference in the object they describe.
Assertion languages describe static objects, i.e. an assertion is always used to determine a set of program states.
Thus, it is suitable to directly model assertions as sets.
Comparably, programming languages describe dynamic objects, transitions of program states.
Verification projects for different real programming languages need to embody different features of those languages, e.g., concurrency.
However, there exists no simple shallow embedding of a programming language that can describe these features concisely, e.g., we have mentioned in previous part that denotational shallow embeddings are difficult to express concurrency.
Therefore, formalizing different programming languages' syntax in deep embeddings is preferred than directly using shallow embeddings.

In this paper, we will stick to the shallow embedding of assertion languages defined in figure~\ref{fig:shallow-assn}.

\subsection{Embeddings of Program Logics}
\label{sec:nomen-logic}
When formalizing logics\footnote{In this paper, we mainly discuss the embedding of ``program logic'' and will use ``logic'' for short without ambiguity.} (e.g. the propositional logic or a Hoare logic), a shallowly embedded logic defines a statement to be \textit{valid} directly using semantics.
We mainly consider the partial correctness in this paper, and one common definition of shallowly embedded Hoare triple is: $\vDash \triple{P}{c}{Q}$ iff. for any initial state $\sigma_1$, if $\sigma_1$ satisfies assertion $P$, then the execution of $c$ from $\sigma_1$ does not cause error; and for and any ending state $\sigma_2$ reachable from $\sigma_1$ in the execution of $c$, $\sigma_2$ satisfy $Q$.
Verification tool developers can prove many properties about valid triples and use them as ``proof rules'' to assist users in program verification.
In practice, there are different ways to shallowly embed a program logic, which we discuss in section~\ref{sec:embed-big}, section~\ref{sec:embed-small}, section~\ref{sec:embed-cont}.


In contrast, a deeply embedded logic formalizes the proof tree inductively by giving admitted proof rules, and we say a statement is \textit{provable} under the logic, denoted by $\vdash \triple{P}{c}{Q}$, iff. it can be constructed from these proof rules.
Users of a verification tool can use these proof rules to build a proof tree of their program specifications.

To ensure proof rules given in a deeply embedded logic are consistent with program behaviors,
one needs to additionally prove the logic sound, that is, every \textit{provable} statement is also \textit{valid}, i.e., for any statement $S$, if $\vdash S$ then $\vDash S$.
In simple and common cases, this soundness theorem can be proved by induction over proof trees, i.e. it suffices to prove that every proof rule will always generate valid Hoare triples from valid triples. Using the sequential rule (\textsc{hoare-seq}) as an example,
$$
\inferrule*[left=hoare-seq]{
    \vdash \{P\} c_1 \{Q\} \\ \vdash \{Q\} c_2 \{R\}
}{
    \vdash \{P\} c_1 \cmdseq c_2 \{R\} 
}
$$
the induction step is to prove:
\begin{eqnarray}
\vDash \{P\} c_1 \{Q\} \text{ and } \vDash \{Q\} c_2 \{R\} \text{ imply }\vDash \{P\} c_1 \cmdseq c_2 \{R\} \label{seq_sound}
\end{eqnarray}
In comparison, there is no counterpart of this soundness property when using a shallowly embedded logic.
To use this sequential rule in a shallowly embedded logic, one need to directly prove property (\ref{seq_sound}).
This fact implies that the implementation of a sound deeply embedded program logic always accompanies an underlying shallowly embedded one.

In some nontrivial cases, one has to introduce an auxiliary validity definition $\VDash S $ in order to prove the soundness of a logic in two steps: (1) for any triple $S$, prove $\vdash S$ implies $\VDash S$ by induction over the proof tree; (2) show that $\VDash S$ does imply $\vDash S$ by semantic analysis.
For example, Brookes's concurrent separation logic soundness proof \cite{brookes2007semantics} uses this technique. 
In comparison, a shallow embedding strategy will directly formalize the auxiliary validity, and establish ``proof rules'' based on it. 
The fact that $\VDash S$ implies $\vDash S$ is still necessary and should be proved separately, and it is also called ``the adequacy property'' \cite{jung2018iris} in some literature.
We discuss this soundness proof technique later in section~\ref{sec:discuss-sound}.

In conclusion, we can choose arbitrary combinations of shallow/deep embeddings among the programming language, the assertion language, and the program logic to instantiate a foundational verification tool.
In this paper, we mainly focus on a deeply embedded programming language and a shallowly embedded assertion language, but consider different embeddings of program logics which we will discuss soon in section~\ref{sec:embed}.




\section{Different Embeddings of Program Logics}
\label{sec:embed}

As we have mentioned, there exists different ways to embed a program logic.
In this section, we will present four mainstream embeddings of Hoare logic from our survey of existing verification projects: a deep embedding and three different shallow embeddings.


\

Based on this language, this section then use three different shallow embeddings (section~\ref{sec:embed-big}, section~\ref{sec:embed-small}, section~\ref{sec:embed-cont}) and a deep embedding (section~\ref{sec:embed-deep}) to formalize program logics.
For each embedding method, we also demonstrate which existing Hoare-logic-based verification projects and how do their program logic fit into the category as table~\ref{fig:table-choice} shows.
Under each formalization, the chosen primary rules in figure~\ref{fig:setting} are sound.

Meanwhile, there also exists other non-Hoare-logic based verification approaches and some Hoare-logic based framework may also employ non-Hoare-logic based reasoning styles, which we discuss in section~\ref{sec:embed-other}.

Verification projects we review in this section all have certain unique features and complex mechanisms and it is difficult to cover all of them, e.g., verification supports for concurrency.
Some also do not support control flow reasoning and we do not discuss how to extend them with control flows to support our toy language.
We only discuss some basics of their embeddings and related features that can classify them into each category.

\subsection{Big-step (BigS.) based Shallow Embedding}
\label{sec:embed-big}


\Cref{fig:note-bigstep} describes notations defining the big-step semantics of the While-CF language, where $(c, \sigma_1) \Downarrow (\text{ek}, \sigma_2)$ states that from program state $\sigma_1$, the program $c$ may terminate with exit kind $\ek$, which could be normal exit, denoted by $\epsilon$, or break, continue exit, and the state will be changed to $\sigma_2$.
We use $(c, \sigma) \Uparrow$ to denote that an error would occur for the execution of $c$ from state $\sigma$.
The big-step semantics is defined recursively, e.g., the semantics of the sequencing command is \textsc{Seq1} and \textsc{Seq2}.
For the sake of space, we put the full definition (which is standard) in appendix~\ref{sec:Abigs}.

\begin{figure}[ht]
\centering
$\ek \in \exitkind := \epsilon \lsep \ekb \lsep \ekc$ \ \ $\sigma \in \text{state}$

Big-step Relation: $(c, \sigma_1) \Downarrow (\text{ek}, \sigma_2)$\ \ \ \
Error: $(c, \sigma) \Uparrow$
\begin{mathpar}
    \inferrule*[left=Seq1]{
        (c_1, \sigma_1) \Downarrow (\epsilon, \sigma_3) \\\\
        (c_2, \sigma_3) \Downarrow (\ek, \sigma_2)
    }{
        (c_1 \cseq c_2, \sigma_1) \Downarrow (\ek, \sigma_2)
    }
    \and
    \inferrule*[left=Seq2]{
        (c_1, \sigma_1) \Downarrow (\ek, \sigma_2)
    }{
        (c_1 \cseq c_2, \sigma_1) \Downarrow (\ek, \sigma_2)
    }
\end{mathpar}
\caption{Notations for Big-step semantics}
\label{fig:note-bigstep}
\end{figure}



Based on this big-step semantics, a triple is defined to be valid, \linebreak
$\bigvalid \triple{P}{c}{Q, [R_\ekb, R_\ekc]}$,
iff. for any $\sigma_1$ satisfying precondition $P$,
(1)~the execution of $c$ from $\sigma_1$ is safe and does not cause error, and
(2)~if the execution terminates, the ending program state satisfies the corresponding post-condition.
$$
\begin{array}{rcl}
\bigvalid \triple{P}{c}{Q, [R_\ekb, R_\ekc]} & \text{iff.} & \text{for all $\sigma_1 \vDash P$, } \lnot\,(c, \sigma_1)\Uparrow\\
&& \text{and for all } \ek, \sigma_2, \text{if }(c, \sigma_1) \Downarrow (\ek, \sigma_2) \\
&& \quad\text{then }\ek = \epsilon \text{ implies } \sigma_2 \vDash Q\\
&& \quad\text{and }\ek = \ekb \text{ implies } \sigma_2 \vDash R_\ekb\\
&& \quad\text{and }\ek = \ekc \text{ implies } \sigma_2 \vDash R_\ekc
\end{array}
$$


\paragraph{Related Projects.} \


Klein \textit{et al.} \cite{bulwahn2008imperative} formalizes
an imperative functional programming language in Isabelle/HOL, which is a shallowly embedded one using the state monad in HOL.
Based on this language, Lammich \cite{lammich2012separation} builds a logic in big-step based shallow embedding in Isabelle/HOL.
The logic embedding is almost identical to the one discussed above, but they do not consider control flows and only require the ending state of the normal exit to satisfy the post-condition.
Thus their definition does not have the last two conjuncts above.
Based on Lammich's logic \cite{lammich2012separation}, Zhan \cite{zhan2018verifying} verifies imperative implementations of some data structures in Isabelle/HOL using its auto2 prover \cite{zhan2016auto2}.
Nipkow's Hoare logic in Isabelle/HOL \cite{nipkow2002hoare,nipkow2002isabelle} also uses big-step embedding.
Many verification has been performed based on these logics in Isabelle, e.g., Lammich and Nipkow \cite{lammich2019proof} proves the correctness of priority search tree and Prim's and Dijkstra's algorithm in Isabelle.

In Cook's famous Hoare logic's soundness and completeness proof \cite{cook1978soundness}, it uses big-step based shallow embedding as the logic's validity definition and proves many properties including logic's soundness and completeness w.r.t. it.

Software foundation \cite{SF} is a famous textbook for teaching Coq formalization.
Its second volume introduces Hoare logics both in the big-step based embedding and the deep embedding, where the former is also the validity definition of the latter.
Its sixth volume introduces the separation logic built with the big-step based embedding.
The simplicity of the big-step based embedding makes it a good introduction of the Hoare logic and the separation logic for beginners.

\subsection{Weakest Precondition (WP) based Shallow Embedding}
\label{sec:embed-small}

Another shallow embedding method is to embed the logic using the weakest precondition defined by small-step semantics.
We want to emphasize that the weakest precondition here is directly defined by semantics instead of an encoding by some existing Hoare logic.

We use $(c, \kappa, \sigma) \rightarrow_c (c', \kappa', \sigma')$ to describe a small-step in command reduction and $\rightarrow_c^*$ to describe the reflexive transitive closure of $\rightarrow_c$.
The small-step reduction is a binary relation between triples of the focused term\footnote{The focused term is the next command to execute and is so named in some literature.} $c$, the continuation (control stack) $\kappa$, and the program state $\sigma$.
The control flow commands makes it slightly different from text book definitions.
We mainly follows CompCert Clight's definition style and list most important semantic rules in figure~\ref{fig:note-smallstep}.

\begin{description}
    \item[\eqref{eq:for-intro}]
    When the focused term begins with a for loop, it pushes the loop into the control flow stack and loads the loop body (followed by \ccontinue) into the focused term.
    \item[\eqref{eq:skip-for}]
    When the increment step for previous iteration finishes, it loads the loop body for next iteration\footnote{$\kloops{c_1}{c_2}{1}$ means in-loop-body, and $\kloops{c_1}{c_2}{2}$ means in-increment-step.}.
    \item[\eqref{eq:continue-for}] When the focused term reduces to $\ccontinue$, it loads the increment step $c_2$ and updates the innermost loop's continuation to Kloop$_2$.
    \item[\eqref{eq:break-for}] When the focused term reduces to $\cbreak$, it pops the innermost loop and sets the program to $\cskip$ to continue execute the remaining control stack.
    \item[\eqref{eq:continue-seq}, \eqref{eq:break-seq}] When the focused term reduces to $\ccontinue$ or $\cbreak$, it will keep skipping KSeq continuations.
    As a result, if an execution from $(c,\kappa,\sigma)$ terminates, it must terminate at $(\cskip, \epsilon, \_)$, $(\cbreak, \epsilon, \_)$, or $(\ccontinue, \epsilon, \_)$.
\end{description}

\setcounter{equation}{0}
\begin{figure}[ht]
\centering
$
\kappa \in \text{continuation} := \epsilon \lsep \kseq{c} \cdot \kappa \lsep \kloops{c_1}{c_2}{1} \cdot \kappa \lsep \kloops{c_1}{c_2}{2} \cdot \kappa
$
\begin{align}
((\cfor{c_1}{c_2}), \kappa, \sigma) &\rightarrow_c
(c_1 \cseq \ccontinue, \kloops{c_1}{c_2}{1} \cdot \kappa, \sigma) \label{eq:for-intro} 
\\
(\cskip, \kloops{c_1}{c_2}{2} \cdot \kappa, \sigma) &\rightarrow_c
(c_1 \cseq \ccontinue, \kloops{c_1}{c_2}{1} \cdot \kappa, \sigma)
\label{eq:skip-for} 
\\
(\ccontinue, \kloops{c_1}{c_2}{1} \cdot \kappa, \sigma) &\rightarrow_c
(c_2, \kloops{c_1}{c_2}{2} \cdot \kappa, \sigma) \label{eq:continue-for} 
\\
(\cbreak, \kloops{c_1}{c_2}{1} \cdot \kappa, \sigma) &\rightarrow_c
(\cskip, \kappa, \sigma) \label{eq:break-for} 
\\
(\ccontinue, \kseq{c} \cdot \kappa, \sigma) &\rightarrow_c
(\ccontinue, \kappa, \sigma) \label{eq:continue-seq} \\
(\cbreak, \kseq{c} \cdot \kappa, \sigma) &\rightarrow_c
(\cbreak, \kappa, \sigma) \label{eq:break-seq}
\end{align}
\caption{Notations for Small-step semantics}
\label{fig:note-smallstep}
\end{figure}

Shallowly embedded weakest pre-condition can be defined on this small-step semantics.
We use $\sigma~\vDash~\textsf{WP} (c, \kappa) \progspec{Q, [\vec{R}]}$ to denote that a program state $\sigma$ satisfies the weakest pre-condition for the focused program $c$ and the continuation $\kappa$ with post-conditions $Q$ and $[\vec{R}]$. The $\textsf{WP}$ is defined as the largest set such that $\sigma~\vDash~\textsf{WP} (c, \kappa) \progspec{Q, [\vec{R}]}$ iff.
\begin{itemize}
\item \textbf{Terminal Case:} $\kappa$ is $\epsilon$ and (1) $c = \cskip$ and $\sigma \vDash Q$ or (2) $c = \cbreak$ and $\sigma \vDash R_\ekb$ or (3) $c = \ccontinue$ and $\sigma \vDash R_\ekc$;
\item \textbf{Preservation Case:} Or $(c, \kappa, \sigma)$ can be further reduced by $\rightarrow_c$. And for any $c', \kappa', \sigma'$ such that $(c, \kappa, \sigma) \rightarrow_c (c', \kappa', \sigma')$, it still has $\sigma' \vDash \textsf{WP} (c', \kappa') \progspec{Q, [\vec{R}]}$.
\end{itemize}
The terminal case ensures that the ending program state will satisfy post-conditions, and the preservation case guarantees that it can always step forward and eventually reach an ending state satisfying post-conditions.

A triple is defined to be valid under the weakest precondition based shallow embedding, $\vDash_w \triple{P}{c}{Q}$, iff. for any state $\sigma$ satisfying precondition $P$, we have $\sigma \vDash \textsf{WP} (c, \epsilon) \progspec{Q, [\vec{R}]}$.


\paragraph{Remark.}
The definition of the co-inductive weakest precondition above uses the Tarski fixed point theorem \cite{tarski1955lattice}.
\begin{thm}[Tarski fixed point theorem]
    For a complete lattice $(L, \leq)$ and a monotone function $f: L \rightarrow L$, the set of all fixed points of $f$ is also a complete lattice with $\sup\{x\in L \,\vert\, x \leq f(x)\}$ as the greatest fixed point.
\end{thm}
We can also define it using the Bourbaki-Witt fixed point theorem \cite{witt1950beweisstudien,bourbaki1949theoreme}, which is used by many existing frameworks to define the weakest precondition based embedding and is used in our Coq formalization.
\begin{thm}[Bourbaki-Witt fixed point theorem]
    If $(X, \leq)$ is a non-empty complete ordered chain, and $f: X \rightarrow X$ satisfies $f(x) \geq x$ forall $x$, then $f$ has a fixed point.
    For some $x_0 \in X$, let $g: \mathbb{N} \rightarrow X$ be a function with
    $$
    g(0) = x_0 \qquad g(n+1) = f(g(n)),
    $$
    then $\lim_{n\rightarrow \infty}g(n)$ is a fixed point.
\end{thm}
Below is a formal definition of the weakest pre-condition using the Bourbaki-Witt greatest fixed point.
We use an ordinal number $n$ to bound the number of small-step reductions until termination
and the limit of $\WP(n, c, \kappa)$ when $n$ goes to infinite defines the weakest precondition.
\begin{align}
\sigma \vDash \WP(n, c, \kappa)\progspec{Q,[\vec{R}]}
&\text{ iff. }
\left(\begin{gathered}
    \kappa = \epsilon \land n = 0 \land\\
    \left(\begin{gathered}
        (c = \cskip \land \sigma \vDash Q) \lor \\
        (c = \cbreak \land \sigma \vDash R_\ekb) \lor \\
        (c = \ccontinue \land \sigma \vDash R_\ekc)
    \end{gathered}\right)
\end{gathered}
\right) \\
& \qquad \lor
\left(\begin{gathered}
n > 0 \land
\text{reducible}(c, \kappa, \sigma) \\
\land \forall c',\kappa',\sigma'. (c,\kappa,\sigma)\rightarrow_c(c',\kappa',\sigma') \\
\Rightarrow \sigma' \vDash \WP(n-1, c', \kappa')\progspec{Q,[\vec{R}]}
\end{gathered}\right)
\notag
\\
\sigma \vDash \WP(c, \kappa)\progspec{Q,[\vec{R}]}
&\text{ iff. }
\forall n \in \mathbb{N}. \sigma \vDash \WP(n, c, \kappa)\progspec{Q,[\vec{R}]}
\end{align}

Both definition of the weakest precondition using the greatest fixed point are just approach to define a co-inductive type.
Although we can directly write a co-inductive definition in Coq, using the Bourbaki-Witt fixed point approach makes proofs easier and thus we use it in our Coq formalization.
We use the Tarski fixed point approach in the paper (here and in section~\ref{sec:embed-cont} for $\text{safe}(c,\kappa,\sigma)$) for a clear presentation.

If the reader is familiar with the step-index technique \cite{appel2001indexed}, they may find that a meta-logic equipped with step-indices (e.g., Iris \cite{jung2018iris}) can easily define the weakest precondition as the Bourbaki-Witt fixed point.
The step-index can be used as the ordinal number $n$ and there are lots of infrastructures (e.g., the later modality $\later$) to support the definition.
But in general, step-indexing is not necessary for the shallowly embedded weakest precondition as we have shown here and in our Coq formalization.
It is necessary only when the definition involves recursive types that are not inductive.
For example \cite{jung2018iris}, we cannot give an inductive definition for the type iProp below, since iProp has non strictly positive occurrence in this definition.
$$
\text{iProp} := \text{Res} \rightarrow \text{Prop} \qquad \text{where } \text{Res} := F(\text{iProp}) \text{ is parameterized by iProp}
$$
If the state model and the language semantics does not rely on these non-inductive types, then the weakest precondition's definition obviously needs no step-index.
We refer interested readers to Appel and McAllester's work on the step-indexed model \cite{appel2001indexed} for more details about the step-indexing technique.


\paragraph{Related Projects.}

Fine-grained Concurrent Separation Logic (FCSL) \cite{FCSL-mechanized} is a framework for verifying concurrent programs.
It uses mixed embedding for its programming language and a weakest precondition based shallow embedding for its program logic.
Its program is defined by an action tree, where each edge, the atomic action, is defined by a transition between state, i.e., a shallow embedding.
In the action tree's syntax definition below, $\omega$ is a leaf node indicating the divergence of a thread and $\text{ret } v$ is a leaf node indicating the termination of a thread with return value $v$.
The concatenation $a :: k$ of the action $a$ and the context $k$ is the sequential execution $a$ and $k$ with $a$'s result as arguments.
Similarly, parallel composition $(t_1 \| t_2) :: k$ first interleaves executions of tree $t_1$ and $t_2$ and passes their results to the context $k$.
$$
\begin{array}{rcl}
    t, t_1, t_2 \in \text{Tree} &:=& \omega \lsep \text{ret } v
    \lsep a :: (k: \text{Val}\rightarrow\text{Tree})\\
    &\lsep& (t_1 \| t_2) :: (k: \text{Val} \times \text{Val} \rightarrow\text{Tree}) \\
    a \in \text{Action} &:=& \text{state}\rightarrow\text{state}\rightarrow\text{Prop}
\end{array}
$$

The definition below is FCSL's \cite{FCSL-mechanized} embedding of their Hoare triple (we omit details about concurrency).
The always predicate \eqref{eq:fcsl-always-k} is similar to our weakest precondition's definition, which asserts that the memory safety at each step and at each step the assertion $P$ should hold for the program state and the action tree.
This assertion of always predicate in the Hoare triple \eqref{eq:fcsl-triple} ensures that the post-condition holds at leaf nodes.
It also uses the Bourbaki-Witt fixed point approach to define the weakest precondition
The $k$ in \eqref{eq:fcsl-always-k}\eqref{eq:fcsl-always} serves the purpose of an ordinal number.
\begin{align}
\text{always}^{k}\,\sigma\,t\,Q &\text{ iff. } \text{memsafe}(\sigma, t) \land Q(\sigma,t) \land \label{eq:fcsl-always-k} \\
&\quad\quad \forall \sigma', t'.\, (k > 0 \land ((\sigma, t)\rightarrow (\sigma', t'))) \Rightarrow \text{always}^{k-1}\,\sigma'\,t'\,Q \notag \\
\text{always}\,\sigma\,t\,Q &\text{ iff. } \forall k \in \mathbb{N}.\, \text{always}^{k}\,\sigma\,t\,Q \label{eq:fcsl-always}\\
\vDash \triple{P}{c}{Q} &\text{ iff. } \forall \sigma.\, \sigma \vDash P \Rightarrow \text{always}\,\sigma\,t\,(\lambda \sigma, t.\, \forall v.\, t = \text{Ret }v \Rightarrow Q(\sigma,v)) \label{eq:fcsl-triple}
\end{align}

Iris is a higher order concurrent separation logic for verifying correctness of functional programs.
It uses deeply embedded $\lambda$-calculus-like programming languages and the weakest precondition based embedding.
When concurrency is not involved (section 6.3.2 of \cite{jung2018iris}), Iris embeds their logic by weakest pre-conditions below as a separation logic proposition in Iris (iProp).
\begin{align}
    \WPRE\, e\, \{\Phi\}&\triangleq
           (e \in \textit{Val} \land \Phi(e))
           \label{eq:iris-term} \\
    &\lor \bigl(\forall \sigma.\, e \notin \textit{Val} \land S(\sigma) \wand \bigl(\text{reducible}(e, \sigma) \notag \\
    &\quad \land \later
        \forall e', \sigma'.\, \left((e, \sigma) \tred (e', \sigma')\right) \wand
    \left(S(\sigma') * \WPRE\, e'\, \{\Phi\}\right)
    \bigr)
    \bigr) \label{eq:iris-pres}
\end{align}
The first disjunct \eqref{eq:iris-term} asserts that if the expression $e$ is now a value, then the evaluation has terminated and the program state and the evaluation result $e$ should satisfy $\Phi$.
The second disjunct \eqref{eq:iris-pres} specifies behaviors when $e$ is not an terminal and should be further reduced.
$S(\sigma)$ injects program state $\sigma$ into an iProp.
The disjunct \eqref{eq:iris-pres} will first consume a piece of memory injected by $\sigma$ and asserts that expression $e$ is reducible with this piece of memory ($\text{reducible}(e, \sigma)$).
And for any new expressions $e'$ and $\sigma'$ it can reduce to by small-step semantics $\rightarrow_t$, we can still have the memory $S(\sigma')$ it reduces to and the new expression $e'$ still preserves the weakest pre-condition.
Iris's propositions iProp uses the step-index to solve the circularity when defining higher-ordered ghost states, which happens to be a suitable choice of the ordinal number in the definition of the co-inductive weakest precondition.
The later modality $\later$ here reduces the step-index (the ordinal number) in the weakest precondition $\WPRE\, e'\, \{\Phi\}$ at the next step.


Iris's Hoare triple\footnote{
It is worth mentioning that Iris also provides mechanisms to directly reason about weakest preconditions, which makes it less like a Hoare-logic based framework.
We discuss the weakest precondition based verification method in section~\ref{sec:iris}.
} $\triple{P}{e}{\Phi}$ is embedded as $P \wand \WPRE\,e\,\{\Phi\}$.
We can observe that Iris's embedding is very similar to our prototypical weakest precondition based embedding, and in section~\ref{sec:project-iris-shallow}, we will use Iris as an example under the functional program setting to demonstrate our research into extended proof rules.

FCSL and Iris both support concurrent program verifications through mechanisms of angelic updates (view shifts in Iris).
Iris is also extended to support prophecy variables \cite{jung2019future}.
However, these features are out of this paper's scope.
For simplicity, definitions in this section only consider sequential programs and remove features supporting concurrency.
For the complete definition, readers may refer to their original papers \cite{FCSL-mechanized,jung2018iris}.


\subsection{Continuation (Cont.) based Shallow Embedding}
\label{sec:embed-cont}

Continuation based shallow embedding defines the Hoare triple through Hoare tuple $\progspec{P} (c, k)$, pronounced $P$ guards $(c, k)$, which asserts that program $c$ and continuation $k$ can safely executes from program states satisfying the pre-condition $P$.
A triple $\triple{P}{c}{Q, [\vec{R}]}$ is valid iff. for arbitrary continuation $\kappa$,
$$
\text{if } \quad\begin{cases}
    \quad\guard{Q}{\cskip}{\kappa} &\\
    \quad\guard{R_\ekb}{\cbreak}{\kappa} &\\
    \quad\guard{R_\ekc}{\ccontinue}{\kappa} &
\end{cases}
\text{then } \quad \guard{P}{c}{\kappa}.
$$
It states that for any continuation that safely executes from different post-conditions with corresponding exit kind, it should safely executes after program $c$'s execution from the pre-condition.

An assertion $P$ guards the execution of a program $c$ and a continuation $\kappa$, $\guard{P}{c}{\kappa}$, iff. for any program state $\sigma$ satisfying the pre-condition $P$, we have safe($c, \kappa, \sigma$), which is the largest set of configurations with following properties.
\begin{itemize}
    \item \textbf{Terminal Case:} $\kappa$ is $\epsilon$, and $c$ is \cskip, \cbreak, or \ccontinue;
    \item \textbf{Preservation Case:} Or $(c, \kappa, \sigma)$ can be further reduced by $\rightarrow_c$, and for any $(c', \kappa', \sigma')$ it reduces to, we should have safe$(c', \kappa', \sigma')$.
\end{itemize}

\paragraph{Related Projects.}

The original shallowly embedded VST \cite{VST} uses a deeply embedded programming language and a program logic in continuation based shallow embedding.
The Hoare triples in shallowly embedded VST extends our definition with function invocations and concurrent separation logic.
Besides these features, which we discuss later in section~\ref{sec:project-vst}, its embedding is almost identical to the prototypical definition above.

\subsection{Deep Embeddings}
\label{sec:embed-deep}

As clarified in section~\ref{sec:nomen}, a deep embedding of a program logic is a proof system with inductively defined proof trees.
Different deep embeddings of program logics with fixed languages and forms of judgement have different sets of admitted proof rules.
For example, we may consider a proof system with primary proof rules from figure~\ref{fig:setting} as an example deep embedding.
In Coq, it is easy to define the deep embedding with an inductive relation shown below.
Here, \texttt{provable P c Q Rb Rc} stands for the Hoare triple $\vdash \triple{P}{c}{Q, [R_b, R_c]}$.
\begin{lstlisting}[language=Coq]
Inductive provable: Assertion -> com -> Assertion (* normal post *) -> Assertion (* break post *) -> Assertion (* continue post *) -> Prop :=
| hoare_skip: forall P, provable P $\cskip$ P bot bot
| hoare_break: forall P, provable P $\cbreak$ bot P bot
| hoare_continue: forall P, provable P $\ccontinue$ bot bot P
| hoare_seq: forall c1 c2 P P' Q Rb Rc,
    provable P c1 P' Rb Rc -> provable P' c2 Q Rb Rc ->
    provable P (c1$\cseq$c2) Q Rb Rc
| hoare_loop: forall c1 c2 P I Q,
    provability P c1 I Q I -> provable I c2 P Q bot ->
    provable P ($\cfor{c_1}{c_2}$) Q bot bot
| $\cdots$.
\end{lstlisting}

An arbitrary set of admitted proof rules may not be correct, and we need to prove the soundness of such proof system.
To do so, we need to define the validity of a logic judgement, i.e., the corresponding shallowly embedded logic.
We can prove the soundness by showing that each proof rule is valid in the shallow embedding and the shallowly embedded logic is sound.
We can easily prove primary proof rules in figure~\ref{fig:setting} to be valid under all three shallow embeddings in section~\ref{sec:embed-big}, section~\ref{sec:embed-small}, section~\ref{sec:embed-cont}.

\paragraph{Related Projects.} \
Simpl \cite{schirmer2006verification} is a tool in Isabelle/HOL for verifying sequential programs, which has been used to verify seL4 code \cite{klein2009sel4}.
It uses deeply embedded programming language and program logic.
Its judgement is $\vdash\triple{P}{c}{Q,Q_\ek}$ with the following validity definition.
$$
\begin{array}{rcl}
\vDash \triple{P}{c}{Q, Q_\ek} &\triangleq& \forall \sigma_1, \sigma_2.\, \sigma_1 \in \text{Normal } P \rightarrow\\
&& \tuple{c, \sigma_1} \Downarrow \sigma_2 \rightarrow \sigma_2 \not\in \text{Fault} \rightarrow \\
&& \sigma_2 \in \text{Normal }Q \cup \text{ Abrupt }Q_\ek
\end{array}
$$
It says after the execution from program states satisfying the precondition, if the ending state is not a fault one, then it satisfies the normal and abrupt post-conditions depending on its exit kind.

CSimpl \cite{sanan2017csimpl} is a framework for concurrent program verification based on rely guarantee reasoning.
Its programming language is a deeply embedded imperative language supporting concurrency.
Their logic is deeply embedded by giving inference rules of the logic with judgement $R, G \vdash \triple{P}{c}{Q, Q_\ek}$, where $P$ is pre-condition and $Q, Q_\ek$ are normal and control flow post-conditions respectively.
They choose to use rely-guarantee method to reason about concurrency, and use $R, G$ to specify the rely and guarantee of a program.
The judgement specifies how the environment can modify ($R$) and what a program can do ($G$) to the shared resource.
To show the soundness of their proof system, they define judgement's validity in three steps.
\begin{itemize}
    \item They use $\text{assum}(c, P, R)$ to denote a set of small-step reduction steps of $c$ (lists of intermediate commands and states) from pre-condition $P$ and under environment interference $R$.
    \item They use $\text{comm}(G, Q, Q_\ek)$ to denote a set of reduction steps that obey guarantee $G$ and have terminal state satisfying post-conditions $Q, a$.
    \item The judgement $R, G \vdash \triple{P}{c}{Q, Q_\ek}$ is valid iff. $\text{assum}(c, P, R)$ is a subset of $\text{comm}(G, Q, Q_\ek)$. The inclusion implies any trace of $c$ from pre-condition $P$ and under rely $R$, will terminate in post-conditions $Q, Q_\ek$ and generate guarantee $G$.
\end{itemize}
The definition of the set $\text{comm}(G, Q, Q_\ek)$ is almost identical to the weakest precondition in section~\ref{sec:embed-small} but with extra guarantee constraints on the program's behaviors.
In conclusion, CSimpl uses deeply embedded logic with weakest precondition based shallow embedding as its validity.

Certified assembly program (CAP) \cite{yu2003building} and XCAP \cite{ni2006certified}  use deeply embedded language and deeply embedded logic with judgement $\vdash \{P\}\, c$ to verify assembly programs.
The logic guarantees that an assembly program $c$ can execute safely if the program state satisfies the pre-condition $P$.
They formalize the logic with a set of inference rules as its deep embedding and prove the soundness under the continuation-based shallow embedding.
Their objective is to verify program's safety instead of functional correctness, therefore the judgement and logic looks different from the one in our paper.

Xu \textit{et al.} \cite{xu2016practical} develop $\mu$C, a framework for verifying preemptive operating systems.
The framework uses deeply embedded program logic.
The logic is proven sound w.r.t. a weakest precondition based shallow embedding, which is similar to those in section~\refeq{sec:embed-small}.

Software foundation \cite{SF} introduces a toy example of a deeply embedded Hoare logic with big-step based shallow embedding as its validity, as we have mentioned in section~\ref{sec:embed-big}.






\subsection{Other Logic based Verification Methods}
\label{sec:embed-other}

So far, we have discuss four Hoare logic embeddings and foundational verification frameworks that use these embeddings.
Most of these frameworks use Hoare triples as program specifications and prove them mainly by applying proof rules to Hoare triples, which we refer to as Hoare-logic-based verification.
Nevertheless, there exists other logic based verification methods that do not use Hoare triples and Hoare logics as their primary tools.
And some verification frameworks that encodes Hoare triples in their logic may feature verification technique that is non Hoare-logic-based.

In this section, we will briefly review some non-Hoare-logic-based verification techniques: Hoare monad and Dijkstra monad in section~\ref{sec:monad}, Iris (weakest precondition based verification) in section~\ref{sec:iris}, and Characteristic Formulae in section~\ref{sec:cf}.
It is not clear whether these approaches is better than pure Hoare-logic based ones, and the comparison among them is not the focus of this paper.



\subsubsection{Hoare Monad \& Dijkstra Monad}
\label{sec:monad}

Instead of Hoare logic, we can also use Hoare monad and Dijkstra monad to assert and verify the correctness of a program.
Based on their Hoare type theory \cite{nanevski2008hoare}, Nanevski \textit{et al.} \cite{nanevski2008ynot} type a program using the Hoare monad $\texttt{ST}\, P\, A\, Q$, where the type of the program asserts the its execution from state satisfying the pre-condition $P$ will return with a value of type $A$, and the return value and the ending state will satisfy the post-condition $Q$ over both the return value and the program state.
The type of a large program is derived from types of statements that assembles it.
\begin{mathpar}
    \inferrule*[left=bind-hst]{
        \vdash e_1: (x:B) \rightarrow \texttt{ST}\,(R\,x)\,A\,Q \\
        \vdash e_2: \texttt{ST}\,P\,B\,R
    }{
        \vdash (e_1\,e_2): \texttt{ST}\,P\,A\,Q
    }
\end{mathpar}
For example, the rule above defines how to compose two Hoare monad through the bind operation (lambda function application), which is similar to the \textsc{hoare-seq} rule.
The type of $e_1\,e_2$ is derived by the type of $e_1$ and $e_2$.


The typing proof of Hoare monad cannot be easily automated due to some existential quantifiers over some intermediate program state.
A series of works \cite{swamy2013verifying, jacobs2015dijkstra, ahman2017dijkstra, maillard2019dijkstra} develops Dijkstra monad to improve the automation of such technique and allow reasoning about program with more features like exception, non-determinism, and IO.
Different from Hoare monad, Dijkstra monad types a program as $\texttt{WP\_ST}\,A\,\textsf{wp}$, a predicate transformer mapping a postcondition of the computation to its precondition, where $A$ is the return type and \textsf{wp} is a weakest precondition.
The weakest precondition \textsf{wp} is of the type
$$(A \times \text{state} \rightarrow \text{Prop}) \rightarrow \text{state} \rightarrow \text{Prop}$$
which is exactly a predicate transformer from the postcondition to the precondition.
Below shows the typing rule for the bind in Dijkstra monad.
$$
\inferrule*[left=bind-dst]{
    \vdash e_1:(x:B) \rightarrow \texttt{WP\_ST}\,A\,(\textsf{wp}_1\,x)\\
    \vdash e_2:\texttt{WP\_ST}\,B\,\textsf{wp}_2
}{
    \vdash (e_1\,e_2):\texttt{WP\_ST}\,A\,(\lambda Q\, s.\, \textsf{wp}_2\,(\lambda x\, s_1.\, \textsf{wp}_1\,x\,Q\,s_1)\,s)
}
$$
This rule composes predicate transformers $\textsf{wp}_1\,x$ (parameterized over the return value $x$) and $\textsf{wp}_2$ into a new one to type $e_1\,e_2$.

The typing of programs as Hoare monads and Dijkstra monads resembles the definition of the program's big-step semantics, as all of them are definitions by structural induction over the program's syntax.
Still using the bind operation as an example, \textsc{bind-bigs} defines the big-step semantics of $e_1\,e_2$.
$$
\inferrule*[left=bind-bigs]{
    (e_1\,x, \sigma_3) \Downarrow (v, \sigma_2)\\
    (e_2, \sigma_1) \Downarrow (x, \sigma_3)
}{
    (e_1\,e_2, \sigma_1) \Downarrow (v, \sigma_2)
}
$$
It is similar to \textsc{bind-hst} except that the relationship between the beginning state and ending state defined in \textsc{bind-bigs} becomes the relationship between the precondition and postcondition in \textsc{bind-hst}.
This similarity is even more obvious in \textsc{bind-dst}, where programs' effects on the pre/post-conditions are composed instead of composing their effects on the program state.




\subsubsection{Weakest Precondition based Verification}
\label{sec:iris}

In section~\ref{sec:embed-small}, we have embedded a Hoare triple using the weakest precondition: $\vDash_w \triple{P}{c}{Q}$, iff. for any state $\sigma$ satisfying precondition $P$, we have $\sigma \vDash \textsf{WP} (c, \epsilon) \progspec{Q, [\vec{R}]}$.
Although we will use the Hoare logic to reason about triples using this embedding, it is also possible to directly reason about the weakest precondition, which is the approach taken by Iris \cite{jung2018iris}.

We use the \textsc{if-seq} rule to demonstrate the weakest precondition based verification method and its advantage.
Suppose we want to prove $\triple{P}{(\cif{e}{c_1}{c_2}) \cseq c_3}{Q}$ in Iris for some $c_1, c_2, c_3, e$.
It is equivalent to prove that $P \wand \WP ((\cif{e}{c_1}{c_2}) \cseq c_3) \progspec{Q}$, which has the following weakest precondition derivation~\eqref{eq:ifseq-iris-1}.
\begin{equation}
\label{eq:ifseq-iris-1}
\inferrule*[right=wp-bind]{
    \inferrule*[right=wp-if, rightskip=3em]{
        (P*\doublebrackets{e}=\btrue) \wand \WP c_1 \progspec{\lambda v. \WP (v \cseq c_3) \progspec{Q}} \quad (\dagger) \\
        (P*\doublebrackets{e}=\bfalse) \wand \WP c_2 \progspec{\lambda v. \WP (v \cseq c_3) \progspec{Q}} \quad (\ddagger)
    }{
        P \wand \WP (\cif{e}{c_1}{c_2}) \progspec{\lambda v. \WP (v \cseq c_3) \progspec{Q}}
    }
}{
    P \wand \WP ((\cif{e}{c_1}{c_2}) \cseq c_3) \progspec{Q}
}
\end{equation}
We can first apply the \textsc{wp-bind} rule to hide the second part $c_3$ of the sequential composition into the post-condition as another weakest precondition.
Then, by \textsc{wp-if} rule, we need to prove the weakest preconditions of two branches, ($\dagger$) and ($\ddagger$), separately.
Then, take ($\dagger$) for example, we can further symbolically execute $c_1$ in Iris's weakest precondition calculus until it reaches a terminal value $v_1$ with the memory satisfying $P'$.
Then by \textsc{wp-val} and \textsc{wp-seq}, we can further reduce the proof goal into a weakest precondition only about the remaining $c_3$ and eventually reach a tautology after symbolic executions of $c_3$'s weakest precondition.
We can do a similar proof for the other branch ($\ddagger$) as well.
\begin{equation}
\label{eq:iris-ifseq}
\inferrule*[]{
    \inferrule*[]{
        \inferrule*[right=wp-val, rightskip=4em]{
            \inferrule*[right=wp-seq, rightskip=4em]{
                \inferrule*[]{\vdots}{
                    P' \wand \WP c_3 \progspec{Q}
                }
            }{
                P' \wand \WP (v_1 \cseq c_3) \progspec{Q}
            }
        }{
            P' \wand \WP v_1 \progspec{\lambda v. \WP (v \cseq c_3) \progspec{Q}}
        }
    }{
        \vdots
    }
}{
    (P*\doublebrackets{e}=\btrue) \wand \WP c_1 \progspec{\lambda v. \WP (v \cseq c_3) \progspec{Q}}
}
\end{equation}
As we can see, the proof of two branches ($\dagger$) and ($\ddagger$) is exactly the proof of two branches of $\cif{e}{c_1 \cseq c_3}{c_2 \cseq c_3}$.
The ability to put future computations (like $c_3$ in this example) into the post-condition and postpone their proofs essentially makes extended rules like \textsc{if-seq} for free.

Users of Iris mainly use this kind of symbolic execution for weakest precondition, instead of the symbolic execution for Hoare triples.
Therefore, most of our discussions about extended rules do not apply to the Iris framework.
In section~\ref{sec:project}, we use Iris's embedding (since they are relatively easy to use) of Hoare logic to demonstrate how our theory apply to some weakest precondition based embedding.
But it is only for demonstration, and if provers want to use Iris in reality, they may want to use Iris's weakest precondition calculus.


\subsubsection{Characteristic Formulae}
\label{sec:cf}

Chargu{\'e}raud \cite{chargueraud2010program, chargueraud2011characteristic, gueneau2017verified} uses characteristic formulae to verify programs and builds CFML above its logic for the verification of imperative Caml programs.
The main idea of characteristic formulae is to define a function $\textsf{cf}$ of type $\text{command} \rightarrow (\text{assertion} \rightarrow \text{assertion} \rightarrow \text{Prop})$.
The result $\textsf{cf}(c)$ defines the relationship between the precondition and the postcondition of command $c$'s Hoare triple.
For example, the result for the sequential composition is defined below.
\begin{equation}
\label{eq:cf}
\textsf{cf}(c_1 \cseq c_2) \triangleq \lambda P.\, \lambda Q.\, \exists R.\, \textsf{cf}(c_1)\,P\,R \land \textsf{cf}(c_2)\,R\,Q    
\end{equation}
This formalization is highly similar to the deeply embedded Hoare logic.
The proof tree for a program in a deeply embedded Hoare logic is built recursively according to the program's syntax tree.
While a characteristic formulae is also generated recursively according to the program' syntax tree in a similar fashion.
For example, the definition~\eqref{eq:cf} resembles the deeply embedded sequence rule below.
$$
\inferrule*[left=seq]{
    \vdash \triple{P}{c_1}{R}\\
    \vdash \triple{R}{c_2}{Q}
}{
    \vdash \triple{P}{c_1 \cseq c_2}{Q}
}
$$

After a characteristic formulae is generated, provers can then verify in the meta-logic whether it satisfies the specification they want to prove.
In this way, they can directly apply features in the meta-logic (usually the one directly supported by the interactive theorem prover like Coq) to do the proof.

On the other hand, the characteristic formulae generator \textsf{cf} is proved sound w.r.t. a big-step semantics,
\begin{equation}
\label{eq:cf-sound}
\forall c, P, Q.\,
\textsf{cf}(c)\,P\,Q \Rightarrow\,
\vDash_b \triple{P}{c}{Q}
\end{equation}
which makes the framework foundational.
Here, the soundness definition~\eqref{eq:cf-sound} is a simplified version using the big-step based embedding (but without control flows) introduced in section~\ref{sec:embed-big}, while the original definition in Chargu{\'e}raud's work \cite{chargueraud2011characteristic} also considers features like separation logics.



\section{Proving Extended Rules in Different Embeddings}
\label{sec:proof}


So far, we introduced three categories of extended proof rules in section~\ref{sec:rules} and reviewed four dominant embedding methods for Hoare logic in section~\ref{sec:embed}.
This section investigates into proofs or disproofs of these rules under different embedding methods.
Before that, we first demonstrate why such research is non-trivial.




\textit{\textbf{Different embeddings define different sets of triples}}.
Given a Hoare logic embedding $\vDash$, we use it to define a set of all triples that are satisfiable under the embedding as $\{S \lsep \vDash S\}$.
For two different embeddings $\vDash_1$ and $\vDash_2$, if there exists some triple $S$ such that ${\vDash_1 S} \nLeftrightarrow {\vDash_2 S}$, then we know sets $\{S \lsep \vDash_1 S\}$ and $\{S \lsep \vDash_2 S\}$ are not equal.
Typically, a property holds for elements in one set does not necessarily hold for those in a different set.
Extra elements in a different set may not satisfy the property.
As a result, \textit{\textbf{a rule admissible in one embedding is not necessarily admissible in another embedding}}.
\begin{itemize}
    \item Shallow embeddings we have considered so far define different sets of triples.
    Although it is possible to bridge the big-step based embedding and the weakest precondition based embedding by proving the equivalence of two semantics, the equivalence between these embeddings and the continuation based one is not trivial and we believe that they define different sets of triples.
    \item The deep embedding in this paper defines different sets of triples.
    This is true unless the deep one is proven both sound and complete w.r.t. the shallow one.
    However, completeness is usually not necessary for a program verification tool because users mainly require the correctness (soundness) of the tool.
    It is also too challenging to give the completeness proof for a complex program logic.
    Therefore, we assume the deeply embedded logic is not complete.
\end{itemize}

In conclusion, some extended rule admissible in one of four embeddings may not be admissible in the others.
It is necessary to check whether each extended rule holds under each of these embeddings.

\begin{table}[h]
{
\centering
\resizebox{\textwidth}{!}{
\begin{tabular}{|c|c|c|c|c|}
\hline
\textbf{Proof Rule} & \begin{tabular}{c} \textbf{Deep Embed}$^d$ \\ (section~\ref{sec:proof-deep}) \end{tabular}
& \begin{tabular}{c} \textbf{Big-step} \\ (section~\ref{sec:proof-big}) \end{tabular}
& \begin{tabular}{c} \textbf{Weakest Pre.} \\ (section~\ref{sec:proof-small}) \end{tabular}
& \begin{tabular}{c} \textbf{Continuation} \\ (section~\ref{sec:proof-cont}) \end{tabular}
    \\\hline
\textsc{seq-inv}        &  Simple   & Simple  & 
\begin{tabular}{c}
    Medium: \\
By Simulation\textsuperscript{b}
\end{tabular}   &  
\begin{tabular}{c}
    Difficult: \\
$ \begin{matrix}
    \text{By Simulation} \\
    \&\ \text{Construction}^\text{c}
\end{matrix} $
\end{tabular} \\\hline
\textsc{nocontinue} &  Simple   & Simple  &  
\begin{tabular}{c}  Medium: \\ $ \begin{matrix}
    \text{By Simulation}
\end{matrix} $ \end{tabular}   &  
\begin{tabular}{c}
        Difficult: \\
    $ \begin{matrix}
        \text{By Simulation} \\
        \&\ \text{Construction}^\text{c}
    \end{matrix} $
\end{tabular} \\\hline
\textsc{if-seq}         & Simple & Simple  & 
\begin{tabular}{c}
    Medium: \\
    By Simulation
\end{tabular}   &
\begin{tabular}{c}
    Medium: \\
    By Simulation
\end{tabular}  \\\hline
\textsc{loop-nocontinue}& Simple   & Simple  &  
\begin{tabular}{c}
    Medium: \\
    By Simulation
\end{tabular}   &
\begin{tabular}{c}
    Medium: \\
    By Simulation
\end{tabular} \\\hline
\textsc{hoare-ex}       & 
    Difficult\textsuperscript{a}
& Simple &  Simple &  Simple \\\hline
\end{tabular}
}
\footnotesize
\begin{flushleft}
(a)~\textsc{hoare-ex} holds in deeply embedded logic, but the proof is complicated and requires impredicative assertions. \linebreak
(b)~\textsc{seq-inv} holds for our toy logic with weakest precondition based embedding, but Iris's weakest precondition definition, which also uses this embedding, is special and \textsc{seq-inv} fails in Iris. We explain this in section~\ref{sec:project-iris-shallow}.\linebreak
(c)~Many proofs in the continuation based embedding are intricate. We can prove \textsc{nocontinue} and \textsc{seq-inv} with complex constructions of continuation $\kappa$.
But the feasibility of such construction and the complexity of relevant proofs remains unknown for different and more complicated programming languages.\linebreak
(d)~The deep embedding and its proofs are implemented in our deeply embedded VST, which is integrated into the VST repository (https://github.com/PrincetonUniversity/VST).
\end{flushleft}
}
\caption{Extended Proof Rules in Different Embedding Methods}
\label{fig:summary-rule-embed}
\end{table}

We find all extended rules in figure~\ref{fig:extended} can be proved for all four embeddings, but some proofs are complicated and some proof obligations may be unreasonably heavy in reality for a more complex language.
We list our techniques to overcome these difficulties in table~\ref{fig:summary-rule-embed} and make some brief comments here.
\begin{itemize}
    \item Most proofs for the \textbf{deep} embedding, except the proof of \textsc{hoare-ex}, are simple inductions over proof trees.
    Rules are proved correct mainly by definitions for the \textbf{big-step} based embedding.
    These proofs are relatively trivial.
    \item Most proofs in the \textbf{weakest precondition} based embedding involves defining a simulation relation between pairs of expressions and program states.
    The \textsc{if-seq} and \textsc{loop-nocontinue} proof in the \textbf{continuation} based embedding uses the same idea.
    The principle behind these proofs are not difficult but do requires significantly more work than those for \textbf{deep} embedding and \textbf{big-step} based embedding.
    Thus we would say they are of medium difficulty.
    \item Proofs for \textsc{seq-inv} and \textsc{nocontinue} in the \textbf{continuation} based embedding not only require this simulation technique, but also demand complex constructions of continuations.
    These proofs are very complicated.
    \item \textsc{hoare-ex}'s proofs under shallow embeddings are trivial but are relatively complicated under the deep embedding.
\end{itemize}

In following sections, we will demonstrate our proofs under different embeddings using our WhileCF toy language.
These proofs are one of the major contributions of the paper and are formalized in Coq \cite{exrule-repo}.

To review, we paste extended rules we are going to prove in figure~\ref{fig:review}.
\begin{figure}[h]
\paragraph{Extended Rules --- Transformation Rules} \
\vspace{-2ex}
\begin{mathpar}
\inferrule*[left=if-seq]{
    \vdash \triple{P}{\cif{e}{c_1 \cseq c_3}{c_2 \cseq c_3}}{Q, [\vec{R}]}
}{
    \vdash \triple{P}{(\cif{e}{c_1}{c_2}) \cseq c_3}{Q, [\vec{R}]}
} \and
\inferrule[loop-nocontinue]{
    c_1, c_2 \text{ contain no } \ccontinue \\
    \vdash \triple{P}{\cfor{c_1 \cseq c_2}{\cskip}}{Q, [\vec{R}]}
}{
    \vdash \triple{P}{\cfor{c_1}{c_2}}{Q, [\vec{R}]}
}
\end{mathpar}

\paragraph{Extended Rules --- Structural Rules} \
\vspace{-2ex}
\begin{mathpar}
    \inferrule*[lab=hoare-ex]{
        \text{forall } x.\, \vdash \triple{P(x)}{c}{Q, [\vec{R}]}
    }{
        \vdash \triple{\exists x.\, P(x)}{c}{Q, [\vec{R}]}
    }
    \and
    \inferrule*[left=nocontinue]{
        c \text{ contains no } \ccontinue \\\\
        \vdash \triple{P}{c}{Q, [R_\ekb, R_\ekc]}
    }{
        \vdash \triple{P}{c}{Q, [R_\ekb, R'_\ekc]}
    }
\end{mathpar}

\paragraph{Extended Rules --- Inversion Rules} \
\vspace{-2ex}
\begin{mathpar}
    \inferrule*[left=seq-inv]{
        \vdash \triple{P}{c_1 \cseq c_2}{Q, [\vec{R}]}
    }{
        \text{exists } Q'.\, \vdash \triple{P}{c_1}{Q', [\vec{R}]}
        \text{ and } \vdash \triple{Q'}{c_2}{Q, [\vec{R}]}
    }
\end{mathpar}

\caption{Representative extended proof rules to be proved}
\label{fig:review}
\end{figure}

\subsection{Deep Embedding}
\label{sec:proof-deep}


\paragraph{The Choice of the Deep Embedding.}
Before proving extended rules in \textit{some} deeply embedded language, we should first fix \textit{the} deep embedding we want to discuss.
Theoretically, difference choices of the primary rule set determine different deep embeddings of program logic.
We find that, if we only pick compositional rules, singleton command rules and the consequence rule, we can prove all inversion rules quite easily by proof-tree induction, and on the top of that, we can then establish most transformation rules and structural rules in a straightforward way.
We choose the deep embedding with rules in figure~\ref{fig:setting} as primary rules.
We first stick to this specific deep embedding in this subsection to demonstrate extended rules' proofs and discuss other possible choices afterwards in section~\ref{sec:choice}.

\paragraph{\textbf{Proving Inversion Rules.}}

The proofs of inversion rules in the deep embedding are mainly by induction over the original proof trees.
We prove \textsc{seq-inv} and \textsc{loop-inv} in theorem~\ref{th:seq-inv} and theorem~\ref{th:loop-inv} for demonstration.
Other inversion rules like \textsc{if-inv} have similar proofs.

\begin{thm}
\label{th:seq-inv}
For any $P$, $c_1$, $c_2$, $Q$ and $\vec{R}$, if $\vdash \triple{P}{c_1 \cseq c_2}{Q, [\vec{R}]}$, then there exists another assertion $S$ such that $\vdash \triple{P}{c_1}{S, [\vec{R}]}$
and $\vdash \triple{S}{c_2}{Q, [\vec{R}]}$.
\end{thm}
\begin{proof}
We prove it by induction over the proof tree of shape
$\vdash \triple{ \cdot }{c_1 \cseq c_2}{\cdot, [\cdot]}$.
In fact, the last step of proving $\triple{P}{c_1 \cseq c_2}{Q, [\vec{R}]}$ is either \textsc{hoare-seq} or \textsc{hoare-consequence}.
For the former situation, the conclusion obviously holds.
For the latter situation, we can find $P'$, $Q'$ and $\vec{R'}$ such that 
$$
\vdash \triple{P'}{c_1 \cseq c_2}{Q', [\vec{R'}]},
P \vdash P', Q' \vdash Q, R'_{\ekb} \vdash R_{\ekb}, \text{ and } R'_{\ekc} \vdash R_{\ekc}.
$$
By induction hypothesis, we can find $S$ such that $\vdash \triple{P'}{c_1}{S, [\vec{R'}]}$ and 
$\vdash\triple{Q}{c_2}{Q', [\vec{R'}]}$.
Then, we get $\vdash \triple{P}{c_1}{S, [\vec{R}]}$
and $\vdash \triple{S}{c_2}{Q, [\vec{R}]}$ by \textsc{hoare-consequence}.
\end{proof}

\begin{thm}
    \label{th:loop-inv}
    For any $P, c_1, c_2, Q$ and $\vec{R}$, if $\vdash \triple{P}{\cfor{c_1}{c_2}}{Q, [R_\ekb, R_\ekc]}$, then there exists two assertions $I_1, I_2$ such that
    $$
    \vdash \triple{I_1}{c_1}{I_2, [Q, I_2]} \text{ and } \vdash \triple{I_2}{c_2}{I_1, [Q, \bot]}, \text{ and } P \vdash I_1.
    $$
\end{thm}
\begin{proof}
    We prove it by induction over the proof tree of shape $\vdash \triple{\cdot}{\cfor{c_1}{c_2}}{\cdot}$.
    Similar to the previous proof, the last step in the proof tree is either \textsc{hoare-loop} or \textsc{hoare-consequence}.
    For the first case, the conclusion trivially holds.
    For the second case, we can find $P'$, $Q'$ and $\vec{R'}$ such that 
    $$
    \vdash \triple{P'}{\cfor{c_1}{c_2}}{Q', [\vec{R'}]},
    P \vdash P', Q' \vdash Q, R'_{\ekb} \vdash R_{\ekb}, \text{ and } R'_{\ekc} \vdash R_{\ekc}.
    $$
    By induction hypothesis, we can find $I_1, I_2$ such that $\vdash \triple{I_1}{c_1}{I_2, [Q', I_2]}$ and 
    $\vdash\triple{I_2}{c_2}{I_1, [Q', \bot]}$ and $P' \vdash I_1$.
    This also implies $P \vdash I_1$ and by \textsc{hoare-consequence} we have $I_1, I_2$ such that $\vdash \triple{I_1}{c_1}{I_2, [Q, I_2]}$ and 
    $\vdash\triple{I_2}{c_2}{I_1, [Q, \bot]}$.
\end{proof}




\vspace{10pt}
\noindent\textit{\textbf{Proving} \textsc{hoare-ex} \textbf{rule}}.
\vspace{5pt}

\noindent
\textsc{hoare-ex} (theorem~\ref{th:hoare-ex-deep}) is hard to prove in a deeply embedded logic since the proof tree of $\vdash \triple{P(x)}{c}{Q, [\vec{R}]}$ may be different for a different $x$.
We prove it by induction on $c$'s syntax tree instead. Here, we only show one induction step in the proof---the case for sequential composition. Other induction steps are proved similarly. A complete proof can be found in our Coq development.

This proof is based on one important assumption about the assertion language: we can quantify over assertions inside an assertion and inject from the meta-logic's propositions into the assertion language, and moreover:

\begin{hypo}
    \label{hypo:ex}
    For any assertion $P, Q, \vec{R}$ and program $c$, if $\vdash \triple{P}{c}{Q, [\vec{R}]}$ holds, then we have
    $
    P \ \vdash \ \exists P_0. \ \left(\vdash \triple{P_0}{c}{Q, [\vec{R}]}\right) \wedge P_0
    $.
\end{hypo}
This hypothesis above is an entailment in the assertion logic and the right hand side is an assertion.
In this assertion, $P_0$ is an existentially quantified assertion inside another assertion, and
$\left(\vdash \triple{P_0}{c}{Q, [\vec{R}]}\right)$ is a meta-logic proposition, stating that a triple is provable, used as a conjunct in a compound assertion.
This hypothesis is obviously true since the existentially quantified $P_0$ on the right hand side can be instantiated by $P$.

\paragraph{Remark.}
This kind of assertions, where universe quantifiers and existential quantifiers can quantify over assertions and Hoare triples can be admitted as assertions, are known as impredicative assertions or impredicative polymorphism \cite{ni2006certified}.
Our proof here assumes impredicative assertions are available in the assertion language, which is indeed true as we discuss the impredicative assertions later in section~\ref{sec:discuss-impred}

\


Back to our proof, hypothesis~\ref{hypo:ex} and \textsc{seq-inv} immediately validates lemma~\ref{th:seq-post-inv}, which is used in the proof of \textsc{hoare-ex} (theorem~\ref{th:hoare-ex-deep}).

\begin{lem}
\label{th:seq-post-inv}
For any $P$, $c_1$, $c_2$, $R$ and $\vec{S}$, if $\vdash \triple{P}{c_1 \cseq c_2}{R, [\vec{S}]}$, then $$
\vdash \left\{P\right\}c_1\left\{\exists Q. \left(\vdash \triple{Q}{c_2}{R, [\vec{S}]}\right) \wedge Q, [\vec{S}]\right\}.
$$
\end{lem}

\begin{thm} \label{th:hoare-ex-deep} 
    For any $T, P, c, Q$ and $\vec{R}$, if for any $x$ of type $T$, we have $\vdash \triple{P(x)}{c}{Q, [\vec{R}]}$, then $\vdash \triple{\exists x.\, P(x)}{c}{Q, [\vec{R}]}$.
\end{thm}
\begin{proof}
By induction on $c$'s syntax tree, the induction step for the sequential composition is to prove $\vdash \triple{\exists x. \ P(x)}{ c_1 \cseq c_2 }{R, [\vec{S}]}$ under the following assumptions:

\begin{tabular}{ll}
(IH1) & For any type $T$ and $P$, $Q$, $R$, \\
& if $\vdash \triple{P(x)}{ c_1}{Q, [\vec{R}]}$ holds for any $x:T$,  then $\vdash \triple{\exists x. \ P(x)}{ c_1 }{Q, [\vec{R}]}$. \\
(IH2) & For any type $T$ and $P$, $Q$, $R$, \\
& if $\vdash \triple{P(x)}{ c_2}{Q, [\vec{R}]}$ holds for any $x:T$, then $\vdash \triple{\exists x. \ P(x)}{ c_2 }{Q, [\vec{R}]}$. \\
(Assu) & $\vdash \triple{P(x)}{ c_1 \cseq c_2 }{R, [\vec{S}]}$ holds for any $x:T$.
\end{tabular}

First, by (Assu) and lemma~\ref{th:seq-post-inv},
$ \vdash \left\{P(x)\right\}{ c_1 }\left\{\exists Q.\, \left(\vdash \triple{Q}{c_2}{R, [\vec{S}]}\right) \wedge Q, [\vec{S}]\right\}$
holds for any $x$ of type $T$. Then, using (IH1), we can get:
$$ \vdash \left\{\exists x. \ P(x)\right\}{ c_1 }\left\{\exists Q. \ \left(\vdash \triple{Q}{c_2}{R, [\vec{S}]}\right) \wedge Q, [\vec{S}]\right\}.$$
Now, if we can prove
$ \vdash \triple{\exists Q. \ \left(\vdash \triple{Q}{c_2}{R, [\vec{S}]}\right) \wedge Q}{ c_2 }{R, [\vec{S}]}$,
then our conclusion will immediately follow according to \textsc{hoare-seq}. In fact, it is straightforward.
By (IH2), we only need to prove that for any assertion $Q$,
$$ \vdash \triple{\left(\vdash \triple{Q}{c_2}{R, [\vec{S}]}\right) \wedge Q}{ c_2 }{R, [\vec{S}]}.$$
Suppose $T_0$ is the set of proofs of $\vdash \triple{Q}{c_2}{R, [\vec{S}]}$, then
$$\left(\vdash \triple{Q}{c_2}{R, [\vec{S}]}\right) \wedge Q \dashv \vdash \exists x: T_0. \ Q$$
Thus, by \textsc{hoare-consequence}, we only need to prove $\vdash \triple{\exists x: T_0. \ Q}{c_2}{R, [\vec{S}]}$.
By (IH2) again, we only need to prove: if $\vdash \triple{Q}{c_2}{R, [\vec{S}]}$ has a proof, then $\vdash \triple{Q}{c_2}{R, [\vec{S}]}$; this is tautology.
Now we complete the induction step for sequential composition.
\end{proof}

Above proof of the sequential composition branch in proving theorem~\ref{th:hoare-ex-deep} is based on \textsc{seq-inv} through lemma~\ref{th:seq-post-inv}.
In fact, the induction over the program $c$'s syntax tree has to consider all possible constructors of a command.
Therefore, it is necessary to develop inversion rules and inversion lemmas like lemma~\ref{th:seq-post-inv} for each constructor.
Moreover, only careful design of primary rules for the deeply embedded Hoare logic can establish these inversion rules: \textbf{for each program constructor, there exists exactly one primary rule corresponds to it}.
Otherwise, the indeterminacy of the Hoare triple's proof tree will obstruct the proof of inversion rules\footnote{
    Taking Iris for example, sequential composition is defined by lambda function application.
    As a result, we cannot have both \textsc{hoare-seq} and \textsc{hoare-app} for function application when developing its deeply embedded version in section~\ref{sec:project-iris-shallow}.
    Otherwise, it is impossible to have \textsc{app-inv} with both primary rule presented, because a Hoare triple about function application may be derived by \textsc{hoare-seq} whose premises cannot derive those of \textsc{hoare-app} and its inversion will fail.
}.

\paragraph{\textbf{Proving Other Rules.}}

The proof of other extended rules in the deep embedding follows a common pattern.
\textit{They first extract information from Hoare triples in premises using inversion rules, and then use primary rules to combine them to establish new triples.}
This is exactly the proof scheme presented in figure \ref{fig:loop-nocont-eg} in section \ref{sec:rules-struct} when we proves \textsc{nocontinue} using inversion rules.
The proof for other rules are similar to this one and is omitted here.

This proof technique using inversion rules can also be applied to shallowly embedded logics, if these inversion rules are true.
However, we would like to explore proofs without assuming the correctness of inversion rules in the following sections.
Even though they are true for our toy language and shallow embeddings for their logics, they may not be true in reality when some other features are added to the logic.
In these cases, a proof to circumvent inversion rules would be helpful.




\subsection{Big-step based Shallow Embedding}
\label{sec:proof-big}

\paragraph{\textbf{Proving Transformation Rules.}}
There is a general proof scheme for transformation rules under shallow embeddings.
\textit{We can first prove some notion of ``semantic similarity'' of two programs in the goal and premise, then lift this ``semantic similarity'' to the logic level and derive a relation between Hoare triples, i.e., the extended rule.}
For example, the ``semantic similarity'' in the big-step based embedding is exactly the refinement of the big-step reduction.
\begin{definition}
    We say program $c_1$ refines $c_2$ ($c_1 \sqsubseteq c_2$) in big-step semantics iff. 
    \begin{itemize}
        \item for any $s_1, s_2, \ek$, if $(c_1, s_1) \Downarrow (\ek, s_2)$, then $(c_2, s_1) \Downarrow (\ek, s_2)$;
        \item and for any $s$, if $(c_1, s) \Uparrow$, then $(c_2, s) \Uparrow$.
    \end{itemize}
\end{definition}
Lemma~\ref{lem:big-triple-trans} lifts the bit-step refinement to the implication between Hoare triples using the big-step embedding.
\begin{lem}
    \label{lem:big-triple-trans}
    For any $c_1, c_2$, if $c_2 \sqsubseteq c_1$, then for any $P, Q, \vec{R}$, $\vDash_b \triple{P}{c_1}{Q, [\vec{R}]}$ implies $\vDash_b \triple{P}{c_2}{Q, [\vec{R}]}$.
\end{lem}
\begin{proof}
    The proof of the theorem is straightforward.
    Assume the pre-state is $s_1$, which satisfies $P$.
    \begin{itemize}
        \item We use contradiction to prove $c_2$ has no error. 
        If it can cause error, $(c_2, s_1) \Uparrow$, then according to $c_2 \sqsubseteq c_1$, we can also construct an error in $c_1$'s execution, $(c_1, s_1) \Uparrow$, which contradicts with $\vDash_b \triple{P}{c_1}{Q, [\vec{R}]}$.
        \item For any ending configuration $(\ek, s_2)$ which $c_2$ can reach, we construct $(c_1, s_1) \Downarrow (\ek, s_2)$ from $(c_2, s_1) \Downarrow (\ek, s_2)$ according to $c_2 \sqsubseteq c_1$.
        By $\vDash_b \triple{P}{c_1}{Q, [\vec{R}]}$, we can show $(\ek, s_2)$ satisfies the post-condition and thus prove $\vDash_b \triple{P}{c_2}{Q, [\vec{R}]}$
    \end{itemize}
\end{proof}

With lemma~\ref{lem:big-triple-trans}, proofs for transformation rules \textsc{if-seq} and \textsc{loop-nocontinue} become apparent with the following two refinements.
$$
\begin{aligned}
(\cif{e}{c_1}{c_2}) \cseq c_3 &\sqsubseteq \cif{e}{c_1 \cseq c_3}{c_2 \cseq c_3}\\
\cfor{c_1}{c_2} &\sqsubseteq \cfor{c_1 \cseq c_2}{\cskip}
\end{aligned}
$$
They are also easy to prove.
For example, we can destruct the big-step reduction
$$
((\cif{e}{c_1}{c_2}) \cseq c_3, s_1) \Downarrow (\epsilon, s_2)
$$
into the following proposition, which essentially encodes the reduction segments of two execution paths.
$$
\begin{aligned}
& (\text{eval}(e, s_1) = \btrue \land (c_1, s_1)\Downarrow (\epsilon, s_3) \land (c_3, s_3) \Downarrow (\epsilon, s_2)) \\
\lor& (\text{eval}(e, s_1) = \bfalse \land (c_2, s_1)\Downarrow (\epsilon, s_3) \land (c_3, s_3) \Downarrow (\epsilon, s_2))
\end{aligned}
$$
We can easily reconstruct a reduction of $\cif{e}{c_1 \cseq c_3}{c_2 \cseq c_3}$ from $s_1$ to $(\epsilon, s_2)$ using this proposition.
Other branches of \textsc{if-seq}'s refinement proof are similar, which can be easily enumerated by discussing four different reductions (three different exit control flow and one error case) of the program.

The proof for the loop-nocontinue is similar but involves induction over the loop iteration.
The reduction of each iteration of $\cfor{c_1}{c_2}$ will be destructed into two reduction segments of $c_1$ and $c_2$, which can be easily combined into a reduction of $c_1 \cseq c_2$ since both contains no $\ccontinue$.
And the reduction of $c_1 \cseq c_2$ is exactly one iteration of $\cfor{c_1 \cseq c_2}{\cskip}$, and combined with induction hypothesis, it reproduces the reduction of $\cfor{c_1 \cseq c_2}{\cskip}$ and proves the refinement.

This is also the approach taken in Software Foundations Vol. 6 \cite{SF} to prove a similar transformation rule in the big-step based embedding.
We will see later in following sections that this approach can also be applied to the weakest precondition based embedding and the continuation based embedding but can have quite different definitions of ``semantic similarity''.


\paragraph{\textbf{Proving Structural Rules.}}
We then show the proof for \textsc{hoare-ex} here and prove other rules in appendix~\ref{sec:Abigproof}.
The conclusion for \textsc{hoare-ex} holds for another two shallow embeddings with similar proofs, which we omit in following sections. 

\begin{thm}
    For all $c, P$, and $Q$, if for all $x$, $\vDash_b \triple{P(x)}{c}{Q}$, then $\vDash_b \triple{\exists x.\, P(x)}{c}{Q}$.
\end{thm}
\begin{proof}
    For $\sigma \vDash \exists x.\, P(x)$, we know there exists some $x$ such that $\sigma \vDash P(x)$.
    By the definition of $\vDash_b$, we can instantiate the premise with initial state $\sigma$ and the conclusion immediately follows by definition. 
\end{proof}

\paragraph{\textbf{Proving Inversion Rules.}}
Inversion rules' proofs of the big-step based embedding is similar to those of the deep embedding.
The big-step semantics is defined inductively w.r.t. the structure of a program, and we can reorganize the big-step relation underlying the triple just like what we do for the deeply embedded logic.
As a result, we can easily prove inversion rules.
We only demonstrate the proof for \textsc{seq-inv} here.
\begin{thm}
    For all $c_1, c_2, P, Q$, and $R$, if $\vDash_b \triple{P}{c_1 \cseq c_2}{Q, [\vec{R}]}$, then there exists $Q'$ such that (1) $\vDash_b \triple{P}{c_1}{Q', [\vec{R}]}$ and (2) $\vDash_b \triple{Q'}{c_2}{Q, [\vec{R}]}$.
\end{thm}
\begin{proof}
    We use the weakest pre-condition of $c_2$ as $Q'$: $\sigma \vDash Q'$ iff. it never cause error, i.e., not $(c_2, \sigma) \Uparrow$, and $\forall \sigma_1. (c_2, \sigma) \Downarrow (\epsilon, \sigma_1)$ implies $\sigma_1 \vDash Q$.

    (1) The safety of $c_1$ is guaranteed by the safety of $c_1 \cseq c_2$ and we only need to show the state after $c_1$ terminates satisfies post-conditions.
    If $c_1$ exits normally, we combine the premise in the weakest pre. and have $(c_1 \cseq c_2, \sigma) \Downarrow (\epsilon, \sigma_1)$.
    By premise, we have $\sigma_1 \vDash Q$ and the intermediate state satisfies $Q'$
    the conclusion holds by the definition of strongest post.
    If $c_1$ exits by control flow, say $(c_1, \sigma) \Downarrow (\ek, \sigma')$, we have $(c_1 \cseq c_2, \sigma) \Downarrow (\ek, \sigma')$ and by premise, we prove the conclusion.

    (2) The conclusion is obvious by the definition of weakest precondition $Q'$.
\end{proof}

\subsection{Weakest precondition based Shallow Embedding}
\label{sec:proof-small}

\paragraph{\textbf{Proving Transformation Rules: Refinement.}}
The principle for proving transformation rules in section~\ref{sec:proof-big} still applies to small-step based embeddings (both weakest precondition based embedding and continuation based embedding).
But we need to find an appropriate definition of ``semantic similarity''.
A first attempt is to use the refinement again, but under the small-step semantics.
\begin{definition}
    We say program $c_1$ refines $c_2$ (denoted $c_1 \sqsubseteq_s c_2$) in small-step semantics iff. for any initial states $\sigma_1$ and terminal configuration $(c, \kappa, \sigma)$, the following is true.
    $$
    ((c_1, \epsilon, \sigma_1) \rightarrow (c, \kappa, \sigma)) \Rightarrow ((c_2, \epsilon, \sigma_1) \rightarrow (c, \kappa, \sigma))
    $$
\end{definition}
And we also have lemma~\ref{lem:small-triple-trans} to lift the semantic refinement to the implication between weakest preconditions, which immediately implies the implication between Hoare triples.
\begin{lem}
    \label{lem:small-triple-trans}
    For any $c_1, c_2$, if $c_2 \sqsubseteq_s c_1$, then for any $\sigma, P, Q, \vec{R}$, $\sigma \vDash \WP (c_1, \epsilon) \progspec{Q, [\vec{R}]}$ implies $\sigma \vDash \WP (c_2, \epsilon) \progspec{Q, [\vec{R}]}$.
\end{lem}
The proof for lemma~\ref{lem:small-triple-trans} is a simple and straightforward co-induction over $\sigma \vDash \WP (c_2, \epsilon) \progspec{Q, [\vec{R}]}$.

To use the lemma, we need to again prove the following two refinements.
$$
\begin{aligned}
(\cif{e}{c_1}{c_2}) \cseq c_3 &\sqsubseteq_s \cif{e}{c_1 \cseq c_3}{c_2 \cseq c_3}\\
\cfor{c_1}{c_2} &\sqsubseteq_s \cfor{c_1 \cseq c_2}{\cskip}
\end{aligned}
$$
A possible proof idea is similar to the one in section~\ref{sec:proof-big}.
We need to destruct the reduction of the left-hand-side program into several reduction segments along all execution paths.
However, the notion of ``reduction segments'' is not for free in the small-step semantics as it is in the big-step semantics.
For example, a direct inversion of the following reduction does not yield the reduction segments of $c_1$, $c_2$, and $c_3$.
$$
((\cif{e}{c_1}{c_2}) \cseq c_3, \epsilon, \sigma_1) \rightarrow^* (c, \kappa, \sigma)
$$
We need to manually prove that the sequential composition can be split into two reductions and the if-statement can be split into two reductions of two branches, i.e., there exists an ending configuration\footnote{
    While the ending configuration of the big-step semantics can only be $(\ek, s)$ for any $s$, in the small-step semantics, it can be $(\cskip, \epsilon, s)$, $(\cbreak, \epsilon, s)$, $(\ccontinue, \epsilon, s)$, or any $(c, \epsilon, s)$ that cannot be reduced any more due to some error, e.g., $(\ccontinue, \kloops{c_1}{c_2}{2}, s)$ since it has no reduction defined.
} $(c', \kappa', \sigma')$, such that
$$
\begin{aligned}
&\bigl(((e, \sigma_1) \rightarrow^* (\btrue, \sigma_1) \land (c_1, \empty, \sigma_1) \rightarrow^* (c', \kappa', \sigma')) \\
&\lor ((e, \sigma_1) \rightarrow^* (\bfalse, \sigma_1) \land (c_2, \empty, \sigma_1) \rightarrow^* (c', \kappa', \sigma'))\bigr) \\
\land& (\kappa' = \epsilon \Rightarrow (c', \kappa' \cdot \kseq{c_3}, \sigma') \rightarrow^* (c, \kappa, \sigma))
\end{aligned}.
$$
The proof needs to be split into two separate lemmas for the sequential composition and the if-statement.
Both are proved by induction over the reduction length.
We can then reorganize these reduction segments to form a reduction of $\cif{e}{c_1 \cseq c_3}{c_2 \cseq c_3}$.

While it is still a reasonable proof for the \textsc{if-seq}, the proof for the loop is chaotic.
In section~\ref{sec:proof-big}, we can directly induct over the number of loop iteration, since it is how the big-step semantic for loops is defined, we cannot do the same thing for loops in small-step semantics.
A direct induction can only be performed to induct over the small-step reduction length!
We need to have another inductive definition for reduction segments of a loop, which we present informally in definition~\ref{def:loop-seg}.
\begin{definition}
    \label{def:loop-seg}
    For any ending configuration $(c, \kappa, s)$, the reduction from $(\cfor{c_1}{c_2}, \epsilon, s_1)$ to $(c, \kappa, s)$ can be decomposed into an $n-$iteration segment iff.
    \begin{itemize}
        \item $n = 0$, then $(c, \kappa, s)$ is reached by some $\cbreak$ from $(c_1 \cseq c_2, \epsilon, s_1)$, or it is some stuck configuration reachable from $(c_1 \cseq c_2, \epsilon, s_1)$, or $c_2$ eventually reduces to $\ccontinue$ and $(c, \kappa, s) = (\ccontinue, \kloops{c_1}{c_2}{2}, s)$ gets stuck.
        \item $n = n' + 1$, then $(c_1 \cseq c_2, \epsilon, s_1) \rightarrow^* (\cskip, \epsilon, s')$, and $(\cfor{c_1}{c_2}, \epsilon, s') \rightarrow^* (c, \kappa, s)$ can be decomposed into an $n'-$iteration segment.
    \end{itemize}
\end{definition}
Lots of details of handling control flows have been neglected in this definition.
To prove such decomposition is possible, we use induction over the reduction length, and it involves many auxiliary definition, which we prefer not to demonstrate here.
Instead, we present another more organized proof technique: simulation.

\paragraph{\textbf{Proving Structural Rules: Simulation.}}
The idea is to directly induct over the length of reductions for the left-hand-side program in a refinement, and for each head reduction, we find zero or finitely many head reductions of the right-hand-side program that corresponds to it.
This is exactly the idea of simulation, and we use the simulation relation (definition~\ref{def:sim}) to record this kind of correspondence, which is used for both the weakest precondition based embedding and the continuation based embedding.
\begin{definition}[Simulation]
    \label{def:sim}
    A relation $\sim$ between two programs is a simulation relation iff. for any $c_1, \kappa_1, c_2, \kappa_2$,  if the \textit{source} program $(c_1, \kappa_1)$ can be simulated by the \textit{target} program $(c_2, \kappa_2)$, i.e. $(c_1, \kappa_1) \sim (c_2, \kappa_2)$, then
    \begin{itemize}
        \item \textbf{Termination:} if $(c_1, \kappa_1)$ has safely terminated, i.e., $\kappa_1$ is empty $\epsilon$ and $c_1$ is $\cskip$, $\cmdbreak\!$, or $\cmdcontinue$, then $(c_2, \kappa_2)$ can reduce to $(c_1, \kappa_1)$ in zero or finite steps without modifying the program state, (or it can stuck in an infinite loop without modifying the state)\footnote{
            We add the disjunction in parentheses to aid the proof for rules only under the continuation based embedding, and we remove it from the definition in this section.
        };
        \item \textbf{Preservation:} if for some program state $\sigma$, $(c_1, \kappa_1, \sigma)$ reduces to $(c_1', \kappa_1', \sigma')$, then there exists $(c_2', \kappa_2')$ simulates $(c_1', \kappa_1')$ by $(c_1', \kappa_1') \sim (c_2', \kappa_2')$, and $(c_2, \kappa_2)$ can reduce to $(c_2', \kappa_2')$ in zero or finite steps with same state modifications, i.e., $(c_2, \kappa_2, \sigma) \rightarrow_c^* (c_2', \kappa_2', \sigma')$;
        \footnote{
            Notice that \textit{stuttering} might occur when $(c_2, \kappa_2)$ always takes zero step, and $(c_1, \kappa_1)$ may not terminate in this case.
            This is a common problem people want to avoid when proving simulations in compiler verifications.
            However, we are considering partial correctness in this paper and only want certain properties will hold when the program terminates.
            Stuttering is not a problem here.
        }
        \item \textbf{Error:} if for some state $\sigma$, $(c_1, \kappa_1, \sigma)$ does not belong to either of the above cases, i.e., it is stuck by an error, then $(c_2, \kappa_2, \sigma)$ will also reduce into some error in zero or finite steps.
    \end{itemize}
\end{definition}
Intuitively, in the \textbf{termination} case, when the \textit{source} terminates, the \textit{target} that simulates the source will also terminate in a few steps and these steps have no virtual effect on the program state.
In the \textbf{preservation} case, any reduction step of the \textit{source} can be simulated by zero or finite steps of the \textit{target}, which mimic source's modifications to the program state, and program pairs they reduce to preserve the simulation relation.
In the \textbf{error} case, when the \textit{source} causes an error, the \textit{target} will also raise an error along certain execution trace.

One may expect that the simulation directly implies the refinement, but the simulation definition here is more general than the refinement because it does not require two programs reduce to the same error.
However, we can directly use the simulation as the notion of ``semantic similarity'' for embeddings using small-step semantics.
And we can use a simulation lemma (lemma~\ref{th:wp-sim} and lemma~\ref{th:guard-sim}) to lift simulations to the logic level.
The simulation lemma here (lemma~\ref{th:wp-sim}) lifts the simulation between two programs to the implication between their weakest preconditions.
\begin{lem}[Weakest Pre. Simulation]
    \label{th:wp-sim}
    If there is a simulation relation $\sim$ and $(c_1, \kappa_1)$ simulates $(c_2, \kappa_2)$, i.e., $(c_2, \kappa_2) \sim (c_1, \kappa_1)$, then
    for all $\sigma, c_1, c_2, \kappa_1, \kappa_2, Q$, and $\vec{R}$, $\sigma \vDash \textsf{WP}\, (c_1, \kappa_1) \progspec{Q, [\vec{R}]}$ implies $\sigma \vDash \textsf{WP}\, (c_2, \kappa_2) \progspec{Q, [\vec{R}]}$.
\end{lem}
\begin{proof}
    We prove the lemma by induction over the reduction of $(c_2, \kappa_2, \sigma)$.
    \begin{itemize}
        \item If $(c_2, \kappa_2)$ has terminated, the simulation implies $(c_1, \kappa_1)$ will terminate with the same exit kind and we prove the conclusion by reflexivity.
        \item If $(c_2, \kappa_2)$ has error, the simulation implies $(c_1, \kappa_1)$ will also have an error, which contradicts with the safety guarantee (the program never cause error) in its weakest precondition.
        \item If $(c_2, \kappa_2)$ further reduces, the simulation relation guide us how to reduce $(c_1, \kappa_1)$ while preserving the simulation. We then unfold the weakest pre.'s definition and forward the execution accordingly and prove the lemma by the induction hypothesis.
    \end{itemize}
\end{proof}



To prove \textsc{if-seq} and \textsc{loop-nocontinue}, we only need to construct two simulations $\sim_1$ and $\sim_2$ with following properties respectively.
$$
\begin{array}{c}
    ((\cif{e}{c_1}{c_2}) \cseq c_3, \epsilon) \sim_1 (\cif{e}{c_1 \cseq c_3}{c_2 \cseq c_3}, \epsilon)\\
    (\cfor{c_1}{c_2}, \epsilon) \sim_2 (\cfor{c_1 \cseq c_2}{\cskip}, \epsilon)
\end{array}
$$
For example, we can define the simulation for \textsc{loop-nocontinue} as the smallest relation with following properties for any $c_1, c_2$ containing no $\ccontinue$.
{
\footnotesize
\begin{gather*}
(\forall c,\kappa. (c, \kappa) \sim_2 (c, \kappa)) \land \\
\left(
\forall \kappa.
\left(\begin{aligned}
    &(\cfor{c_1}{c_2}, \kappa) \sim_2 (\cfor{c_1\cseq c_2}{\cskip}, \kappa) \\
    \land&
    \bigl((c_1 \cseq \ccontinue, \kloops{c_1}{c_2}{1} \cdot \kappa) \sim_2 
    (c_1, c_2 \cdot \ccontinue \cdot \kloops{c_1 \cseq c_2}{\cskip}{1} \cdot \kappa)\bigr) \\
    \land&
    \left(
        \begin{aligned}
            &\forall c_0, \kappa_0. (c_0, \kappa_0) \text{ has no continue} \Rightarrow\\
            &(c_0, \kappa_0 \cdot \ccontinue \cdot \kloops{c_1}{c_2}{1} \cdot \kappa) \sim_2 \\
            &\quad(c_0, \kappa_0 \cdot c_2 \cdot \ccontinue \cdot \kloops{c_1 \cseq c_2}{\cskip}{1} \cdot \kappa)
        \end{aligned}
    \right) \\
    \land&
    \left(
        (\ccontinue, \kloops{c_1}{c_2}{1} \cdot \kappa) \sim_2
        (c_2, \ccontinue \cdot \kloops{c_1 \cseq c_2}{\cskip}{1} \cdot \kappa)
    \right) \\
    \land&
    \left(
        \begin{aligned}
            &\forall c_0, \kappa_0. (c_0, \kappa_0) \text{ has no continue} \Rightarrow\\
            &(c_0, \kappa_0 \cdot \kloops{c_1}{c_2}{2} \cdot \kappa) \sim_2 
            (c_0, \kappa_0 \cdot \ccontinue \cdot \kloops{c_1 \cseq c_2}{\cskip}{1} \cdot \kappa)
        \end{aligned}
    \right) \\
    \land&
    \left(
        (\cskip, \kloops{c_1}{c_2}{2} \cdot \kappa) \sim_2
        (\cskip, \kloops{c_1 \cseq c_2}{\cskip}{2} \cdot \kappa)
    \right) \\
    \land&
    \left(
        (\cbreak, \kloops{c_1}{c_2}{2} \cdot \kappa) \sim_2
        (\cbreak, \kloops{c_1 \cseq c_2}{\cskip}{2} \cdot \kappa)
    \right)
\end{aligned}\right)
\right)
\end{gather*}
}
This seemingly daunting definition is simply a conjunction of all possible intermediate configurations of two programs that need to be related together.
It is also easy to prove this relation $\sim_2$ is a simulation.
We simply prove that in each conjunction branch, after one small-step reduction of the left-hand-side program, we can make zero or finite steps in the right-hand-side program and arrive at another conjunction branch in this relation.
For example, in the first branch,
$$
(\cfor{c_1}{c_2}, \kappa) \sim_2 (\cfor{c_1\cseq c_2}{\cskip}, \kappa)
$$
after one step of $(\cfor{c_1}{c_2}, \kappa)$, we take three steps of $(\cfor{c_1\cseq c_2}{\cskip}, \kappa)$ and arrive at the second branch
$$
(c_1 \cseq \ccontinue, \kloops{c_1}{c_2}{1} \cdot \kappa) \sim_2 
    (c_1, c_2 \cdot \ccontinue \cdot \kloops{c_1 \cseq c_2}{\cskip}{1} \cdot \kappa).
$$
This is just a simple proof that can be automated thourgh small-step symbolic execution, while the complicated induction proofs are all hidden in the simulation lemma.
And by directly application of lemma~\ref{th:wp-sim}, we prove \textsc{loop-nocontinue}.

In fact, the construction of these simulation relations $\sim_2$ can guide the proof using pure refinement.
All conjunction branches will corresponds to some intermediate proofs using definition~\ref{def:loop-seg}.
Strictly speaking, there is no obvious advantages using either approach, but we believe the simulation approach is more suitable for proving small-step properties by thinking in small-step semantics.
And it helps proofs of other extended rules as we will see soon.

\paragraph{\textbf{Proving Structural Rules.}}
We prove \textsc{nocontinue} by showing the correctness of lemma~\ref{th:wp-nocontinue}.


\begin{lem}
    \label{th:wp-nocontinue}
    For all $\sigma, c, \kappa, Q, R_\ekb, R_\ekc$, and $R_\ekc'$, if $\sigma \vDash \textsf{WP}\, (c, \kappa) \progspec{Q, [R_\ekb, R_\ekc]}$ and $c$ and $\kappa$ contains no $\ccontinue$, then $\sigma \vDash \textsf{WP}\, (c, \kappa) \progspec{Q, [R_\ekb, R_\ekc']}$.
\end{lem}
\begin{proof}
    We first prove in lemma~\ref{th:nocont-preserve} that command reduction preserves the no-continue property, which is obvious by discussing all constructors of the small-step semantics.
    \begin{lem}
        \label{th:nocont-preserve}
        For all $c, c', \kappa, \kappa', \sigma$, and $\sigma'$, if $c$ and $\kappa$ contains no $\ccontinue$ and $(c, \kappa, \sigma) \rightarrow_c (c', \kappa', \sigma')$, then $c$ and $\kappa$ contains no $\ccontinue$.
    \end{lem}

    We then induct over the number of steps $(c, \kappa)$ takes to terminate.
    \begin{itemize}
        \item For the base case where $(c, \kappa)$ is a terminal, the conclusion is obvious.
        \item For the case where $(c, \kappa, \sigma)$ can step into $(c', \kappa', \sigma')$, we forward the execution of $(c, \kappa, \sigma)$ in both weakest preconditions and only need to prove
        $$
        \sigma' \vDash \textsf{WP}\, (c', \kappa') \progspec{Q, [R_\ekb, R_\ekc]} \text{ implies } \sigma' \vDash \textsf{WP}\, (c', \kappa') \progspec{Q, [R_\ekb, R_\ekc']}
        $$
        which follows from the induction hypothesis and lemma \ref{th:nocont-preserve}. 
    \end{itemize}
\end{proof}

\paragraph{\textbf{Proving Inversion Rules}}
We prove \textsc{seq-inv}, a representative inversion rule, by lemma~\ref{th:small-seq-inv}, where the simulation lemma helps again.

\begin{lem}
    \label{th:small-seq-inv}
    For all $c_1, c_2, Q$, and $\vec{R}$, there exists $Q'$ such that for all $\sigma$,
    (1) if $\sigma \vDash \textsf{WP}\, (c_1 \cseq c_2, \epsilon) \progspec{Q, [\vec{R}]}$ then $\sigma \vDash \textsf{WP}\, (c_1, \epsilon) \progspec{Q', [\vec{R}]}$,
    and (2) if $\sigma \vDash Q'$ then $\sigma \vDash \textsf{WP}\, (c_2, \epsilon) \progspec{Q, [\vec{R}]}$.
\end{lem}
\begin{proof}
    To prove \textsc{seq-inv}, we need to find an intermediate assertion $Q'$, for which we use
    $\textsf{WP}\, (c_2, \epsilon) \progspec{Q, [\vec{R}]}$.
    This immediately gives us (2) and we can prove (1) following the scheme proposed in figure~\ref{fig:wp-seq-inv}, where \textsc{bind-inv} is the inversion of bind rule.
    This is true because we know when $c_1$ exits by break or continue, we will skip $c_2 \cdot \epsilon$ and reach post-conditions in $\vec{R}$; when $c_1$ exits normally, we will step into $c_2 \cdot \epsilon$ and the weakest precondition for it should be satisfied.
    \begin{figure}
    \centering
    \includegraphics[width=0.8\linewidth]{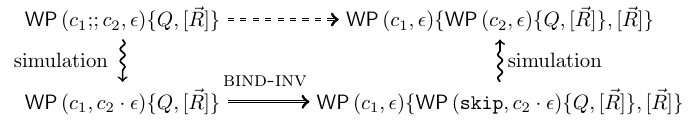}
    \caption{Proof Scheme of \textsc{seq-inv} with Weakest Precondition}
    \label{fig:wp-seq-inv}
    \end{figure}
\end{proof}



\subsection{Continuation based Shallow Embedding}
\label{sec:proof-cont}

\paragraph{\textbf{Proving Transformation Rules.}}
Similar to weakest precondition based embedding, we prefer using the simulation as the ``semantic similarity'' with the simulation lemma (lemma~\ref{th:guard-sim}) that lifts a simulation relation to relations between guard predicates.

Here, we use another definition of the simulation relation.
We add the disjunction in parentheses in definition~\ref{def:sim}, which asserts that when the \textit{source} terminates, the \textit{target} will no longer change the program state and will either terminates in zero or finite steps, or stuck in an infinite loop.
Because under the continuation based embedding, we only require the simulation lemma to provide the derivation between the safety of two programs, i.e., they can proceed without error.
The simulation do not need to simulate their terminations, which they may never reach.
This is an important property which we use in proofs.

\begin{lem}[Guard Simulation]
    \label{th:guard-sim}
    For all $c_1, c_2, \kappa_1, \kappa_2$, and $P$, if $\guard{P}{c_1}{\kappa_1}$ and we have the simulation $(c_2, \kappa_2) \sim (c_1, \kappa_1)$, then $\guard{P}{c_2}{\kappa_2}$.
\end{lem}
\begin{proof}
    The proof is similar to the one for lemma \ref{th:wp-sim} by induction over the reduction of $(c_2, \kappa_2)$.
    The only difference is the \textbf{termination} case, where we no longer require any information from $(c_1, \kappa_1)$ since the conclusion holds by definition.
\end{proof}


With lemma~\ref{th:guard-sim}, we can easily prove \textsc{if-seq} and \textsc{loop-nocontinue}.
The proof idea is similar to the one in section~\ref{sec:proof-small} and we omit it here.

\paragraph{Proving Structural Rules}
Lemma \ref{th:cont-nocont} 
is the key step for proving \textsc{nocontinue} in continuation based embedding, where complex constructions of continuations are involved.

\begin{lem}
    \label{th:cont-nocont}
    For all $c, \kappa_0, P, Q, R_\ekb$, and $R_\ekc$, we have $\guard{Q}{\cskip}{\kappa_0}$ and
    $\guard{R_\ekb}{\cbreak}{\kappa_0}$ 
    imply $\guard{P}{c}{\kappa_0}$,
    if (1) $c$ has no $\ccontinue$; and
    (2) for all $\kappa$, $\guard{Q}{\cskip}{\kappa}$ and 
    $\guard{R_\ekb}{\cbreak}{\kappa}$ and
    $\guard{R_\ekc}{\ccontinue}{\kappa}$ imply $\guard{P}{c}{\kappa}$.
\end{lem}
\begin{proof}
    Here, we set $R_\ekc'$ in \textsc{nocontinue} to $\bot$, which can derives any other $R_\ekc'$ by \textsc{hoare-consequence} rule.
    As a result, we only know the execution of $\kappa_0$ is guarded by $Q$ and $R_\ekb$ but not $R_\ekc$.
    However, we can construct a new $\kappa$ from $\kappa_0$ such that the behaviors of two continuations are ``similar'' when entering by normal exit and break exit, but $\kappa$ ``terminates'' immediately when entering by continue exit to guarantee $\guard{R_\ekc}{\ccontinue}{\kappa}$.
    Also, since $c$ has no $\ccontinue$, we know the behaviors of $(c, \kappa_0)$ and $(c, \kappa)$ are ``similar'' so that they can simulate.

    For such construction to be viable, we first need lemma \ref{th:cont-to-com} to enable conversion from continuations to commands, whose proof is trivial and omitted.

    \begin{lem}
        \label{th:cont-to-com}
        For any continuation $\kappa$, we can construct a command $c_\kappa$, such that $(\cskip, \kappa) \sim (c_\kappa, \epsilon)$, by chaining expressions corresponds to each level in $\kappa$ using sequencing commands and virtual loops\footnote{
            Virtual loops takes the form of $c_3 :: \texttt{for}_i(;;c_2) c_1$, where $c_3$ is the command remaining in the previous iteration and $i = 1,2$.
            Its semantics is defined by the reduction $(c_3 :: \texttt{for}_i(;;c_2) c_1, \kappa, \sigma) \rightarrow_c (c_3, \kloops{c_1}{c_2}{i} \cdot \kappa, \sigma)$.
            Assume $c_\kappa$ is constructed from $\kappa$, we transform $\kappa \cdot \kloops{c_1}{c_2}{i}$ into $c_\kappa :: \texttt{for}_i(;;c_2) c_1$.
        }.
    \end{lem}

    We then show the existence of such construction from $\kappa_0$ to $\kappa$.

    We define program $\texttt{dead} \triangleq \cfor{\cskip}{\cskip}$.
    It will always cause $c \cseq \texttt{dead}$ stuck in an infinite loop after the execution of $c$ and it is always safe to execute $\texttt{dead}$ from any program state.

    We assume $\kappa_0$ takes the form of 
    $$
    \begin{array}{ccccc}
        \underbrace{K_A} & \cdot & \underbrace{\kloops{c_1}{c_2}{i}} & \cdot & \underbrace{K_C} \\
        A & & B & & C
    \end{array}
    $$
    where $B$ is the first (innermost) loop continuation in $\kappa_0$.
    The case where $\kappa_0$ contains no loop, i.e., $B$ and $C$ is empty $\epsilon$ , is trivial and is not discussed here.
    The normal entry into it executes all $ABC$; the break entry into it executes $C$ only; the continue entry into it executes $BC$.
    We use $c_{AB}$ to denote command transformed from continuation $AB$ by lemma~\ref{th:cont-to-com}.
    Now, we construct\footnote{We abbreviate $\kseq{c}\cdot \kappa$ as $c \cdot \kappa$.}
    $$
    \begin{array}{rccccc}
        \kappa \triangleq & \underbrace{c_{AB} \cdot \cbreak} & \cdot & \underbrace{\kloops{\cskip}{\texttt{dead}}{1}} & \cdot & \underbrace{K_C} \\
        & A' & & B' & & C'
    \end{array}
    $$
    The \textbf{normal} entry will execute $c_{AB}$ and then skip the loop in $B'$ to execute $K_C$, i.e., it will execute all continuations $ABC$ in $\kappa_0$;
    the \textbf{break} entry will skip $A'$
    \footnote{Although there is a loop in $c_{AB}$, it is hidden in a sequential continuation and the $\cbreak$ will ignore the loop in $c_{AB}$. Such ill-formed continuation with loop commands in a sequential continuation is only possible in fabricated ones and will not occur in normal small-step reductions.}
    and is scoped by the loop in $B'$ and enter $K_C$ through normal entry, which is equivalent to first skipping loop in $B$ then executing $K_C$;
    the \textbf{continue} entry will skip $A'$ and stuck in a dead loop in $B'$.
    As a result, $\guard{Q}{\cskip}{\kappa_0}$
    will imply $\guard{Q}{\cskip}{\kappa}$, $\guard{Q_\ekb}{\cbreak}{\kappa_0}$ will imply $\guard{Q_\ekb}{\ccontinue}{\kappa}$, and 
    $\guard{Q_\ekc}{\ccontinue}{\kappa}$ unconditionally holds because a dead loop always make progress.

    By premise, we have $\guard{P}{c}{\kappa}$.
    Notice that $(c, \kappa)$ simulates $(c, \kappa_0)$ because $\kappa$ simulates $\kappa_0$ when $c$ exit normally or by break, but $c$ can only exit in these ways instead of continue exit.
    By lemma \ref{th:guard-sim}, we prove $\guard{P}{c}{\kappa_0}$.
\end{proof}


\paragraph{Proving Inversion Rules}
Similar to the proof of \textsc{nocontinue}, the proof of \textsc{seq-inv} (lemma \ref{th:cont-seq-inv}) is also based on constructions of continuations.

\begin{lem}
    \label{th:cont-seq-inv}
    For all $c_1, c_2$ and $P$, 
    if for all $\kappa$, $\guard{Q}{\cskip}{\kappa}$ and $\guard{R_\ekb}{\cbreak}{\kappa}$ and \linebreak
    $\guard{R_\ekc}{\ccontinue}{\kappa}$ imply $\guard{P}{c_1 \cseq c_2}{\kappa}$,
    then there exists $Q'$ such that for all $\kappa_0$,
    \begin{enumerate}
        \item[(1)] $\guard{Q'}{\cskip}{\kappa_0}$ and $\guard{R_\ekb}{\cbreak}{\kappa_0}$ and $\guard{R_\ekc}{\ccontinue}{\kappa_0}$ imply $\guard{P}{c_1}{\kappa_0}$,
        \item[(2)] $\guard{Q}{\cskip}{\kappa_0}$ and $\guard{R_\ekb}{\cbreak}{\kappa_0}$ and $\guard{R_\ekc}{\ccontinue}{\kappa_0}$ imply $\guard{P}{c_2}{\kappa_0}$.
    \end{enumerate}
\end{lem}
\begin{proof}
    We construct $Q'$ as the strongest post. of $c_1$: $\sigma \vDash Q'$ iff. $\exists \sigma_0.\, (c_1, \epsilon, \sigma_0) \rightarrow_c (\cskip, \epsilon, \sigma)$ and $\sigma_o \vDash P$.

    (1) Similar to the proof of \textsc{nocontinue}, we need to construct a $\kappa$ from $\kappa_0$ to utilize the premise by satisfying $\guard{Q}{\cskip}{\kappa}$.
    We take $\kappa \triangleq \texttt{dead} \cdot \kappa_0$.
    $\kappa$ simulates $\kappa_0$ through break and continue entry and we have $\guard{R_\ekb}{\cbreak}{\kappa}$ and $\guard{R_\ekc}{\ccontinue}{\kappa}$ by lemma \ref{th:guard-sim}.
    $\kappa$ always loops through the normal entry and we have $\guard{Q}{\cskip}{\kappa}$.
    By premise, we have $\guard{P}{c_1 \cseq c_2}{\kappa}$.
    This guarantees $c_1$ can execute safely from $P$, and we only need to show that $\kappa_0$ can execute safely after $c_1$.

    When $c_1$ exits through break and continue, $(c_1 \cseq c_2) \cdot \kappa$ simulates $c_1 \cdot \kappa_0$ and we have $\guard{P}{c_1}{\kappa_0}$;
    $\guard{Q'}{\cskip}{\kappa_0}$ tells us when $c_1$ exits normally, we can keep executing $\kappa_0$ safely.
    Together, we have $\guard{P}{c_1}{\kappa_0}$ in all cases.

    (2) By specializing the premise with $\kappa_0$, we can reduce the proof goal to $\guard{P}{c_1 \cseq c_2}{\kappa_0}$ implies $\guard{Q'}{c_2}{\kappa_0}$.
    This is obvious because
    $\guard{P}{c_1 \cseq c_2}{\kappa_0}$ shows the safety $c_2 \cdot \kappa_0$ after we exit $c_1$ normally, and $Q'$, program states after $c_1$'s normal exit, which apparently guards $c_2 \cdot \kappa_0$.
\end{proof}

\paragraph{Summary.}
In summary, we have proved representatives (figure~\ref{fig:review}) of extended rules in three categories under four different embeddings.
\begin{itemize}
    \item The deep embedding allows most of the proofs to be simple, but it has a relative complicated proof for the \textsc{hoare-ex} rule.
    \item The big-step based embedding has relatively simple proofs for all extended rules.
    And there is a general proof principle for transformation rules under all shallow embeddings: first prove a relationship between two programs' semantics, then lift it to a relationship between two programs' Hoare triple.
    \item In the weakest precondition based embedding, proofs are of medium difficulty.
    We can apply the general proof principle for it but we prefer using the simulation relation as the relationship at the semantic level for embeddings based on small-step semantics.
    The simulation relation is relatively easy to construct and is used throughout all extended rule's proof in this embedding, but it could be tedious to define when the program is complicated.
    \item In the continuation based embedding, complexities for transformation rules' proofs are still acceptable with the general proof principle, but is very difficult for other extended rules.
    We need complex constructions for new continuations to utilize information in the premise.
\end{itemize}
We prefer the deep embedding because it manages proofs mainly at the logic level, instead of the semantic level.
We do not need to worry about simulations and defining continuations in this embedding.
The only problem is the \textsc{hoare-ex} rule, which we have proved in this section and the proof can be maintained when extending the language and logic, as long as there exists exactly one primary rule corresponds to it for each program constructor as mentioned in section~\ref{sec:proof-deep}.




\section{From Shallow Embedding to Deep Embedding}
\label{sec:d2s}

In section~\ref{sec:proof}, we have seen proofs for extended proof rules under different program logic embeddings.
The proofs under shallow embeddings (except the big-step based one) is relatively complicated for certain extended rules, e.g., \textsc{if-seq} and \textsc{nocontinue}.
Proofs under both the weakest precondition based one and the continuation based one involves the notion of simulation.
Proofs under the continuation based one even require complex constructions of continuations.
Although it may be possible to formalize these complicated proofs in real verification tools, we propose an easier way to work around them by building a deeply embedded logic above the existing shallowly embedded logic and prove these extended rules under the deeply embedded one.

\begin{figure}[h]
    \centering
    \includegraphics[width=\linewidth]{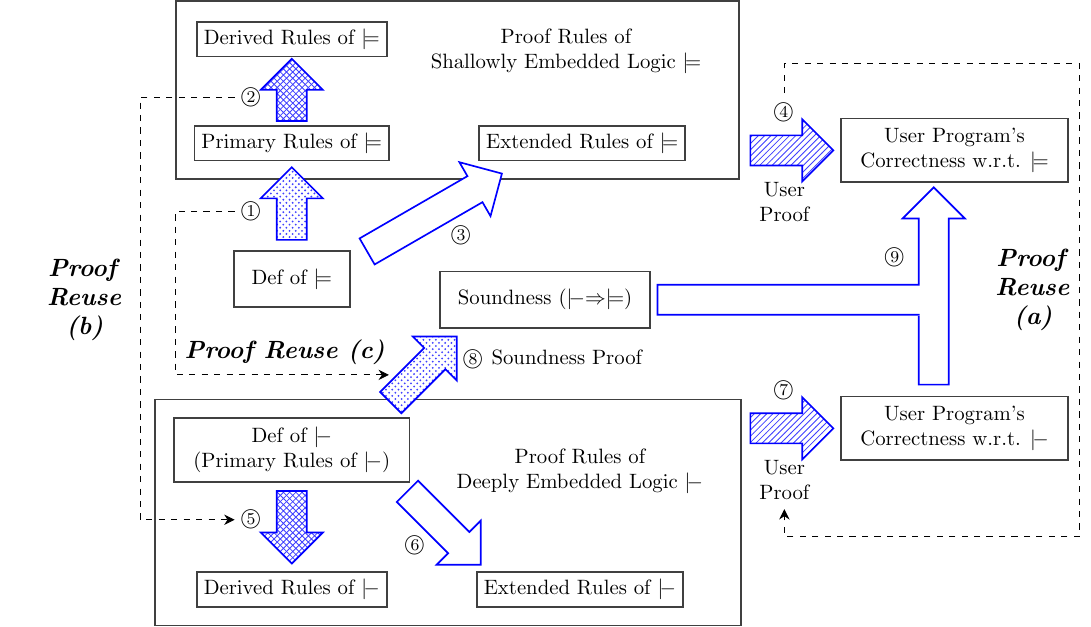}
    \caption{The Framework for Building Deep Embedding from Shallow Embedding. Proofs are represented by arrows labelled with numbers. Proofs of arrows with matching texture can be reused in place of the other.}
    \label{fig:d2s-frame}
\end{figure}

Figure~\ref{fig:d2s-frame} presents our framework for builduing a deeply embedded logic by reusing most of the proofs from an existing shallowly embedded one.
The existing shallow embedding contains the followings:
\begin{itemize}
    \item The definition of the shallow embedding $\vDash \triple{P}{c}{Q, \vec{R}}$.
    \item A set of primary rules and their proofs $\circnum{1}$ directly based on the definition.
    \item A set of derived rules and their proofs $\circnum{2}$ built on primary rules.
    \item A set of extended rules and their proofs $\circnum{3}$ directly based on the definition.
    \item User's proof $\circnum{4}$ of their programs using only\footnote{
        The embedding's definition is not exposed to users as part of its interface.
        User's proof only consists of application of these proof rules and proofs irrelevant to the program logic.
    } rules of these three types.
\end{itemize}
The deeply embedded verification tool we want needs to provide the followings:
\begin{itemize}
    \item The definition of the embedding, i.e., the set of primary rules for proving $\vdash \triple{P}{c}{Q, \vec{R}}$.
    \item A set of derived rules and a set of extended rules along with their proofs $\circnum{5}\circnum{6}$ directly built on primary rules.
    \item User's proof $\circnum{7}$ of their own programs using rules of these three types.
    \item To guarantee the correctness of the deep embedding, a soundness theorem and its proof $\circnum{8}$ is required.
\end{itemize}

We use same sets of proofs rule in the deep embedding as the shallow one.
The primary rules, derived rules, and extended rules in the shallow embedding and the deep embedding have same forms.
In this way, many proofs in the shallow embedding side can be reused to assist the construction of the deep embedding as indicated in figure~\ref{fig:d2s-frame}.
\begin{description}
    \item[\textbf{Proof reuse (a)}.] Since all proof rules in two embedding have same forms, users can upgrade their proofs in the shallow embedding $\circnum{4}$ to those in the deep one $\circnum{7}$ with only one change in proofs.
    They only need to replace occurances of proof rules with their counterparts in the deep one.
    The new proof establishes the program correctness as deeply embedded Hoare triples.
    \item[\textbf{Proof reuse (b)}.] As mentioned in section~\ref{sec:intro}, derived rules are directly derived from primary rules with trivial proofs. Since primary rules in two embeddings have same forms, derived rules' proofs $\circnum{2}$ and $\circnum{5}$ are in the same situation as case \textbf{\textit{(a)}} and we can reuse them.
    \item[\textbf{Proof reuse (c)}.] The most important part in a deeply embedded logic is its soundness proof $\circnum{8}$, which guarantees the logic's correctness.
    However, we can construct the soundness proof for free by reusing existing primary rules' proofs $\circnum{1}$ in the shallow embedding.
    To prove the soundness, we start by induction over the proof tree and obtain subgoals that are identical to shallowly embedded primary rules, which are already proved.

    Moreover, given the soundness and user program correctness in the deep embedding, we can also derive program correctness in the shallow embedding $\circnum{9}$, if it is of user's desire.
    It is a trivial instantiation of the soundness theorem on deeply embedded Hoare triples of user programs, where both are generated by reusing proofs in the shallow embedding.
\end{description}

\

However, we do not need to reuse proofs of extended rules.
As we have shown in section~\ref{sec:proof},
their proofs $\circnum{3}$ in the shallow embedding are often complex constructions, while in the deep embedding, their proofs $\circnum{6}$ are mostly simple inductions over the proof tree.
Proving extended rules are simpler in a deep embedding than in a shallow embedding.
Moreover, to our knowledge, many shallowly embedded verification tools (e.g. Iris) do not provide extended rules and their proofs.
There is no way to reuse what does not exist, but we can easily equip them with these extended rules when lifted to deep embeddings.

In conclusion, building a deeply embedded verification tools from an existing shallowly embedded one is mainly simple reuses of existing proofs and supports easier proofs of extended rules.
Existing user's proofs can also be easily upgraded to the deeply embedded version.

\section{Choices of Primary Rules in Deep Embedding}
\label{sec:choice}

At the beginning of section~\ref{sec:proof-deep}, we have fixed our deeply embedded logic with only primary rules in figure~\ref{fig:setting}.
In fact, we may choose deep embeddings with other admitted primary rules.
For example, we may choose \textsc{frame} rule as a primary rule as discussed in section~\ref{sec:frame-dep}.
In this section, we want to discuss the choices of making extended rules as primary rules.

\vspace{10pt}
\noindent\textit{\textbf{Choosing}} \textsc{hoare-ex} \textit{\textbf{as a Primary Rule.}}
\vspace{5pt}

\noindent In section~\ref{sec:proof}, we notice that the proof of \textsc{hoare-ex} is simpler in shallow embeddings than its proof in the deep embedding.
It seems a good choice to make it a primary rule and use the technique in section~\ref{sec:d2s} to lift its simple shallow embedding proof to the deep embedding's soundness proof.
However, we still prefer not choosing it as primary proof rule in the deeply embedded logic and prove it as an extended rule.
This is also the design choice in our deeply embedded VST.
Adding \textsc{hoare-ex} as a primary rule causes problems.


Adding more primary rules to the proof system would invalidate the original proofs of other extended rules.
In deep embedding, we mainly use induction over proof trees when proving extended rules, and now we need to discuss one more case in such induction where the Hoare triple is constructed by \textsc{hoare-ex}.
Taking \textsc{seq-inv} as an example, the extra case is that $\vdash \triple{P}{c_1\cmdseq c_2}{Q, [\vec{R}]}$ is constructed by \textsc{hoare-ex}, i.e., $P$ must takes the form of $\exists x: T.P'(x)$.
And the induction hypothesis is for any $x$ of type $T$ we can find an $S$ such that $\vdash \triple{P'(x)}{c_1}{S, [\vec{R}]}$ and $\vdash \triple{S}{c_2}{Q, [\vec{R}]}$.
We then need to use the axiom of choice (AC) to complete the proof: we need to find a function $S'$
from $T$ to assertions,
such that for any $x$ of type $T$, $\vdash \triple{P'(x)}{c_1}{S'(x), [\vec{R}]}$ and $\vdash \triple{S'(x)}{c_2}{Q, [\vec{R}]}$.
It is controversial whether we should accept AC when we prove meta-theorems.
Original VST does not need AC.
Moreover, to complete this proof in Coq, we actually need the computational version of AC.
To prevent such problems when we want more extended rule, we treat \textsc{hoare-ex} as an extended rule and prove it as a meta-theorem.

\vspace*{10pt}
\noindent\textit{\textbf{Choosing Transformation Rules as Primary Rules.}}
\vspace*{5pt}

Transformation rules can not treated as primary rules either.
Otherwise, some program may have multiple proofs through the original proof rules and through the transformation rule.
This invalidates our previous proof of inversion rules.

For example, if we admit \textsc{loop-nocontinue} as a primary rule, then in the proof of \textsc{loop-inv}, we need to consider one more induction case:
given $\vdash \triple{P}{\cfor{c_1\cseq c_2}{\cskip}}{Q, [\vec{R}]}$, we need to show that there exists $I_1$ and $I_2$ such that $\vdash \triple{I_1}{c_1}{I_2, [Q, I_2]}$ and $\vdash \triple{I_2}{c_2}{I_1, [Q, \bot]}$ and $P \vdash I_1$.
Different from the case in section~\ref{sec:proof-big}, the induction hypothesis cannot extract any information from the premise.
The loop must take the form of $\cfor{c_1}{c_2}$ to apply the induction hypothesis.
As a result, we fail to prove \textsc{loop-inv}.

Moreover, admitting transformation rules as primary rules acquires developers to give their soundness proof in some shallow embedding.
As we have seen in section~\ref{sec:proof}, their proofs under weakest precondition based embedding and continuation based embedding is not as simple as other primary rules' proofs.

\vspace*{10pt}
\noindent\textit{\textbf{Choosing other Structural Rules as Primary Rules.}}
\vspace*{5pt}

Both \textsc{frame} rule and \textsc{hypo-frame} rule are preferred not chosen as primary rules for the same reason as transformation rules.
They make proof tree induction complicated.
Moreover, a complex validity definition is required to support the soundness of the logic with \textsc{hypo-frame} rule as a primary rule.



\vspace*{10pt}
\noindent\textit{\textbf{Choosing Inversion Rules as Primary Rules.}}
\vspace*{5pt}

In fact, choosing inversion rules as primary rules is not even an option.
Conclusions of primary rules must take the form of Hoare triples,
while those of inversion rules are often propositions lead by meta-language's existential quantifier.
We can not define the provable relation for these propositions.

\

\vspace*{10pt}
\noindent\textit{\textbf{Summary.}}
\vspace*{5pt}

It is possible to choose structural rules and transformation rules as primary rules, but it makes proof tree induction and soundness proof complicated.
It is impossible to formalize inversion rules as primary rules.
As a result, our current choice of primary rules (compositional rules, the consequence rule, and singleton command rules) is the most suitable one for instantiate a deeply embedded verification tool.

\section{Discussion}
\label{sec:hypo}



In this section, we discus various extensions of the Hoare logic discussed in previous sections and their potential influence on the logic embedding.
Specifically, we will first consider program logics with load and store to discuss some basics of separation logic embeddings.
Then we introduce rules related to procedure calls above it and further discuss embeddings for the logic with procedure calls.
We will also discuss various features of the program logic and assertion language like total correctness, non-determinism, and step-indexed definitions, along with the soundness proof technique mentioned in section~\ref{sec:nomen-logic}.

\subsection{Embeddings of Separation Logic}
\label{sec:discuss-sep}

When reasoning about memory operations, it is useful to extend the Hoare logic to a separation logic \cite{reynolds02}.
A separation logic adds a logic connective, ``$*$'', separating conjunction, to the assertion language. The assertion $P * Q$ asserts two disjoint memory satisfying $P$ and $Q$ respectively.
Typically, a separation logic adds proof rules for memory related assignments (rules for load and store) and the \textsc{frame} rule to Hoare logic.
Also, when function call is involved in a programming language, separation logic allows verifiers to prove a Hoare triple about global memory while the callee function specification only states its local effects.

Figure~\ref{fig:sep} extends our previous toy WhileCF language with two memory operations.
\begin{itemize}
\item Simple load, $x = [y]$, where $x$ and $y$ are program variables, and this command loads the value from the location with address $y$ into variable $x$.
\item Simple store, $[x] = y$, where $x$ and $y$ are program variables, and this command stores the value of $y$ into the location with address $y$.
\end{itemize}

\begin{figure}[h]
\paragraph{While-CF with Memory Operations}
$$
\begin{array}{rcl}
c \in \text{command} &:=& \cdots \ \mid \ x = [y] \ \mid \ [x] = y
\end{array}
$$
\paragraph{Separation Logic Proof Rules}
\begin{mathpar}
    \inferrule[hoare-load]{}{
        \vdash \triple{l \mapsto v \land \doublebrackets{y} = l}{\ \ x = [y] \ \ }{l \mapsto v \land  \doublebrackets{x} = v \land \doublebrackets{y} = l, [\bot, \bot]}}
    \and
    \inferrule[hoare-store]{}{
        \vdash \triple{l \mapsto v \land \doublebrackets{x} = l \land \doublebrackets{y} = u}{\ \ [x] = y \ \ }{l \mapsto u \land \doublebrackets{x} = l \land \doublebrackets{y} = u, [\bot, \bot]}}
    \and
    \inferrule[frame]{
        c \text{ does not modify program variables freely occurring in } $F$ \\\\
        \vdash \triple{P}{c}{Q, [Q_{\text{brk}}, Q_{\text{con}}]}
    }{
        \vdash \triple{P * F}{c}{Q * F, [Q_{\text{brk}} * F, Q_{\text{con}} * F]}
    }
\end{mathpar}
    \caption{Additional Proof Rules for Separation Logics}
    \label{fig:sep}
\end{figure}

In addition, we assume that there is no memory load in an expression.
Every memory operation is performed as a load or store command.
This is similar to the Clight program in CompCert \cite{krebbers2014formal} and the C program in VST-Floyd \cite{VST-Floyd}.

To verify these memory operations, figure~\ref{fig:sep} adds two standard primary rules, \textsc{hoare-load} and \textsc{hoare-store}, to our previous logic in section~\ref{sec:rules}.
The \textsc{frame} rule is the main focus of this section.

\subsubsection{Frame Rule and More Potential Shallow Embeddings}
\label{sec:frame}

The \textsc{frame} rule allows prover to remove unused separating conjuncts (assertions in separation logics) of memory from both pre-/post-conditions and prove the specification with remaining memories.
Gotsman \textit{et al.} \cite{gotsman2011precision} proved: if the operational semantics satisfies the frame property, then the \textsc{frame} rule will always generate valid conclusions from valid assumptions.
The following hypotheses state the frame property via a big-step semantics and via a small-step semantics.
We use $\sigma' \oplus \sigma$ to represent the disjoint union of two pieces of memory, $\sigma'$ and $\sigma$.

\begin{hypo}
\label{hyp:frame-big}
For any $c$, $\text{ek}$, $\sigma_1$, $\sigma'_2$ and $\sigma$,
\begin{itemize}
\item if $(c, \sigma_1 \oplus \sigma) \Uparrow$, then $(c, \sigma_1) \Uparrow$;
\item if $(c, \sigma_1 \oplus \sigma) \Downarrow (\text{ek}, \sigma'_2)$, then either $(c, \sigma_1) \Uparrow$ or
   $\sigma'_2 = \sigma_2 \oplus \sigma $ for some $\sigma_2$ and  $(c, \sigma_1) \Downarrow (\text{ek}, \sigma_2)$.
\end{itemize}
\end{hypo}

\begin{hypo}
\label{hyp:frame-small}
For any $c_1$, $c_2$, $\kappa_1$, $\kappa_2$, $\sigma_1$, $\sigma'_2$ and $\sigma$,
\begin{itemize}
\item if $(c_1, \kappa_1, \sigma_1 \oplus \sigma) \not\rightarrow_c$, then $(c_1, \kappa_1, \sigma_1) \not\rightarrow_c$;
\item if $(c_1, \kappa_1, \sigma_1 \oplus \sigma) \rightarrow_c (c_2, \kappa_2, \sigma'_2)$, then either $(c_1,\kappa_1,  \sigma_1)  \not\rightarrow_c$ or
   $\sigma'_2 = \sigma_2 \oplus \sigma $ for some $\sigma_2$ and  $(c_1, \kappa_1, \sigma_1) \rightarrow_c (c_2, \kappa_2, \sigma_2)$.
\end{itemize}
\end{hypo}

These two properties are true for most reasonable languages (including realistic languages like C). And Gotsman's conclusion holds for $\vDash_b$,  $\vDash_w$, and $\vDash_c$ (respectively stand for the validity definition under big-step based embedding, weakest precondition based embedding, and continuation-based embedding).
For example, theorem~\ref{thm:frame-big} proves the \textsc{frame} rule in the big-step based embedding.

\begin{thm}
\label{thm:frame-big}
    The \textsc{frame} rule holds for the big-step based embedding defined in section~\ref{sec:embed-big}, if the hypothesis~\ref{hyp:frame-big} is true.
\end{thm}
\begin{proof}
To prove $\bigvalid \triple{P*F}{c}{Q*F, [R_\ekb*F, R_\ekc*F]}$, we need to show that for any $\sigma \vDash P*F$ (1) $\lnot\,(c, \sigma)\Uparrow$, and (2) $\text{for all } \ek, \sigma_2, \text{if }(c, \sigma) \Downarrow (\ek, \sigma_2)$, $\sigma_2$ satisfies corresponding post-conditions with the frame $F$ according to $\ek$.

According to the premise $\bigvalid \triple{P}{c}{Q, [R_\ekb, R_\ekc]}$, we know for any $\sigma_1 \vDash P$ (i) $\lnot\,(c, \sigma_1)\Uparrow$, and (ii) $\text{for all } \ek, \sigma_2, \text{if }(c, \sigma_1) \Downarrow (\ek, \sigma_2)$, $\sigma_2$ satisfies corresponding post-conditions without the frame $F$ according to $\ek$.

We can split $\sigma$, according to $\sigma \vDash P*F$, into $\sigma_1$ and $\sigma'$ satisfying
$$
\sigma = \sigma_1 \oplus \sigma' \quad\text{and}\quad \sigma_1 \vDash P \quad\text{and}\quad \sigma' \vDash F.
$$
By (i), we know $\lnot\,(c, \sigma_1)\Uparrow$. And by hypothesis~\ref{hyp:frame-big}, we have $\lnot\,(c, \sigma_1 \oplus \sigma')\Uparrow$ because otherwise we have a contradiction and we prove (1).

To prove (2), by hypothesis~\ref{hyp:frame-big}, for any $\ek, \sigma_2$ with $(c, \sigma_1 \oplus \sigma') \Downarrow (\ek, \sigma_2)$, we know that
\begin{itemize}
    \item either $(c, \sigma_1)\Uparrow$, which contradicts with the premise (i) and finishes the proof;
    \item or there exists $\sigma_2'$ with $\sigma_2 = \sigma_2' \oplus \sigma'$ and $(c, \sigma_1) \Downarrow (\ek, \sigma_2')$.
    By (ii), we know $\sigma_2'$ satisfies corresponding post-conditions without the frame $F$ according to $\ek$ (i.e., $\sigma_2' \vDash Q$, $\sigma_2' \vDash R_\ekb$, or $\sigma_2' \vDash R_\ekc$).
    Since $\sigma' \Vdash F$, we know $\sigma_2 = \sigma_2' \oplus \sigma'$ satisfies corresponding post-conditions with the frame $F$ according to $\ek$ (i.e., $\sigma_2 \vDash Q*F$, $\sigma_2 \vDash R_\ekb*F$, or $\sigma_2 \vDash R_\ekc*F$), which finishes the proof.
\end{itemize}
\end{proof}

Interestingly, a different style of shallow embeddings for separation logic \cite{torp2005semantics} are widely used and the \textsc{frame} rule will hold intrinsically.
The \textsc{frame} rule in this formalization is called the "baked-in frame rule" in some literatures.
For any shallow embedding $\vDash_x$ among $\vDash_b$,  $\vDash_w$, and $\vDash_c$, it defines the Hoare triple as
$$\vDash_{\text{sep}-x} \triple{P}{c}{Q, [R_\ekb, R_\ekc]} \ \text{iff. for any} \ F, \vDash_x \triple{P*F}{c}{Q*F, [R_\ekb * F, R_\ekc * F]}. $$
Taking continuation based shallow embedding for example, we can use another definition: $\vDash_{\text{sep}-c} \triple{P}{c}{Q, [R_\ekb, R_\ekc]}$ iff. for arbitrary continuation $\kappa$ and \textit{arbitrary frame $F$} (also a separation logic assertion),
$$
\text{if } \quad \begin{cases}
    \quad\guard{Q * F}{\cskip}{\kappa} \\
    \quad\guard{Q_\ekb * F}{\cbreak}{\kappa} \\
    \quad\guard{Q_\ekc * F}{\ccontinue}{\kappa} 
\end{cases}
\text{then } \quad \guard{P * F}{c}{\kappa}.
$$
This is the way how shallowly embedded VST defines its Hoare triples.
Under such embedding, we can easily verify the \textsc{frame} rule without using the frame property of small-step operational semantics.
The premise $\vDash_{\text{sep}-c}~\!\triple{P}{c}{Q, [R_\ekb, R_\ekc]}$ means that for arbitrary frame $F''$, $$\vDash_{c} \triple{P * F''}{c}{Q * F'', [R_\ekb * F'', R_\ekc * F'']}.$$
Let $F'' = F * F'$. Then we know $\vDash_{c} \triple{(P * F) * F'}{c}{(Q * F) * F', [(R_\ekb * F) * F', (R_\ekc * F) * F']}$ for any $F$ and $F'$, which exactly states the conclusion of \textsf{frame}:
$$\vDash_{\text{sep}-c} \triple{P*F}{c}{Q*F}.$$
Although such sep$-x$ embedding does not invalidate any primary rules or extended rules we have mentioned so far, it does make other proofs more complicated since we need to deals with this extra frame $F$.
In comparison, both embedding styles are feasible.
The original one makes proofs of other primary rules and extended rules relatively concise, while the sep$-x$ version makes the proof of the frame rule simple.

To summarize, in a shallow embedding, we can either directly prove it if hypothesis~\ref{hyp:frame-big} and hypothesis~\ref{hyp:frame-small} are true, or use the baked-in frame rule and reprove other rules using the new shallow embedding $\vDash_{\text{sep-x}}$.
It does not matter whether this is a primary rule or an extended rule in a shallow embedding because both type of proof rules are just lemmas.

\subsubsection{Frame Rule and Deep Embeddings}
\label{sec:frame-dep}
In deeply embedded program logics, we can either choose to make \textsc{frame} a primary rule or treat it as an extended rule.

The former option means that the separation logic's soundness is based on the fact that \textsc{frame} preserves Hoare triple's validity.
This proof is very similar to the verification of \textsc{frame} in shallowly embedded logics.
When developing a deeply embedded logic based on a shallowly embedded logic (see section~\ref{sec:d2s}), one could directly reuse the proof of \textsc{frame} rule in the shallow embedding.
However, as discussed in previous sections, adding one more primary rule to the proof system may invalidate the other proofs because every case analysis and inductions on proof trees will have one more branch to prove.
In the case of \textsc{frame} rule, although extended rules we have discussed may not be invalidated by this extra case, such choice would cause more uncertainty when we want to add more extended rules where this extra case may be unprovable.
Keeping the proof system concise with only singleton rules, compositional rules, and the consequence rule will make the logic more extensible and thus we give up this design choice and prove frame rule as an extended rule.

Nevertheless, we are not yet able to prove it with our previous deeply embedded program logic by simple induction over the proof tree of $\triple{P}{c}{Q}$.
The induction steps of \textsc{hoare-load} and \textsc{hoare-store} are not provable.
We should allow these rules to be compatible with extra frames.
For example, we use \textsc{hoare-load-frame} to replace \textsc{hoare-load}.
\begin{mathpar}
    \inferrule[hoare-load-frame]{}{
        \vdash \triple{F * l \mapsto v \land \doublebrackets{y} = l}{\ \ x = [y] \ \ }{F * l \mapsto v \land  \doublebrackets{x} = v \land \doublebrackets{y} = l, [\bot, \bot]}}
\end{mathpar}
Then, we can easily prove the \textsc{frame} rule by induction over the proof tree.
Similarly, for a language with more commands, we only need to guarantee proof rules for memory related operations are compatible with extra frames.

To summarize, there are also two ways to formalize the \textsc{frame} rule in a deep embedding: either directly admit it as a primary rule, or bake the \textsc{frame} rule into singleton rules for memory operations.
In both way, we need to admit the \textsc{frame} rule as part of the primary rules.

\subsection{Procedure Calls \& Hypothetical Frame Rule}
\label{sec:hypo-frame}

In this section, we further extend the language and logic in figure~\ref{fig:sep} with procedure calls.
Figure~\ref{fig:sep-call} adds the command for procedure calls, $\text{call} \ f()$, where procedure calls have no arguments or return value.
Function arguments and return values are written to specific locations to be transferred between the caller and the callee.

Different from logics in previous sections, the judgement of the logic in figure~\ref{fig:sep-call} takes the form of $\Delta \vdash \triple{P}{c}{Q, [\vec{R}]}$.
Here, $\Delta$ is the function hypothesis (or type-context in VST), containing a list of Hoare triples $\triple{P_i}{k_i}{Q_i}$ specifying the pre/post-condition of the procedure with identifier $k_i\in\text{FuncID}$.
With the function hypothesis, \textsc{hoare-call} can determine the effect of procedure call by looking up the callee's specification in $\Delta$.

\begin{figure}[h]
\paragraph{While-CF with Memory Operations and Procedure Calls}
$$
\begin{array}{rcl}
c \in \text{command} &:=& \cdots \ \mid \ x = [y] \ \mid \ [x] = y \ \mid \ \text{call} \ f()
\end{array}
$$
\paragraph{Procedure Call Proof Rules}
\begin{mathpar}
    \inferrule[hoare-call]{}{
        \triple{P_i}{k_i}{Q_i}_{\text{for } i \leq n} \vdash \triple{P_i}{ \ \ \text{call} \ k_i \ \ }{Q_i, [\bot, \bot]}
    }
    \and
    \inferrule[hypo-frame]{
        c \text{ modifies program variables freely occuring in } R \text{ only through } k_i \\\\
        \triple{P_i}{k_i}{Q_i}_{\text{for } i \leq n} \vdash \triple{P}{c}{Q, [Q_{\text{brk}}, Q_{\text{con}}]}
    }{
        \triple{P_i*R}{k_i}{Q_i*R}_{\text{for } i \leq n} \vdash \triple{P*R}{c}{Q*R,  [Q_{\text{brk}} * R, Q_{\text{con}} * R]}
    }
\end{mathpar}
    \caption{Additional Proof Rules for Separation Logics and Procedure Calls}
    \label{fig:sep-call}
\end{figure}

Shallow embeddings of judgement involving procedure calls\cite{o2004separation} are different from those in section~\ref{sec:embed}.
Taking the denotational semantics based shallow embedding $\vDash_b$ for example, we introduce $\eta$ in \eqref{eq:eta}, a mapping from procedures $k_i$ to their denotations, to avoid circularity.
\begin{equation}
    \eta\in \text{FuncID} \rightarrow \text{state} \rightarrow \text{state} \rightarrow \text{exit\_kind} \rightarrow \text{Prop}
\label{eq:eta}    
\end{equation}
We use $\eta(k_i)$ to lookup procedure's denotation and $c\doublebrackets{\eta}$ to represent the program after substituting call to procedures by their denotations.
The shallow embedding $\triple{P_i}{k_i}{Q_i}_{\text{for } i \leq n} \vDash_b \triple{P}{c}{Q, [\vec{R}]}$ 
is defined as:\linebreak
for any $\eta$, $\vDash_b \triple{P_i}{\eta(k_i)}{Q_i, [\bot, \bot]}$ for all $i \leq n$ can imply $\vDash_b \triple{P}{c\doublebrackets{\eta}}{Q, [\vec{R}]}$.



In a logic combining separation logic and procedure call, we can generalize the \textsc{frame} rule to the \textsc{hypothetical-frame} rule in figure~\ref{fig:sep-call}, which was first proposed by O'Hearn \textit{et. al} \cite{o2004separation}.
This rule allows extend specifications of both the main program $c$ and procedures $k_i$ it invokes by composing an extra frame $R$ to their preconditions ($P_i$ and $P$) and postconditions ($Q_i$ and $Q$).
It is easily derivable in a deeply embedded logic by induction over the proof tree with modified atomic rules in section~\ref{sec:frame}, and it is proved for the shallow embedding we introduced in this section by O'Hearn \textit{et. al} \cite{o2004separation}.


\subsection{Total Correctness}
\label{sec:discuss-total}
Throughout the paper, we have been discussing the Hoare logic for partial correctness, where we only require the program to produce no error and its output and ending state satisfying certain assertions.
Another correctness property that people care about a program is its total correctness.
The total correctness adds one more condition on top of the partial correctness: the execution of the program terminates.

In a simple While language with only assignment, sequential composition, if-statement, and while-loop, the only way to cause non-termination is executing an infinite-loop.
A famous approach to prove the termination of a loop is through loop variant, and it is easy to instantiate in a Hoare logic \cite{manna1974axiomatic}.
For example, the \textsc{while-total} rule below uses an expression $e$ as the loop-variant, which evaluates to some nature number.
The value of $e$ is guaranteed to decrease and when it reaches 0, the loop terminates.
With the \textsc{while-total} rule and other standard proof rules for other language constructs (e.g., those in figure~\ref{fig:setting}), a program logic can reason about the total correctness of the While language.
\begin{mathpar}
    \inferrule*[left=while-total]{
        I \land \doublebrackets{b} = \btrue \vdash \doublebrackets{e} > 0 \\
        \forall n. \vdash_{\text{total}} \triple{I \land \doublebrackets{b} = \btrue \land \doublebrackets{e} = n}{c}{I \land \doublebrackets{e} < n}
    }{
        \vdash_{\text{total}} \triple{I}{\cwhile{b}{c}}{I \land \doublebrackets{b} = \bfalse}
    }
\end{mathpar}

In practice, people may care how long does a program terminates rather than just whether a program terminates.
Time complexity analysis and cost analysis can also be integrated into a Hoare logic.
One approach to accomplish that is to reason about time credits, a resource that each step of execution needs to consume, in the pre-/post-conditions of a Hoare triple.
Chargu{\'e}raud \textit{et. al} \cite{chargueraud2019verifying} formalize a separation logic framework with time credits and use it to verify the amortized complexity of a union-find implementation.
We may equipped the proof rule for the while-loop with time credits as \textsc{while-credit} rule\footnote{
    Different from previous loop rules, \textsc{while-credit} is more like the \textsc{while-unroll1} rule which verifies one iteration of the loop per application of the rule.
    Provers need to induct over the number of total iteration in the meta-logic to use this proof rule.
} below.
The time credit is represented by $\$n$ in the assertion for some positive integer $n$.
To verify each iteration of a loop, the user of this rule need to split the time credit resource into two parts ($\$n_1$ and $\$n_2$), and this iteration will consumes one part of the credit $\$n_1$ while the remaining iteration will consume the other $\$n_2$.
A loop can finish within certain steps represented by a time credit, if the prover can find a way to dispense the time credit correctly to each iteration.
\begin{mathpar}
    \inferrule*[left=while-credit]{
        \vdash_{\text{credit}} \triple{I \land \doublebrackets{b} = \btrue \land \$n_1 * \$n_2}{c}{I * \$n_2} \\
        \vdash_{\text{credit}} \triple{I * \$n_2}{\cwhile{b}{c}}{I \land \doublebrackets{b} = \bfalse}
    }{
        \vdash_{\text{credit}} \triple{I * \$n_1 * \$n_2}{\cwhile{b}{c}}{I \land \doublebrackets{b} = \bfalse}
    }
\end{mathpar}
As we can see, the time credit approach is also a kind of the loop-variant.
But it is more expressive and can assert the upper bound of a program's execution time.

This paper mainly focuses on partial correctness, but we believe many extended rules also apply to total correctness verifications, which is left as a future work.

\subsection{Non-Determinism}
\label{sec:discuss-nondeterm}

Non-determinism is also a feature people want to support with Hoare logic, especially when verifying non-deterministic programs like concurrent programs.
There are mainly two types of non-determinism related to a programming language: the demonic non-determinism and the angelic non-determinism \cite{back2012refinement}.

In a program, demonic non-determinism is caused by commands that may non-deterministically yield different behaviors that could falsify the specification, while angelic non-determinism is caused by commands that may choose among different behaviors the one that could satisfy the specification.
For example, the simple programming language below consists of both demonic ($\sqcap$) and angelic ($\sqcup$) non-determinism.
$$
c \in \text{ command } := \cskip \lsep x = e \lsep c_1 \cseq c_2 \lsep c_1 \sqcap c_2 \lsep c_1 \sqcup c_2
$$
The \textsc{demonic} and \textsc{angelic{$_i$}} rule clearly explains the difference between two types of non-determinism.
When verifying a demonic choice $c_1 \sqcap c_2$ satisfying some Hoare triple, the prover need to prove that both choices satisfies it, otherwise, the program can non-deterministically choose to execute the one that does not satisfy the triple.
But when verifying an angelic choice $c_1 \sqcup c_2$, the prover only need to show that one of the choices are correct, and the program will non-deterministically choose to execute the one that satisfies the triple.
\begin{mathpar}
\inferrule*[left=demonic]{
    \vdash \triple{P}{c_1}{Q} \\
    \vdash \triple{P}{c_2}{Q}
}{
    \vdash \triple{P}{c_1 \sqcap c_2}{Q}
}
\and
\inferrule*[left=angelic{$_i$}]{
    \vdash \triple{P}{c_i}{Q}
}{
    \vdash \triple{P}{c_1 \sqcup c_2}{Q}
}
\end{mathpar}

However, in most of realistic programming languages, there are only counterparts for demonic non-deterministic, e.g., random number generator and concurrent interleaving.
The angelic non-determinism is often integrated when developing a program logic for the programming language.
For example, when verifying concurrent programs, verifiers use auxiliary states (also known as ghost states) to support their proof and they will insert auxiliary code at certain lines of the original code \cite{}.
These auxiliary code can only update auxiliary states and their purpose is to maintain certain invariant of the entire concurrent system.
Before proving concurrent programs, these auxiliary code are non-deterministic (no one knows what code will be inserted at each line).
When proving these programs, provers will choose auxiliary code so that these code will help establish specifications, and therefore, they are angelic non-deterministic.

Iris \cite{jung2018iris} takes one step further by embedding these ghost updates as the so-called frame preserving updates in the assertion language.
Informally, the frame preserving update $\fpupdate P$ means that after some ghost transitions admitted by the underlying state, the new state will satisfy the assertion $P$.
Below shows Iris's weakest precondition in section~\ref{sec:embed-small} equipped with the frame preserving update.
Different from the previous definition, here, when given the iProp of the pre-state $S(\sigma)$, we can first make some frame preserving updates, after which the expression $e$ executes one step on the physical memory.
After this one step execution, we can also make some frame preserving updates to the ghost state so that the weakest precondition for the remaining expression $\WPRE\, e'\, \{\Phi\}$ is satisfied on the new physical state and ghost state.
Provers can choose angelically which ghost transition to make when proving $\fpupdate P$, and thus making it angelic non-deterministic.
$$
\begin{aligned}
    \WPRE\, e\,  \{\Phi\}&\triangleq
           (e \in \textit{Val} \land {\color{blue}\fpupdate} \Phi(e)) \\
    &\lor \bigl(\forall \sigma.\, e \notin \textit{Val} \land S(\sigma) \wand {\color{blue}\fpupdate} \bigl(\text{reducible}(e, \sigma) \notag \\
    &\quad \land \later
        \forall e', \sigma'.\, \left((e, \sigma) \tred (e', \sigma')\right) \wand
    {\color{blue}\fpupdate} \left(S(\sigma') * \WPRE\, e'\, \{\Phi\}\right)
    \bigr)
    \bigr)
\end{aligned}
$$

Although this paper's logic formalization does not consider angelic non-determinism, we do allow demonic non-determinism.
Firstly, three shallow Hoare logic embeddings in section~\ref{sec:embed} requires any post-state reachable in the execution satisfy certain properties\footnote{
    Any post-state need to satisfy the post-condition, the weakest precondition, and the $\text{safe}(-,-,-)$ predicate respectively in three shallow embeddings.
}.
Therefore, if a program has demonic non-determinism, the logic will ask provers to verify all possible executions of the program to be correct.
Secondly, all extended proof rules discussed in the paper do not require the program to be deterministic, and moreover, our proofs of extended rules under all embeddings do not rely on the determinism of the program.


\subsection{Impredicative Assertions}
\label{sec:discuss-impred}

Impredicative assertions are assertions whose universe quantifiers and existential quantifiers can quantify over assertions variables.
And these assertion languages should also admits Hoare triples as assertions.
This feature is required for our \textsc{hoare-ex} proof in section~\ref{sec:proof-deep} to work.

For a shallowly embedded assertion language, which we use for our proofs in section~\ref{sec:proof}, it is easy to formalize these impredicative assertions because we can inject meta-logic propositions into assertions.
Definitions \eqref{eq:inter-ex} and \eqref{eq:inter-tripl} shows how to define the interpretation (shallow embeddings) of these assertions.
\begin{gather}
\sigma \vDash \exists P_0. S(P_0) \ \text{iff.} \ \text{there exists an assertion } P_0 \ \text{s.t.} \ \sigma \vDash S(P_0)
\label{eq:inter-ex}
\\
\sigma \vDash M \ \text{iff.} \ M \ \text{where} \ M \ \text{is a meta-logic proposition like } \vdash\triple{P}{c}{Q, [\vec{R}]}
\label{eq:inter-tripl}
\end{gather}

If we would use a deeply embedded assertion language, we need to construct the counterpart of the provable judgement and higher-order quantifiers in the assertion language.
This is possible, and defining such a syntax system is not hard. The following definition could be a reasonable candidate:
$$
\begin{array}{rcl}
x & \in & \text{individual-variables-name} \\
A & \in & \text{assertion-variable-name} \\
t & ::= & \dots \ \mid \ x \\
P & ::= & t = t \ \mid \ \neg P \ \mid \ P \wedge P \ \mid \ \exists x. \ P \ \mid \ \exists A. \ P \ \mid \ \ \vdash \triple{P}{c}{P, [P, P]}
\end{array}$$

Technical problems may appear when defining assertions' interpretation $\sigma \vDash P$.
A simple syntactical substitution of the quantified variable for its value will result in a non-decreasing recursive interpretation function.
For example, in the definition below, the assertion $P'$ that substitutes $A$ in $P$ may be larger than $P$, which makes the interpretation function non-terminating.
$$
\sigma \vDash \exists A. P \iff \text{there exists assertion } P' \text{ such that } \sigma \vDash P[A/P']
$$

Instead of a direct syntactical substitution, we can use an interpretation assignment $J$ to store interpretations of free variables.
$$
J \in \text{assertion-variable-name} \rightarrow (\text{state} \rightarrow \text{Prop})
$$
We use the judgement $\sigma, J \vDash P$ to help define the interpretation of an assertion.
For higher-ordered existential quantifiers, we choose in the meta-level the interpretation $d$ of the assertion variable $A$.
And then assign $d$ to $A$ in the assignment $J$.
$$
\sigma, J \vDash \exists A. P \iff \text{there exists } d \in \text{state} \rightarrow \text{Prop} \text{ such that } \sigma, J[A/d] \vDash P
$$
This makes an assertion to be potentially an open term with free variable $A$ occurring in $P$.
When we need to interpret a free assertion variable $A$ on a state $\sigma$, we simply query $J$ for $A$'s interpretation function and apply it to $\sigma$.
$$
\sigma, J \vDash A \iff J[A](\sigma)
$$
It is also easy to interpret Hoare triple assertions because we can directly interpret them as their validities (shallow embeddings).
Taking big-step based embedding for example, we simply copy and paste the definition in section~\ref{sec:embed-big} and add assignment $J$ when interpreting pre-/post-conditions.
$$
\sigma, J \vDash (\vdash \triple{P}{c}{Q, [\vec{R}]}) \iff
\left(
\begin{aligned}
& \text{for all $\sigma_1,J \vDash P$, } \lnot\,(c, \sigma_1)\Uparrow\\
& \text{and for all } \ek, \sigma_2, \text{if }(c, \sigma_1) \Downarrow (\ek, \sigma_2) \\
& \quad\text{then }\ek = \epsilon \text{ implies } \sigma_2,J \vDash Q\\
& \quad\text{and }\ek = \ekb \text{ implies } \sigma_2,J \vDash R_\ekb\\
& \quad\text{and }\ek = \ekc \text{ implies } \sigma_2,J \vDash R_\ekc
\end{aligned}
\right)
$$
For a closed assertion $P$, we will interpret it on any assignment because it should not depend on it.
$$
\sigma \vDash P \iff P \text{ is closed and for any } J,  \sigma, J \vDash P
$$

Another approach is to inductively define the interpretation of assertions as a deep embedding, which is proposed by Zhaozhong and Zhong \cite{ni2006certified}.
For example, the interpretation rule for the existential quantifier is the following.
\begin{mathpar}
\inferrule*[]{
    P' \in \text{assertion} \and
    \sigma \vdash P[A/P']
}{
    \sigma \vdash \exists A. P
}
\end{mathpar}
This definition only defines the interpretations for assertions that have finite derivation trees.
Assertions with unbounded size (brought by higher-ordered quantifiers) will not have a finite derivation tree as well as an interpretation.
This deeply embedded interpretation is then proved sound w.r.t. assertions true meaning.
Interested readers may refer to their work \cite{ni2006certified} for more details.

\subsection{The Soundness Proof Technique}
\label{sec:discuss-sound}

Section~\ref{sec:nomen-logic} mentioned a soundness proof technique, where logic developers first prove the logic is sound w.r.t. some auxiliary validity definition $\VDash S$, i.e.,
\begin{equation}
\label{eq:aux-sound}
\text{forall } \vdash S \text{ implies } \VDash S,
\end{equation}
and then prove this auxiliary validity implies the real validity $\vDash S$, i.e..
\begin{equation}
\label{eq:real-sound}
\text{forall } \VDash S \text{ implies } \vDash S.
\end{equation}
The auxiliary validity $\VDash S$ differs from the real validity $\vdash S$ in that $\VDash S$ usually has extra information about programs execution which makes proving \eqref{eq:aux-sound} easier, while this information is not available when proving the real validity $\vDash S$ and therefore a lemma \eqref{eq:real-sound} that erases these information is required.

For example, in Brookes's soundness proof for the concurrent separation logic \cite{brookes2007semantics}, they use the thread local enabling as the semantics that defines the auxiliary validity.
The thread local enabling semantics is annotated with available resource invariants, and whenever an action needs access to shared resources, it will access these invariant annotations and make sure it can still preserve the invariant when it finishes.
In this way, all threads will respect the invariant and cooperate property when executing concurrently under the thread local enabling semantics.
These invariant notations come directly from the concurrent separation logic judgement, and therefore, it is easy to prove that any provable judgement with invariant $\Gamma$ is valid under the thread local enabling semantics with the invariant annotation $\Gamma$.
However, this annotation is not available in the machine semantics, the real semantics of a program's execution on a machine, since the machine will not tell any thread what invariant they should respect.
Therefore, Brookes uses a connection property to show that erasing the invariant annotation in a thread local enabling reduction produce an equivalent machine semantics reduction, and thus the auxiliary validity can induce the real validity.

In some works, e.g., in Iris \cite{jung2018iris}, adequacy theorem serves the purpose of proving the logic's soundness w.r.t. the real validity.
Iris uses shallowly embedded weakest precondition to embed the Hoare logic, which allows angelic updates to ghost states when the program takes a step in the physical memory.
Only with both ghost states and physical memory can a weakest precondition be established, while a program is supposed to be correct with only access to physical memories.
Therefore, they use an adequacy theorem, rephrased under the sequential execution in theorem~\ref{thm:adequacy}, to derive the real validity of a whole program $e$ with no ghost component in the precondition $\text{True}$ and the post-condition $\Phi$ (since it is a first-order predicate).
As a result, theorem~\ref{thm:adequacy} indicates that any whole program Hoare triple proved in Iris is valid under the machine semantics and without any ghost updates.
\begin{thm}[Iris Adequacy]
\label{thm:adequacy}
Let $\Phi$ be a first-order predicate. If $\text{True} \wand \WP\,e\,\{\Phi\}$ is true and for any $\sigma, e', \sigma'$ such that $(e, \sigma) \rightarrow_t^* (e', \sigma')$, then
\begin{itemize}
    \item Either $(e', \sigma')$ can be further reduced.
    \item Or $e'$ is some terminal value $v$, and $\Phi(v)$ is true.
\end{itemize}
\end{thm}

To conclude, proving soundness of a program logic with advanced features like concurrency can be challenging.
And one may need to introduce a layer of auxiliary validity in the soundness proof.

\section{Real Projects}
\label{sec:project}




Besides the program logic for our WhileCF toy language discussed so far, we also implement proposed theories in two real verification projects.
Specifically, in section~\ref{sec:project-vst}, we use the scheme in section~\ref{sec:d2s} to lift originally shallowly embedded VST\cite{VST} into a deeply embedded VST, where extended rules are easily proved.
In section~\ref{sec:project-iris-shallow}, we reproduce the weakest precondition based shallow embedding with control flows (section~\ref{sec:embed-small}) in Iris\cite{jung2018iris} as the logic Iris-CF.
We then apply the lifting scheme in section~\ref{sec:d2s} to this shallowly embedded Iris-CF to obtain a deeply embedded logic, Iris-Imp\footnote{
    Iris-CF and Iris-Imp are only for demonstration of our theory.
    Users of Iris may prefer using the weakest precondition reasoning style instead of the Hoare logic discussed in this paper.
    But Iris-CF and Iris-Imp are still good examples to show how to lift weakest precondition based shallow embeddings to deep embeddings, and may benefit developers of other framework using this embedding.
}.

\subsection{Deeply Embedded VST}
\label{sec:project-vst}


We build deeply embedded VST\footnote{
    Our deeply embedded VST has been integrated into the VST repository (https://github.com/PrincetonUniversity/VST) and we include the version when it was first integrated.
    It consists of two files, \texttt{SeparationLogicAsLogic.v} and \texttt{SeparationLogicAsLogicSoundness.v} in the \texttt{floyd} folder.
} based on the shallowly embedded VST, a verification framework for C programs.
It is originally designed to use continuation based shallow embedding to formalize the Hoare logic.
Different from our toy language and toy logic, VST is designed to support named function invocations as well and we demonstrate the details they add to the triple.


VST uses the continuation based shallow embedding with small-step semantics defined as $ge \vdash (c, \kappa, m) \rightarrow (c', \kappa', m')$.
Different from our small-step semantics in section~\ref{sec:embed}, it contains a global environment $ge$ that stores external variables and functions.
When external call occurs, it will look up the function body from $ge$ and add it to the continuation $\kappa$.

VST's triples take the form of $\Delta \vDash \triple{P}{c}{Q, [\vec{R}]}$, where $\Delta$ is the type-context containing function specifications.
Its formalization in Coq is similar to the continuation based shallow embedding ($\vDash_{\text{sep}-c}$) discussed in this paper.
Specifically, the triple means: if every function \texttt{obeys} specifications in $\Delta$, then for all continuations $\kappa$ and all separation logic frames $F$,
$$
\text{if }
\begin{cases}
    \guard{Q * F}{\cskip}{\kappa}\\
    \text{and }\guard{R_\ekb * F}{\cmdbreak}{\kappa}\\
    \text{and }\guard{R_\ekc * F}{\ccontinue}{\kappa}\\
    \text{and }\guard{R_\ekr * F}{\cmdreturn}{\kappa}
\end{cases},
\text{ then } \guard{P * F}{c}{\kappa}.
$$
A function $f$ \texttt{obeys} specification $\triple{P}{f}{Q}$ in $\Delta$ iff. its function body $c$ has $\Delta \vDash \triple{P}{c}{Q}$.
VST uses step-index to resolve the circularity in the definition.

The shallowly embedded VST does not provide extended rules like \textsc{seq-inv} and \textsc{nocontinue}, which we prove in section~\ref{sec:proof-cont} and find their proof construction challenging.
We then seek to formalize a deeply embedded Hoare logic.
As proposed in section~\ref{sec:d2s}, we first define the provable predicate for Hoare triple by primary proof rules (atomic rules, compositional rules, and the consequence rule) and then directly use the validity definition and soundness proofs in the original shallowly embedded VST to establish the soundness of our deeply embedded VST.
We then prove extended rules in the deep embedding by induction over proof tree.
As discussed in section~\ref{sec:d2s} and section~\ref{sec:frame-dep}, we choose to treat \textsc{hoare-ex} and \textsc{frame} rule as extended rules and prove them along with other extended rules, although they were proved sound in the original shallowly embedded VST.
Also, VST did support proof rules like \textsc{hoare-load-frame} and \textsc{hoare-store-frame} in its original development. Thus, those modifications discussed in the end of section~\ref{sec:frame-dep} are not even needed.

\vspace*{10pt}
\noindent\textit{\textbf{Evaluation.}}
\vspace*{5pt}

\linespread{1.1}
\begin{table}[h]
    \centering
    \begin{tabular}{|c|c|}
        \hline
        & lines of code \\\hline
        \begin{tabular}{c}
            Shallow Embedding Proofs of Primary Rules
        \end{tabular} & 12068 \\\hline
        \begin{tabular}{c}
            Deep Embedding Definition
        \end{tabular} & 297 \\\hline
        \begin{tabular}{c}
            Deep Embedding Soundness Proof
        \end{tabular} & 261 \\\hline
        \begin{tabular}{c}
            Deep Embedding Extended Rule Proofs
        \end{tabular} & 2114 \\\hline
        \begin{tabular}{c}
            Modification to User Proofs
        \end{tabular} & \underline{\textbf{0}} \\\hline
    \end{tabular}
    \caption{Evaluation of Deeply Embedded VST}
    \label{tab:eval}
\end{table}
\linespread{1.0}

We demonstrate evaluation of our deeply embedded VST in table~\ref{tab:eval}.
It only takes very few lines of code to formalize primary rules and logic soundness by reusing original proofs in the shallow embedding.
We are able to formalize extended rules in the deep embedding with reasonable lines of code, which was not available in the shallow embedding.
Moreover, there is no need to modify user proofs and users can upgrade their projects depending on the shallowly embedded VST to our deeply embedded one for free.

\subsection{Shallowly Embedded Iris-CF \& Deeply Embedded Iris-Imp}
\label{sec:project-iris-shallow}


Iris is a representative proof system that uses small-step semantics and weakest precondition to embed their program logic.
There have been works to extend Iris with control flow reasoning \cite{sammler2021refinedc, timany2019mechanized} but not using the multiple post-condition.
So we first equip Iris with it and present Iris-CF so that we can have an instance of weakest precondition based embedding in this paper and discuss the situation of extended rules in Iris-CF.
To evaluate our lifting approach in section~\ref{sec:d2s}, we lift Iris-CF to a deeply embedded program logic Iris-Imp.

Iris-lambda is a ML-like language for many Iris project's demonstration, which can be extended to more realistic programming languages (e.g., the $\lambda_{\texttt{Rust}}$ in the RustBelt\cite{jung2017rustbelt}).
To evaluate our theory in the Iris framework, we extended Iris-lambda into Iris-CF:
$$
\begin{aligned}
    v \in \textit{Val} ::=\, & () \lsep z \lsep \btrue \lsep \bfalse \lsep l \lsep \lambda x.e \lsep \cdots \\
    e \in \textit{Expr} ::=\, & v \lsep x \lsep e_1(e_2) \lsep \cmdref(e) \lsep !e \lsep e_1 \leftarrow e_2 \lsep \cmdfork{e} \lsep {e_1 \cmdseq e_2} \lsep \underline{\cmdloop{e} e} \\
    \lsep\, & \cmdif e_1 \cmdthen e_2 \cmdelse e_3 \lsep \underline{\cmdbreak e} \lsep \underline{\cmdcontinue} \lsep \underline{\cmdcall e} \lsep \underline{\cmdreturn e} \lsep \cdots \\
    K \in \textit{Ctx} ::=\, & \bullet \lsep K(e) \lsep v(K) \lsep \cmdref(K) \lsep !K \lsep K \leftarrow e \lsep v \leftarrow K \lsep {K \cmdseq e} \lsep \underline{\cmdloop{e} K} \\
    \lsep\, & \cmdif K \cmdthen e_2 \cmdelse e_3 \lsep \underline{\cmdbreak K}\lsep \underline{\cmdcall K} \lsep \underline{\cmdreturn K} \lsep \cdots
\end{aligned}
$$
\underline{Underlined} ones are what we add to the language to support control flow. 
We refer the reader to appendix for its semantics and Iris's original paper \cite{jung2018iris} for explanations of Iris-lambda's features.

\begin{figure}[h]
\centering
$$
\begin{aligned}
    \wpre{e}{\Phi_N}{\Phi_B}{\Phi_C}{\Phi_R} \triangleq&
            (e \in \textit{Val} \land \upd \Phi_N(e)) \\
    \lor & (\exists v \in \textit{Val}.\, e = \cmdbreak v \land \upd \Phi_B(v)) \\
    \lor & (e = \cmdcontinue \land \upd \Phi_C()) \\
    \lor & (\exists v \in \textit{Val}.\, e = \cmdreturn v \land \upd \Phi_R(v)) \\
    \lor & \biggl(\forall \sigma.\, e \notin \text{terminals} \land S(\sigma) \wand \upd \bigl(\cred(e, \sigma) \\
    & \land \later
        \forall e', \sigma', \vec{e}_f.\, \bigl((e, \sigma) \tred (e', \sigma', \vec{e}_f)\bigr) \wand \upd \\
    & \quad \bigl(S(\sigma') * \wpre{e'}{\Phi_N}{\Phi_B}{\Phi_C}{\Phi_R} \\
    & \quad\quad * \mathlarger{\mathlarger{\mathlarger{*}}}_{e'\in\vec{e}_f} \textsf{WP}\, e' \progspec{\top} \bigr)\bigr)\biggr)
\end{aligned}
$$
\caption{Weakest Precondition in Iris-ControlFlow}
\label{fig:wp-icf}
\end{figure}

We define the weakest precondition in Iris logic with control flow in figure~\ref{fig:wp-icf}.
Hoare triple's validity is defined by weakest precondition as we demonstrated in related projects in section~\ref{sec:embed-small}.
The definition is consistent with our previous weakest precondition based shallow embedding in section~\ref{sec:embed-small}, where we have 4 branches for 4 different terminals, value, break, continue, and return terminals, and 1 branch for the preservation case.
The definition uses many Iris notations and features.
For example, it has later modality $\later$, which is necessary in Iris's step-indexed model to prevent circularity; it has the update modality $\upd$, which allows angelic updates as mentioned in section~\ref{sec:discuss-nondeterm}; it includes weakest-preconditions for forked threads so that it can reason about concurrency.
Most of them are not strongly related to the topic of this paper, and we refer readers to Iris's original paper \cite{jung2018iris} for a detail explanation.


All primary proof rules in section~\ref{sec:setting} are established in Iris-CF and figure~\ref{fig:iris-cf-rule} shows a snippet of it.

\begin{figure}[h]
\centering
\begin{mathpar}
\inferrule[val]{}{
    \vdash \triple{P(v)}{v}{P, [\bot, \bot, \bot]}
}
\and
\inferrule[break]{
    \vdash \triple{P}{e}{Q}
}{
    \vdash \triple{P}{\cmdbreak e}{\bot, [Q, \bot, \bot]}
}
\and
\inferrule[continue]{}{
    \vdash \triple{P}{\cmdcontinue}{\bot, [\bot, \lambda\_.\,P, \bot]}
}
\and
\inferrule[return]{
    \vdash \triple{P}{e}{Q}
}{
    \vdash \triple{P}{\cmdreturn e}{\bot, [\bot, \bot, Q]}
}
\and
\inferrule[seq]{
    \vdash \triple{P}{e_1}{\lambda\_.\,Q, [\vec{R}]}
    \and
    \vdash \triple{Q}{e_2}{\Phi, [\vec{R}]}
}{
    \vdash \triple{P}{e_1 \cseq e_2}{\Phi, [\vec{R}]}
}
\and
\inferrule[call]{
    \vdash \triple{P}{e}{Q, [\bot, \bot, Q]}
}{
    \vdash \triple{P}{\cmdcall}{Q, [\bot, \bot, \bot]}
}
\and
\inferrule*[left=loop]{
    \vdash \triple{I}{e}{\lambda\_.\,I, [Q, \lambda\_.\,I, R]}
}{
    \vdash \triple{I}{\cmdloop{e} e}{Q, [\bot,\bot,R]}
}
\and
\inferrule*[left=if]{
    \vdash \triple{P}{e_1}{\lambda v.\, (v = \texttt{true} \land R_1) \lor (v = \texttt{false} \land R_2), [\vec{R}]}
    \\\\
    \vdash \triple{R_1}{e_2}{Q, [\vec{R}]}
    \and
    \vdash \triple{R_2}{e_3}{Q, [\vec{R}]}
}{
    \vdash \triple{P}{\cmdif e_1 \cmdthen e_2 \cmdelse e_3}{Q, [\vec{R}]}
}
\end{mathpar}
\caption{A Snippet of Primary Rules in Iris-CF}
\label[]{fig:iris-cf-rule}
\end{figure}

\paragraph{Extended Rules in Shallowly Embedded Iris-CF.}

Almost all extended rules from section~\ref{sec:rules} are valid\footnote{
    The \textsc{nocontinue} rule has an additional condition that there are no unscoped \texttt{continue} command in heap, because otherwise, in a higher-ordered language like Iris-lambda, these \texttt{continue} will be loaded into the program and invalidate the rule's conclusion.
    We do not consider \textsc{loop-nocontinue} in Iris because it does not use \texttt{for}-loops in the language.
}.
However, the \textsc{seq-inv} rule is not true due to the step-indexed model of Iris.
Specifically, we need to match the later modality for programs on the two side of the simulation relation, which means in $(c_1, \kappa_1) \sim (c_2, \kappa_2)$, each step of $(c_2, \kappa_2)$ must simulate at least one step of $(c_1, \kappa_1)$.
In other words, in the \textbf{termination} case, if the \textit{source} terminates, the \textit{target} must terminate immediately.
One special case of \textsc{seq-inv} is
$$
\begin{array}{rl}
    & \textsf{WP}\, v \cseq c \progspec{\Phi} \wand \triangleright \textsf{WP}\, c \progspec{\Phi} \\
    \Leftrightarrow & \big(\forall \sigma.\, S(\sigma) \wand \cdots \wand (S(\sigma) * \textsf{WP}\, c \progspec{\Phi})\big) \wand \triangleright \textsf{WP}\, c \progspec{\Phi}
\end{array}
$$
where $v \cseq c$ has one more step than $c$.
If the current resource conflicts with all $S(\sigma)$, which is possible in Iris because provers can choose the resource algebra, then the premise is a tautology ($\text{False} \wand \cdots \wand (S(\sigma) * \textsf{WP}\, c \progspec{\Phi})$), but we can always find a $c$ whose weakest precondition violates the resource we choose.
Therefore, the special case here can not be proved, which is necessary in the \textsc{seq-inv} proof.






\paragraph{Deeply Embedded Iris-Imp.}

We then lift the shallowly embedded Iris-CF into the deeply embedded Iris-Imp \cite{exrule-repo} for an imperative language:
$$
\begin{aligned}
    e' \in \textit{Expr}_{\textit{imp}} ::=\,
    & v \lsep x \lsep \texttt{ref}'(e') \lsep !'e' \lsep e'_1 \leftarrow' e'_2 \lsep \texttt{fork}'\{e'\} \\
    \lsep\, 
    & {e'_1 \cmdseq e'_2} \lsep \cmdif' e'_1 \cmdthen e'_2 \cmdelse e'_3 \lsep {\texttt{loop}'_{e'} e'} \\
    \lsep\, 
    & {\cmdbreak' e'} \lsep {\cmdcontinue'} \lsep {\cmdcall' e'} \lsep {\cmdreturn' e'} \lsep \cdots
\end{aligned}
$$
which is directly encoded using Iris-CF's expressions.
The $\Vert e'\Vert$ assigns the encoding to Iris-Imp commands and we can then reuse Iris-CF's semantic definitions and logic judgement for developing the validity of Iris-Imp logic (definition~\ref{def:deep-iris-valid} and definition~\ref{def:deep-iris-logic}).
$$
\begin{aligned}
    &\Vert v \Vert \triangleq v &&
    \Vert e'_1 \cseq e'_2 \Vert \triangleq (\lambda \_.\, \Vert e'_2 \Vert) (\Vert e'_1 \Vert) \\
    &\Vert \texttt{ref}'(e') \Vert \triangleq \texttt{ref}(\Vert e'\Vert) &&
    \Vert {\texttt{loop}'_{e'} e'} \Vert \triangleq {\texttt{loop}_{\Vert e' \Vert} \Vert e'\Vert} \\
    &\cdots &&\cdots
\end{aligned}
$$
\begin{definition}[Validity]
    \label{def:deep-iris-valid}
    $\vDash_{\text{iris-imp}} \triple{P}{e'}{Q, [\vec{R}]}$ is true, iff., $\vDash_{\text{iris-cf}} \triple{P}{\Vert e' \Vert}{Q, [\vec{R}]}$ is true, i.e., $P \wand \textsf{WP}\,\Vert e'\Vert\{Q, [\vec{R}]\}$ is true.
\end{definition}
\begin{definition}[Iris-Imp Logic]
    \label{def:deep-iris-logic}
    $\vdash_{\text{iris-imp}} \triple{P}{e'}{Q, [\vec{R}]}$ is true, iff., the triple can be constructed by primary rules of Iris-CF (in figure~\ref{fig:iris-cf-rule}) with the \textsc{bind} rule removed and expressions changed to those in Iris-Imp.
\end{definition}

We can easily prove its soundness in theorem~\ref{thm:deep-iris-sound} by reusing the soundness proof of Iris-CF, because we have defined the validity of Iris-Imp as an encoding of the validity of Iris-CF.
Moreover, we can prove all extended rules (theorem~\ref{thm:deep-iris-extend}) by following proofs for the deep embedding in section~\ref{sec:rules}.

\begin{thm}[Soundness of Iris-Imp Logic]
    \label{thm:deep-iris-sound}
    The deeply embedded logic Iris-Imp is sound, i.e., $\vdash_{\text{iris-imp}} \triple{P}{e'}{Q, [\vec{R}]}$ implies $\vDash_{\text{iris-imp}} \triple{P}{e'}{Q, [\vec{R}]}$.
\end{thm}
\begin{thm}[Extended Rules of Iris-Imp Logic]
    \label{thm:deep-iris-extend}
    All three categories of extended rules in section~\ref{sec:rules} is true for the deeply embedded logic Iris-Imp.
\end{thm}

\section{Related Work}
\label{sec:related-embed}

In section~\ref{sec:embed}, we have already introduced various verification projects using different Hoare logic embeddings.
And in this section, we review other work on program logic embeddings.

Cook's \cite{cook1978soundness} proof of relative completeness for Hoare logic gives another view of problems we have discussed in Section~\ref{sec:proof}.
We now assume the deeply embedded logic $\vdash S$ comes with its completeness proof w.r.t. the shallowly embedded logic $\vDash S$, in stead of the incompleteness assumption in Section~\ref{sec:proof}.
With the completeness, the proof of inversion rules under all shallow embeddings can be greatly simplified.
Take \textsc{seq-inv} as an example, the completeness allows entailment from $\vDash \triple{P}{c_1\cseq c_2}{Q, [\vec{R}]}$ to $\vdash \triple{P}{c_1\cseq c_2}{Q, [\vec{R}]}$, and use the \textsc{seq-inv} under the deep embedding to get $\vdash \triple{P}{c_1}{S, [\vec{R}]}$ and $\vdash \triple{S}{c_2}{Q, [\vec{R}]}$.
Lastly, by the soundness of the deeply embedded logic, we get $\vDash \triple{P}{c_1}{S, [\vec{R}]}$ and $\vDash \triple{S}{c_2}{Q, [\vec{R}]}$ and finish the proof.
The completeness enables proofs of this style for inversion rules, which can be further utilized to simplify proofs of structural rules and transformation rules.
Nevertheless, as we have mentioned in Section~\ref{sec:proof}, the completeness proof is often missing for logics in real verification projects, and our complicated constructive proofs under the shallow embedding should be trivial, when compared to the even more complex completeness proof.

Nipkow \cite{wildmoser2004certifying} compares shallowly and deeply embedded assertion languages in an early paper.
He considers a combination of shallowly embedded programming languages, shallowly embedded assertion languages and shallowly embedded program logics.
He also considers a combination of deeply embedded programming languages, deeply embedded assertion languages and executable symbolic execution functions.
His focus is the comparison between these two settings, while we compare multiple logic embeddings with fixed language and assertion embeddings.

PHOAS \cite{chlipala2008parametric} is a different formalization style beside shallow/deep embeddings. It is used for formalizing binders in languages.
If we can reason about triples within assertions in such formalization, our conclusions for the deep embedding may apply to it.

\section{Conclusion}
\label{sec:conclusion}

In this paper, we clarify the nomenclature of ``deep embedding'' and ``shallow embedding'' for formalizing programming languages, assertion languages, and program logics and explore a deep embedding and three different techniques to shallowly embed a program logic.
We identify a set of extended rules that could benefit the verification work, and we point out critical proof steps to validate them under different embeddings.
We find they are relatively easier to prove under the deeply embedded logic because we can use induction over the proof tree.
To alleviate the proof burden of instantiating extended rules for existing verification tools using shallow embeddings, we propose a method to build deeply embedded program logics on existing shallowly embedded ones, where we only need to reuse proofs for primary rules under the shallow embedding to build the deep one.

We also extend our results from conventional Hoare logics to separation logics. We present the deeply embedded VST, where we reuse the original shallowly embedded VST's proofs for soundness proofs of primary rules and prove extended rules under the deep embedding, which is much easier than the proofs in the original VST (e.g. \textsc{if-seq}, \textsc{loop-nocontinue}).
We also use our result to implement Iris-CF, which incorporates control flow reasoning into Iris, and we present both a shallowly embedded Iris-CF logic and a deeply embedded Iris-Imp logic.
With deeply embedded VST and Iris-Imp, we demonstrate the feasibility of our theory in real verification project.

In conclusion, this paper shows that using different embedding to formalize a program logic brings different benefits, where the shallow embedding allows a more straightforward formalization and the deep embedding makes it easier to add extended proof rules.
Our work also indicates the possibility of formalizing a deeply embedded logic based on a shallowly embedded logic by reusing its soundness proof.
Moreover, for any shallowly embedded logic (even with different embedding techniques), we can build its corresponding deeply embedded logic by the same set of primary rules.
This could be an efficient way to extend existing verification tools to support more powerful proof rules like our extended proof rules.



\bibliography{paper}


\newpage
\begin{appendices}





\renewcommand\thefigure{\arabic{figure}}

\section{Big-step semantics}
\label{sec:Abigs}

Figure~\ref{fig:complete-big-1} and figure~\ref{fig:complete-big-2} list the complete big-step semantics for the toy language.
We use $(c, \sigma) \Downarrow (\ek, \sigma')$ to denote that the program $c$ starts execution from state $\sigma$ will safely terminate in state $\sigma'$ with exit kind $\ek$.
We use $(c, \sigma) \Uparrow$ to denote that the program $c$ starts execution from state $\sigma$ will run into an error.

\setcounter{figure}{14}
\begin{figure}[H]
\begin{mathpar}
    \inferrule[ref]{
        \mathrm{eval}(e, \sigma) = v \and
        h(l) = \bot
    }{
        (x = \cmdref(e), (\sigma, h)) \Downarrow (\epsilon, (\sigma[x\mapsto l], h[l\mapsto v]))
    }
    \and
    \inferrule[assign]{
        \mathrm{eval}(e, \sigma) = v
    }{
        (x = e, (\sigma, h)) \Downarrow (\epsilon, (\sigma[x\mapsto l], h))
    }
    \and
    \inferrule[store]{
        \mathrm{eval}(e_1, \sigma) = l \and
        \mathrm{eval}(e_2, \sigma) = v \and
        h(l) \neq \bot
    }{
        ([e_1] = e_2, (\sigma, h)) \Downarrow (\epsilon, (\sigma, h[l \mapsto v]))
    }
    \and
    \inferrule[store\_fail]{
        \mathrm{eval}(e_1, \sigma) = l \and
        h(l) = \bot
    }{
        ([e_1] = e_2, (\sigma, h)) \Uparrow
    }
    \and
    \inferrule[load]{
        \mathrm{eval}(e, \sigma) = l \and
        h(l) = v
    }{
        (x = [e], (\sigma, h)) \Downarrow (\epsilon, (\sigma[x \mapsto v], h))
    }
    \and
    \inferrule[load\_fail]{
        \mathrm{eval}(e, \sigma) = l \and
        h(l) = \bot
    }{
        (x = [e], (\sigma, h)) \Uparrow
    }
    \and
    \inferrule[Seq1]{
        (c_1, \sigma_1) \Downarrow (\epsilon, \sigma_3) \and
        (c_2, \sigma_3) \Downarrow (\ek, \sigma_2)
    }{
        (c_1 \cseq c_2, \sigma_1) \Downarrow (\ek, \sigma_2)
    }
    \and
    \inferrule[Seq2]{
        (c_1, \sigma_1) \Downarrow (\ek, \sigma_2)
    }{
        (c_1 \cseq c_2, \sigma_1) \Downarrow (\ek, \sigma_2)
    }
    \and
    \inferrule[Seq\_Fail1]{
        (c_1, \sigma) \Uparrow
    }{
        (c_1 \cseq c_2, \sigma) \Uparrow
    }
    \and
    \inferrule[Seq\_Fail2]{
        (c_1, \sigma_1) \Downarrow (\epsilon, \sigma_2) \and
        (c_1, \sigma_2) \Uparrow
    }{
        (c_1 \cseq c_2, \sigma_1) \Uparrow
    }
\end{mathpar}
\caption{Big-step semantics of the Toy Language (1)}
\label{fig:complete-big-1}
\end{figure}

\begin{figure}[H]
\begin{mathpar}
    \inferrule[if\_true]{
        \mathrm{eval}(e, \sigma_1) = \text{true} \and
        (c_1, \sigma_1) \Downarrow (\ek, \sigma_2)
    }{
        (\cif{e}{c_1}{c_2}, \sigma_1) \Downarrow (\ek, \sigma_2)
    }
    \and
    \inferrule[if\_true\_fail]{
        \mathrm{eval}(e, \sigma_1) = \text{true} \and
        (c_1, \sigma_1) \Uparrow
    }{
        (\cif{e}{c_1}{c_2}, \sigma_1) \Uparrow
    }
    \and
    \inferrule[if\_false]{
        \mathrm{eval}(e, \sigma_1) = \text{false} \and
        (c_2, \sigma_1) \Downarrow (\ek, \sigma_2)
    }{
        (\cif{e}{c_1}{c_2}, \sigma_1) \Downarrow (\ek, \sigma_2)
    }
    \and
    \inferrule[if\_false\_fail]{
        \mathrm{eval}(e, \sigma_1) = \text{false} \and
        (c_2, \sigma_1) \Uparrow
    }{
        (\cif{e}{c_1}{c_2}, \sigma_1) \Uparrow
    }
    \and
    \inferrule[for]{
        \ek \text{ is not } \ekb \and
        (c_1, \sigma_1) \Downarrow (\ek, \sigma_3) \\\\
        (c_2, \sigma_3) \Downarrow (\epsilon, \sigma_4) \and
        (\cfor{c_1}{c_2}, \sigma_4) \Downarrow (\epsilon, \sigma_2)
    }{
        (\cfor{c_1}{c_2}, \sigma_1) \Downarrow (\epsilon, \sigma_2)
    }
    \and
    \inferrule[for\_fail1]{
        (c_1, \sigma_1) \Uparrow
    }{
        (\cfor{c_1}{c_2}, \sigma_1) \Uparrow
    }
    \and
    \inferrule[for\_fail2]{
        \ek \text{ is not } \ekb \\\\
        (c_1, \sigma_1) \Downarrow (\ek, \sigma_2) \and
        (c_2, \sigma_2) \Uparrow
    }{
        (\cfor{c_1}{c_2}, \sigma_1) \Uparrow
    }
    \and
    \inferrule[for\_break]{
        (c_1, \sigma_1) \Downarrow (\ekb, \sigma_2)
    }{
        (\cfor{c_1}{c_2}, \sigma_1) \Downarrow (\epsilon, \sigma_2)
    }
    \and
    \inferrule[break]{}{
        (\cbreak, \sigma) \Downarrow (\ekb, \sigma)
    }
    \and
    \inferrule[continue]{}{
        (\ccontinue, \sigma) \Downarrow (\ekb, \sigma)
    }
\end{mathpar}
\caption{Big-step semantics of the Toy Language (2)}
\label{fig:complete-big-2}
\end{figure}

\section{Proofs of Some Extended Rules in Big-step semantics}
\label{sec:Abigproof}

We give some proofs of extended rules that are not covered in section~\ref{sec:proof-big}.
All extended rules are proved in our Coq formalization.

\begin{thm}[\textsc{if-seq}]
    Forall $P, Q, \vec{R}, c_1, c_2, c_3$, if
    $$\vDash_b \triple{P}{\cif{e}{c_1\cmdseq c_3}{c_2 \cmdseq c_3}}{Q, [\vec{R}]},$$
    then 
    $$\vDash_b \triple{P}{(\cif{e}{c_1}{c_2}) \cmdseq c_3}{Q, [\vec{R}]}.$$
\end{thm}
\begin{proof}
    Consider the initial state to be $\sigma_1$ with $\sigma_1 \vDash P$.

    If $\mathrm{eval}(e, \sigma_1) = \text{true}$, then we know $c_1 \cmdseq c_3$ will not cause error, and therefore $c_1$ and $c_3$ will not cause error and $(\cif{e}{c_1}{c_2}) \cmdseq c_3$ will not cause error.
    The same happens when $\mathrm{eval}(e, \sigma_1) = \text{false}$, and we know $\cif{e}{c_1\cmdseq c_3}{c_2 \cmdseq c_3}$ will not cause error.

    If $\mathrm{eval}(e, \sigma_1) = \text{true}$ and $c_1$ terminates with normal exit, then we know forall $\sigma_3$ with
    $$
    ((\cif{e}{c_1}{c_2}) \cmdseq c_3, \sigma_1) \Downarrow (\ek, \sigma_3)
    $$
    there exists $\sigma_2$ such that
    $$
        (c_1, \sigma_1) \Downarrow (\epsilon, \sigma_2) \quad\quad
        (c_3, \sigma_2) \Downarrow (\ek, \sigma_3)
    $$
    and we can construct by
    \begin{mathpar}
        \inferrule*[]{
            \inferrule*[]{
                (c_1, \sigma_1) \Downarrow (\epsilon, \sigma_2) \and
                (c_3, \sigma_2) \Downarrow (\ek, \sigma_3)
            }{
                (c_1 \cmdseq c_3, \sigma_1) \Downarrow (\ek, \sigma_3)
            }
            \and \mathrm{eval}(e, \sigma_1) = \text{true}
        }{
            (\cif{e}{c_1\cmdseq c_3}{c_2 \cmdseq c_3}, \sigma_1) \Downarrow (\ek, \sigma_3)
        }
    \end{mathpar}
    and by the triple in the premise, we know $(\ek, \sigma) \vDash \{Q, [\vec{R}]\}$ and we prove the triple in the proof goal.
    If $\mathrm{eval}(e, \sigma_1) = \text{true}$ and $c_1$ terminates with some control flow exit, then $\cif{e}{c_1\cmdseq c_3}{c_2 \cmdseq c_3}$ terminates with the same exit and state, and therefore we know the terminal state satisfies post-conditions.
    The same happens when $\mathrm{eval}(e, \sigma_1) = \text{false}$, and we prove the triple to be valid.
\end{proof}

\begin{thm}[\textsc{nocontinue}]
    Forall $P,Q,R_{\ekb},R_{\ekc},R_{\ekc}',c$, if $c$ does not contain $\ccontinue$ and
    $
    \vDash_b \triple{P}{c}{Q, [R_{\ekb}, R_{\ekc}]},
    $
    then
    $
    \vDash_b \triple{P}{c}{Q, [R_{\ekb}, R_{\ekc}']}.
    $
\end{thm}
\begin{proof}
    Consider the initial state to be $\sigma_1$ with $\sigma_1 \vDash P$.

    The triple in the premise guarantees $c$ will not cause error and we only need to show any $(\ek, \sigma_2)$ that $(c, \sigma_1)$ reduces to satisfies $[R_{\ekb}, R_{\ekc}']$.
    In fact, we can prove by induction over $(c, \sigma_1) \Downarrow (\ek, \sigma_2)$ that $\ek$ is not $\ekc$, i.e., it never exits with continue because $c$ does not contain $\ccontinue$.
    Therefore, $(\ek, \sigma_2)$ satisfies $[R_{\ekb}, \bot]$ and obviously satisfies $[R_{\ekb}, R_{\ekc}']$.
\end{proof}

The proof idea of \textsc{nocontinue} in the big-step based embedding is similar to that in the weakest precondition based embedding, but we do not need to do the simulation because there is no requirement after each step of reduction in the big-step based one, and we only need to show certain properties of the final state it reduces to.
They are very different from the proof in the continuation based one because the interpretation of the post-condition's satisfiability is different.

\section{Semantics Iris-CF}

This section presents key reduction rules for Iris-CF to support control flows.

\setcounter{equation}{0}
\begin{figure}[H]
\begin{align}
    (\cmdloop{e} v, \sigma) &\hred (\cmdloop{e} e, \sigma, \epsilon) && \label{eq:icf-loop-step} \\
    (\cmdloop{e} \cmdcontinue, \sigma) &\hred (\cmdloop{e} e, \sigma, \epsilon) && \label{eq:icf-loop-cont}  \\
    (\cmdloop{e} (\cmdbreak v), \sigma) &\hred (v, \sigma, \epsilon) && \label{eq:icf-loop-break} \\
    (K[\cmdbreak v], \sigma) &\hred (\cmdbreak v, \sigma, \epsilon) && \text{if } K \neq \bullet \text{ and } \cmdbreak \lightning K \label{eq:icf-break-ctx} \\
    (K[\cmdcontinue], \sigma) &\hred (\cmdcontinue, \sigma, \epsilon) && \text{if } K \neq \bullet \text{ and } \cmdcontinue\, \lightning K \label{eq:icf-cont-ctx}
\end{align}
\caption{Head Reductions for Loop}
\label{fig:hred-loop}
\end{figure}
\begin{description}
    \item[\eqref{eq:icf-loop-step}, \eqref{eq:icf-loop-cont}] When the current iteration of a loop evaluate to value or continue terminal, we load the loop body $e$ from the subscription of $\cmdloop{e}$ into the loop context for next iteration.
    \item[\eqref{eq:icf-loop-break}] When the current iteration of a loop evaluate to break terminal, we evaluate the entire loop to value $v$.
    \item[\eqref{eq:icf-break-ctx}, \eqref{eq:icf-cont-ctx}] Break and continue terminals can skip the context around it, if they can ``penetrate'' ($\lightning$) it, e.g., $\cmdcontinue$ penetrates $K \cmdseq \bullet$ but not $\cmdloop{e} \bullet$.
\end{description}

\end{appendices}


\end{document}